\documentclass[letterpaper,11pt]{article}
\usepackage{amsthm,amsmath,amssymb,amsfonts}
\usepackage[american]{babel}
\usepackage{graphicx}
\usepackage[singlespacing]{setspace}
\usepackage{hyperref}
\usepackage{cleveref}
\usepackage[url=false,giveninits=true,maxbibnames=99,style=alphabetic,maxalphanames=5]{biblatex}
\usepackage{mathtools}
\usepackage{enumitem}
\usepackage{csquotes}
\usepackage[noend]{algpseudocode}
\usepackage{algorithm}
\usepackage{xparse}
\usepackage{caption}
\usepackage{subcaption}
\usepackage{fullpage}
\usepackage{placeins}
\usepackage[usenames,dvipsnames]{xcolor}
\usepackage{makecell}
\usepackage{thmtools}
\usepackage{thm-restate}
\usepackage{xspace}
\usepackage{footmisc}

\usepackage[paper=letterpaper, margin=1in]{geometry}

\mathchardef\mhyphen="2D 

\newcommand\newmathabbrev[2]{\newcommand{#1}{\ensuremath{#2}\xspace}}

\newcommand\cfont\mathrm
\newmathabbrev\p{\cfont{P}}
\newmathabbrev{\N}{\mathbb N}
\newmathabbrev\NP{\cfont{NP}}
\newmathabbrev\DTIME{\cfont{DTIME}}
\newmathabbrev\tSAT{3\cfont{\mhyphen{}SAT}}
\newmathabbrev\MA{\cfont{MA}}
\newmathabbrev\AM{\cfont{AM}}
\newmathabbrev\NPDAG{\cfont{NP\mhyphen{}DAG}}
\newmathabbrev\QMADAG{\cfont{QMA\mhyphen{}DAG}}
\newmathabbrev\yes{\mathrm{yes}}
\newmathabbrev\no{\mathrm{no}}
\newmathabbrev\US{\cfont{US}}
\newmathabbrev\FP{\cfont{FP}}
\newmathabbrev\PP{\cfont{PP}}
\newmathabbrev\CeP{\cfont{C_=P}}
\newmathabbrev\coCeP{\cfont{coC_=P}}
\newmathabbrev\PH{\cfont{PH}}
\newmathabbrev\SAT{\cfont{SAT}}
\newmathabbrev\SPP{\cfont{SPP}}
\newmathabbrev\GapP{\cfont{GapP}}
\newmathabbrev\BQP{\cfont{BQP}}
\newmathabbrev\QP{\cfont{QP}}
\newmathabbrev\StoqMA{\cfont{StoqMA}}
\newmathabbrev\coNP{\cfont{coNP}}
\newmathabbrev\AzPP{\cfont{A_0PP}}
\newmathabbrev\QMA{\cfont{QMA}}
\newmathabbrev\BellQMA{\cfont{BellQMA}}
\newmathabbrev\QMAt{\QMA(2)}
\newmathabbrev{\QMAte}{\QMAt_{\exp}}
\newmathabbrev{\PreciseQMAt}{\PreciseQMA(2)}
\newmathabbrev{\pQMAlogt}{\cfont{PQMA}_{\log}(2)}
\newmathabbrev\coQMA{\cfont{coQMA}}
\newmathabbrev\BPP{\cfont{BPP}}
\newmathabbrev\QCMA{\cfont{QCMA}}
\newmathabbrev\pNPlog{\p^{\NP[\log]}}
\newmathabbrev\pNP{\p^{\NP}}
\newmathabbrev\pNPtwo{\p^{\NP[2]}}
\newmathabbrev\pNPone{\p^{\NP[1]}}
\newmathabbrev\pParSAT{\p^{||\SAT}}
\newmathabbrev\pQMApar{\p^{||\QMA}}
\newmathabbrev\pCpar{\p^{||\C}}
\newmathabbrev\pStoqMApar{\p^{||\StoqMA}}
\newmathabbrev\pQMAlog{\p^{\QMA[\log]}}
\newmathabbrev\pClog{\p^{\textup{C}[\log]}}
\newmathabbrev\pC{\p^{\textup{C}}}
\newmathabbrev\QMASPACE{\cfont{QMASPACE}}
\newmathabbrev\SQMASPACE{\cfont{SQMASPACE}}
\newmathabbrev\SQCMASPACE{\cfont{SQCMASPACE}}
\newmathabbrev\QUMASPACE{\cfont{Q_UMASPACE}}
\newmathabbrev\SQMA{\cfont{SQMA}}
\newmathabbrev\SQCMA{\cfont{SQCMA}}
\newmathabbrev\PreciseQMA{\cfont{PreciseQMA}}
\newmathabbrev\pQMAtlog{\p^{\QMA(2)[\log]}}
\newmathabbrev\pStoqMAlog{\p^{\StoqMA[\log]}}
\newmathabbrev\pQMApt{\p^{\Vert\QMA(2)}}
\newmathabbrev\pQMA{\p^{\QMA}}
\newmathabbrev\SharpP{\cfont{\#P}}
\newmathabbrev\pSharP{\p^{\SharpP[1]}}
\newmathabbrev\PromisePP{\cfont{PromisePP}}
\newmathabbrev\lett{\le_\mathrm{tt}}
\newmathabbrev\YES{\mathsf{YES}}
\newmathabbrev\NO{\mathsf{NO}}
\newmathabbrev\PSPACE{\cfont{PSPACE}}
\newmathabbrev\IP{\cfont{IP}}
\newmathabbrev\POLY{\cfont{POLY}}
\newmathabbrev\DAG{\cfont{DAG}}
\newmathabbrev\StoqMADAG{\StoqMA\mhyphen\cfont{DAG}}
\newmathabbrev\CDAG{C\mhyphen\cfont{DAG}}
\newmathabbrev\CDAGf{C\mhyphen\cfont{DAG}_f}
\newmathabbrev\CDAGs{C\mhyphen\cfont{DAG}_s}
\newmathabbrev\CDAGd{C\mhyphen\cfont{DAG}_{d}}
\newmathabbrev\CDAGo{C\mhyphen\cfont{DAG}_1}
\newmathabbrev\LOGS{\cfont{LOGS}}
\newmathabbrev\TAUT{\cfont{TAUTOLOGY}}
\newmathabbrev\SBQP{\cfont{SBQP}}
\newmathabbrev\Fc{F_\coNP}
\newmathabbrev\Fa{F_\AzPP}
\newmathabbrev\GSCON{\cfont{GSCON}}
\newmathabbrev\QSPACE{\cfont{QSPACE}}
\newmathabbrev\QCVERIFY{\cfont{QCVERIFY}}
\newmathabbrev\QSPACElog{\QSPACE_\cfont{log}}
\newmathabbrev\QCMASPACElog{\QCMASPACE_\cfont{log}}
\newmathabbrev\GSCONexp{\GSCON_\cfont{exp}}
\newmathabbrev\SSGSCON{\cfont{SEPARABLE~SPARSE~GSCON}}
\newmathabbrev\SGSCONexp{\cfont{SPARSE~GSCON}_{\exp}}
\newmathabbrev\SSGSCONexp{\cfont{SEPARABLE~SPARSE~GSCON}_{\exp}}
\newmathabbrev\SSH{\cfont{SSH}}
\newmathabbrev\QMAexp{\QMA_\cfont{exp}}
\newmathabbrev\QMAEXP{\QMA_\cfont{EXP}}
\newmathabbrev\UQMA{\cfont{UQMA}}
\newmathabbrev\R{\mathbb R}
\newmathabbrev\Trees{\cfont{TREES}}
\newmathabbrev\apxsim{\cfont{APX\mhyphen{}SIM}}
\newmathabbrev\AWPP{\cfont{AWPP}}
\newmathabbrev\X{\mathcal{X}}
\newmathabbrev\Y{\mathcal{Y}}

\newmathabbrev\Z{\mathcal{Z}}
\newmathabbrev\ZZ{\mathbb{Z}}
\newmathabbrev\Hprop{H_\mathrm{prop}}
\newmathabbrev\Hin{H_\mathrm{in}}
\newmathabbrev\Hout{H_\mathrm{out}}
\newmathabbrev\Hstab{H_\mathrm{stab}}
\newmathabbrev\Lext{\L_\mathrm{ext}}
\newmathabbrev\BTWNP{\cfont{BTW}(\NP)}
\newmathabbrev\BSN{\cfont{BSN}}
\newmathabbrev\SN{\cfont{SN}}
\newmathabbrev\BD{\cfont{BD}}
\newmathabbrev\HYPERTREE{\cfont{NP\mhyphen{}HYPERTREE}}
\newmathabbrev\Hext{H_\mathrm{ext}}
\newmathabbrev\EXP{\cfont{EXP}}
\newmathabbrev\A{\mathcal{A}}
\newmathabbrev\U{\mathcal{U}}

\newcommand\kLH{k\cfont{\mhyphen{}LH}}
\renewcommand\L{\mathcal{L}}

\newcommand\B{\mathcal B}
\newcommand\Piacc{\Pi_\mathrm{acc}}
\newcommand\Pacc{P_\mathrm{acc}}

\newmathabbrev\DAGSSAT{\DAGS(\SAT)}
\newmathabbrev\DAGS{\mathrm{DAGS}}
\newmathabbrev\DAGSNP{\DAGS(\NP)}

\newmathabbrev\AND{\cfont{AND}}

\newmathabbrev\STCONN{{S,T}\cfont{\mhyphen{}CONN}}
\newmathabbrev\CNF{\cfont{CNF}}
\newmathabbrev\NEXP{\cfont{NEXP}}
\newmathabbrev\NPSPACE{\cfont{NPSPACE}}
\newmathabbrev\QCMASPACE{\cfont{QCMASPACE}}
\newmathabbrev\BQPSPACE{\cfont{BQPSPACE}}
\newmathabbrev\BQUSPACE{\cfont{BQ_USPACE}}
\newmathabbrev\BQSPACE{\cfont{BQSPACE}}
\newmathabbrev\SPACE{\cfont{SPACE}}
\newmathabbrev{\PCP}{\cfont{PCP}}
\newmathabbrev\BQUPSPACE{\cfont{BQ_UPSPACE}}
\newmathabbrev{\MIPto}{\cfont{MIP}(2,1)}
\newmathabbrev{\MIP}{\cfont{MIP}}
\newmathabbrev{\PQMAlogt}{\cfont{PQMA}_{\log}(2)}
\newmathabbrev{\CNOT}{\mathrm{CNOT}}
\newmathabbrev{\IDTIH}{1\cfont{D\mhyphen TIH}}

\newcommand{\SWAP}{\mathrm{SWAP}}

\newcommand{\CUp}{C\mhyphen U^{+}}
\newmathabbrev{\IDLH}{\cfont{1D}\mhyphen\cfont{LH}[k]_{\Lambda,d,\alpha,\beta}}
\newmathabbrev{\SQMAp}{\SQMA[k]_{c,s}}
\newmathabbrev{\Cacc}{C_{\mathrm{acc}}}
\newmathabbrev{\Cinit}{C_{\mathrm{init}}}

\protected\def\verythinspace{%
  \ifmmode
    \mskip0.5\thinmuskip
  \else
    \ifhmode
      \kern0.08334em
    \fi
  \fi
}

\newcommand{\bin}{\{0,1\}}

\newcommand{\C}{\mathbb C}

\newcommand{\be}{\begin{equation}}
\newcommand{\ee}{\end{equation}}

\newcommand{\Htt}{\widetilde{H}_t}
\newcommand{\HtT}{\widetilde{H}_T}

\newcommand{\HttI}{{H_t^I}}
\newcommand{\HttX}{{H_t^X}}
\newcommand{\HttiX}{H_{t}^{iX}}

\newcommand{\HttsI}{H_{t^*}^I}
\newcommand{\HttsiX}{H_{t^*}^X}

\newcommand{\HttU}{{H_t^U}}
\newcommand{\HttV}{{H_t^V}}
\newcommand{\Hsym}{\widetilde{H}_{\textup{sym}}}
\newcommand{\Dsym}{\Delta_{\textup{sym}}}
\newcommand{\Dprop}{\Delta_{\textup{prop}}}
\newcommand{\Din}{\Delta_{\textup{in}}}
\newcommand{\Psym}{P^{\textup{sym}}}
\newcommand{\Ht}{\widetilde{H}}
\newcommand{\Hint}{\widetilde{H}_{\rm in}}
\newcommand{\Hpropt}{\widetilde{H}_{\rm prop}}
\newcommand{\Hpropl}{{H}_{\rm prop}}
\newcommand{\Houtt}{\widetilde{H}_{\rm out}}
\newcommand{\kpot}{\ket{\psi_1}\ket{\psi_2}}
\newcommand{\bpot}{\bra{\psi_1}\bra{\psi_2}}
\newcommand{\Todd}{T_{\textup{comp}}}
\newcommand{\Teven}{T_{\textup{proof}}}
\newcommand{\gadget}{FLUX}

\renewcommand{\epsilon}{\varepsilon}

\DeclareMathOperator{\Herm}{Herm}

\DeclareMathOperator{\Tr}{Tr}

\DeclareMathOperator\Span{Span}

\newcommand\hermp[1]{\Herm\left(#1\right)}
\newcommand\lmin{\lambda_{\mathrm{min}}}
\newcommand\lmax{\lambda_{\mathrm{max}}}

\newcommand{\set}[1]{{\left\{#1\right\}}}    

\newcommand\Wlog{W.l.o.g.}

\DeclareMathOperator{\herm}{Herm}

\DeclareMathOperator{\poly}{poly}

\DeclareMathOperator{\real}{Re}

\DeclarePairedDelimiter\bra{\langle}{\rvert}
\DeclarePairedDelimiter\ket{\lvert}{\rangle}

\DeclarePairedDelimiter\abs{\lvert}{\rvert}
\DeclarePairedDelimiter\norm{\lVert}{\rVert}
\DeclarePairedDelimiter\enorm{\lVert}{\rVert_2}
\DeclarePairedDelimiter\fnorm{\lVert}{\rVert_{\mathrm F}}
\DeclarePairedDelimiter\snorm{\lVert}{\rVert_\infty}
\DeclarePairedDelimiter\trnorm{\lVert}{\rVert_{\mathrm{tr}}}
\DeclarePairedDelimiterX\braket[2]{\langle}{\rangle}{#1 \delimsize\vert #2}
\DeclarePairedDelimiterX\ketbra[2]{\lvert}{\rvert}{#1 \delimsize\rangle\delimsize\langle #2}

\newcommand{\braketa}[2]{{#1 #2 \renewcommand\bra\ket #1}}
\newcommand{\braketb}[2]{\bra{#1}#2\ket{#1}}
\newcommand\kb[1]{\ketbra{#1}{#1}}

\declaretheorem[numberwithin=section]{theorem}
\declaretheorem[sibling=theorem]{corollary}
\declaretheorem[sibling=theorem]{lemma}

\declaretheorem[sibling=theorem]{question}
\declaretheorem[sibling=theorem,style=definition]{definition}
\declaretheorem[sibling=theorem,style=remark]{remark}

\newcommand{\hin}{H_{\textup{in}}}
\newcommand{\hprop}{H_{\textup{prop}}}
\newcommand{\hstab}{H_{\textup{stab}}}

\newcommand{\hout}{H_{\textup{out}}}

\newcommand{\aaa}{\eta_1}
\newcommand{\bbb}{\eta_2}
\newcommand{\ccc}{\eta_3}
\newcommand{\ddd}{\eta_4}

\newcommand{\ayes}{A_{\textup{yes}}} 
\newcommand{\ano}{A_{\textup{no}}} 

\newcommand{\psihist}{\psi_{\textup{hist}}}

\newcommand{\unitary}[1]{\textup{U}\left(#1\right)}

\newcommand{\LR}{\ket{\psi_1}_L\ket{\psi_2}_R}
\newcommand{\RL}{\bra{\psi_1}_L\bra{\psi_2}_R}
\newcommand{\LRt}{\ket{\psi_1}_L\otimes\ket{\psi_2}_R}
\newcommand{\RLt}{\bra{\psi_1}_L\otimes\bra{\psi_2}_R}

\newcommand{\pot}{\ket{\psi_1}\ket{\psi_2}}
\newcommand{\poo}{\ket{\psi_1}\ket{\psi_1}}
\newcommand{\pto}{\bra{\psi_1}\bra{\psi_2}}
\newcommand{\ptt}{\bra{\psi_1}\bra{\psi_1}}

\makeatletter
\newcommand{\subalign}[1]{%
  \vcenter{%
    \Let@ \restore@math@cr \default@tag
    \baselineskip\fontdimen10 \scriptfont\tw@
    \advance\baselineskip\fontdimen12 \scriptfont\tw@
    \lineskip\thr@@\fontdimen8 \scriptfont\thr@@
    \lineskiplimit\lineskip
    \ialign{\hfil$\m@th\scriptstyle##$&$\m@th\scriptstyle{}##$\hfil\crcr
      #1\crcr
    }%
  }%
}
\NewDocumentCommand{\LeftComment}{s m}{%
  \Statex \IfBooleanF{#1}{\hspace*{\ALG@thistlm}}\(\triangleright\) #2}

\def\moverlay{\mathpalette\mov@rlay}
\def\mov@rlay#1#2{\leavevmode\vtop{%
   \baselineskip\z@skip \lineskiplimit-\maxdimen
   \ialign{\hfil$\m@th#1##$\hfil\cr#2\crcr}}}
\newcommand{\charfusion}[3][\mathord]{
    #1{\ifx#1\mathop\vphantom{#2}\fi
        \mathpalette\mov@rlay{#2\cr#3}
      }
    \ifx#1\mathop\expandafter\displaylimits\fi}

\makeatother

\algnewcommand{\LineComment}[1]{\State \(\triangleright\) #1}

\algblockdefx[ON]{Blk}{EndBlk}[1]
  {#1}
  {}

\makeatletter
\ifthenelse{\equal{\ALG@noend}{t}}%
  {\algtext*{EndBlk}}
  {}%
\makeatother

\AtEveryBibitem{%
  \clearlist{language}%
}

\usepackage{quantikz}

\begin{document}

\title{Quantum space, ground space traversal, and how to embed multi-prover interactive proofs into unentanglement}

\author{
	Sevag Gharibian\footnote{Department of Computer Science and Institute for Photonic Quantum Systems (PhoQS), Paderborn University, Germany. Email: \{\href{mailto:sevag.gharibian@upb.de}{sevag.gharibian}, \href{mailto:dorian.rudolph@upb.de}{dorian.rudolph}\}@upb.de.}
	\and Dorian Rudolph\footnotemark[1]
}

\maketitle

\begin{abstract}
    A celebrated result in classical complexity theory is Savitch's theorem, which states that non-deterministic polynomial-space computations (\NPSPACE) can be simulated by deterministic poly-space computations (PSPACE). In this work, we initiate the study of a quantum analogue of \NPSPACE, denoted Streaming-\QCMASPACE (\SQCMASPACE), in which an exponentially long classical proof is streamed to a poly-space quantum verifier.
    We first show that a quantum analogue of Savitch's theorem is unlikely to hold, in that $\SQCMASPACE=\NEXP$.
    For completeness, we also introduce the companion class Streaming-\QMASPACE (\SQMASPACE) with an exponentially long streamed \emph{quantum} proof, and show $\SQMASPACE=\QMAEXP$ (the quantum analogue of $\NEXP$).
    Our primary focus, however, is on the study of exponentially long streaming \emph{classical} proofs, where we next show the following two main results.

    The first result shows that, in strong contrast to the classical setting, the solution space of a quantum constraint satisfaction problem (i.e. a local Hamiltonian) is \emph{always} connected when exponentially long proofs are permitted. For this, we show how to simulate any Lipschitz continuous path on the unit hypersphere via a sequence of \emph{local} unitary gates, at the expense of blowing up the circuit size. This shows that quantum error-correcting codes can be unable to detect one codeword erroneously evolving to another {if} the evolution happens sufficiently slowly, and additionally answers an open question of [Gharibian, Sikora, ICALP 2015] regarding the Ground State Connectivity problem.

    Our second main result is that any \SQCMASPACE\ computation can be embedded into ``unentanglement'', i.e. into a quantum constraint satisfaction problem with unentangled provers. Formally, we show how to embed \SQCMASPACE\ into the Sparse Separable Hamiltonian problem of [Chailloux, Sattath, CCC 2012] (known to be QMA(2)-complete for $1/\poly$ promise gap), at the expense of scaling the promise gap with the streamed proof size. As a corollary, we obtain the first systematic construction for obtaining $\QMAt$-type upper bounds on {arbitrary} multi-prover interactive proof systems, where the $\QMAt$ promise gap scales exponentially with the number of bits of communication in the interactive proof. 
    At the heart of our construction is a new technique for exploiting unentanglement to simulate quadratic Boolean functions, which in some sense allows \emph{history} states to encode the \emph{future}.
\end{abstract}
\newpage
\section{Introduction}\label{scn:intro}

Computational complexity theory studies the resources required to solve a given computational problem. The resources of time and space, in particular, are very well-studied, revealing certain interesting discrepancies. For example, while the question of whether {non-deterministic} poly-{time} (NP) equals {deterministic} poly-{time} (P) remains a central open problem in the field, in the context of \emph{space}, the answer is well-known: In 1970, Savitch~\cite{Savitch1970} gave his celebrated result that {non-deterministic} poly-space computations (\NPSPACE) could be simulated by {deterministic} poly-space computations (PSPACE), yielding $\PSPACE=\NPSPACE$.

Motivated by the prospect of a quantum analogue of Savitch's theorem, in this work, we initiate the study of a ``non-deterministic'' quantum analogue of PSPACE, which we call \SQCMASPACE. To define the latter, recall that \NPSPACE\ may be viewed as a PSPACE machine which receives an \emph{exponential} length proof $y\in\set{0,1}^{2^n}$. Of course, a PSPACE verifier cannot even write down $y$ given its limited memory, so a natural way to formalize this idea is to allow $y$ to be \emph{streamed}, bit by bit. This is the approach we take\footnote{One can in principle consider alternative definitions of \SQCMASPACE. For example, \Cref{def:QCMASPACEinformal} allows only one streaming pass of the proof, but one could consider multiple passes. An even stronger access model might allow the ability to query arbitrary single bits of the proof. For our results here, however, a single-pass streaming model suffices, e.g. this definition already captures NEXP, as we show in \Cref{thm:qcmaspace}.} in defining \SQCMASPACE.

\begin{definition}[\SQCMASPACE (informal; see \Cref{def:QCMASPACE})]\label{def:QCMASPACEinformal}
    A promise problem $A=(\ayes,\ano)$ is in $\SQCMASPACE(p,q,r)$ for polynomially-bounded functions $p,q,r$, if there exist thresholds $\alpha(n),\beta(n)$ satisfying $\alpha(n)-\beta(n)\geq 2^{-r(n)}$, and a polynomial-space uniform family of quantum circuits $\set{Q_n}$ such that, for any input $x\in\Sigma^n$:
    \begin{itemize}
    \item If $x\in\ayes$, there exists a streaming proof $y\in\set{0,1}^{2^{p(n)}}$ such that $Q_n$ accepts $(x,y)$ with probability at least $\alpha$.
    \item If $x\in\ano$, then for all streaming proofs ${y}\in \set{0,1}^{2^{p(n)}}$, $Q_n$ accepts $(x,{y})$ with probability at most $\beta$.
    \end{itemize}
\end{definition}
\noindent To avoid cluttering the introduction, we leave our formal definition of \emph{streaming proof} to \Cref{scn:defs} (\Cref{def:stream} therein), and instead make do with the following intuitive definition: To ``stream'' the next bit $y_i$ to the verifier, we imagine the prover applies either ``proof gate'' $I$ (if $y_i=0$) or $X$ (if $y_i=1$), for $I$ and $X$ the single-qubit identity and Pauli $X$ (i.e. NOT) gates, respectively, to a designated qubit $k$ in the verifier's memory, which is initialized to $\ket{0}_k$. The verifier then copies\footnote{The verifier can also simulate the choice \emph{not} to copy the bit into memory, if desired. See the discussion after \Cref{def:stream}.} this bit into its main memory via Controlled-NOT (CNOT), and the prover subsequently uncomputes bit $y_i$ by re-applying $I$ or $X$ to $k$, respectively. In other words, there is no separate proof---we view the entire computation as a sequence of gates on the verifier's memory, some gates of which (the ``proof'' gates) are \emph{a priori} unknown. For clarity, this is similar to how communication is modelled in quantum interactive proofs, where prover and verifier take turns acting on a shared ``message register'' (see e.g. Kitaev and Watrous \cite{Kitaev2000}).

In \Cref{sscn:related-work}, we survey previous works studying quantum notions of PSPACE. Most relevant to our discussion at this point, however, is the work of Fefferman and Remscrim~\cite{FR21}, which defines a quantum variant of \NPSPACE denoted $\QMASPACE$, and which differs from \SQCMASPACE in three respects: The first two differences are that $\QMASPACE$ has a \emph{poly}-length proof which is \emph{quantum}, whereas $\SQCMASPACE$ has an \emph{exponential} length streamed proof, which is \emph{classical}. The third difference is that whereas $\QMASPACE=\PSPACE$~\cite{FR21},  here we show $\SQCMASPACE=\NEXP$ (\Cref{thm:qcmaspace}, stated shortly). To the best of our knowledge, the current work is the first to formalize and study a quantum analogue of \NPSPACE which allows an \emph{exponentially} long classical proof.

\paragraph{Broader theme.} Beyond initiating the study of \SQCMASPACE\ itself, the broader theme of this work asks: \emph{``What can one say about exponentially long proofs verified by poly-space quantum verifiers?''} For example, can allowing an exp-length proof ``trivialize'' a problem which is provably hard for poly-length proofs? Can exponential length proofs be encoded into \emph{poly}-size history state\footnote{A history state \cite{KSV02} is the quantum analogue of a tableau in the Cook-Levin theorem \cite{Cook1971,Levin1973}).} constructions? Here, we give positive answers to both of these questions, for which we now set up the background.

\paragraph{Question 1: Exp- versus poly-length proofs, and the solution space of constraint satisfaction problems (CSPs).} In 2006, Gopalan, Kolaitis, Maneva and Papadimitriou~\cite{Gopalan2006} initiated the study of \emph{reconfiguration problems} for SAT, which ask: Given two solutions $x$ and $y$ to a SAT formula $\phi$, is there a path in the hypercube from $x$ to $y$ on which all intermediate vertices $z$ are also solutions? Alternatively, in graph theoretic terms, is the solution space of $\phi$ \emph{connected}? Reference \cite{Gopalan2006} showed that this decision problem is PSPACE-complete, which in particular implies the problem is not \emph{trivial} --- the solution space can be either connected or disconnected, and deciding between the two is hard.

In the quantum setting, one can ask analogous questions about the ``solution space'' of {quantum} CSPs, and this has implications for the study of quantum error-correcting codes. To begin, the quantum generalization of a MAX-$k$-SAT instance $\phi$ is a \emph{$k$-local Hamiltonian} $H$. $H$ is a $2^n\times 2^n$ Hermitian operator acting on $n$ qubits, but specified \emph{succinctly} via a sum of ``local clauses'' $H_i$ acting on $k$ qubits (analogous to how $\phi$ is specified locally via an AND of $k$-local disjunctions), i.e. $H=\sum_i H_i$. The smallest eigenvalue of $H$, $\lmin(H)$, is the \emph{ground state energy} of $H$ (for $\phi$, this encodes the maximum number of simultaneously satisfiable clauses), and the corresponding space of eigenvectors the \emph{ground space} (for $\phi$, this encodes the space of optimal assignments). In 2002, Kitaev~\cite{KSV02} gave his now celebrated ``quantum Cook-Levin theorem'', which showed that estimating the ground state energy of $H$, known as the $k$-local Hamiltonian problem ($\kLH$), is complete\footnote{QMA is a quantum analogue of Merlin-Arthur (MA), except with a quantum proof and quantum verifier.} for Quantum-Merlin Arthur (QMA).

With the definition of local Hamiltonians (i.e. quantum CSPs) in hand, we can now state the quantum analogue of reconfiguration, defined as follows.
\begin{definition}[Ground State Connectivity (GSCON)~\cite{Gharibian2018} (informal); see \Cref{def:gscon}]
    Given a $k$-local Hamiltonian $H$ with ground states $\ket\psi$ and $\ket\phi$ (represented succinctly via quantum circuits), and parameters $m,l$, does there exist a sequence of $l$-local unitaries $U_1,\dots,U_m$ such that:
    \begin{enumerate}
        \item ($\ket\psi$ mapped to $\ket\phi$) $U_m\dotsm U_1\ket\phi \approx \ket\psi$, and
        \item (intermediate states have low energy) $\forall i\in[m], U_i\dotsm U_1 \ket\psi$ has low energy relative to $H$?
    \end{enumerate}
\end{definition}

\noindent
In words, \GSCON\ asks whether there exists a sequence of $m$ $l$-local unitaries that map $\ket\psi$ to $\ket\phi$ such that intermediate states have low energy (i.e. are also approximate ``solutions'') with respect to $H$. Here, the use of \emph{local} unitaries $U_i$ is crucial, and generalizes the notion of following a path on the hypercube for SAT (which would involve flipping one bit of an assignment per step, or in quantum terms, applying a local $X$ gate). Thus, \GSCON asks: Is the ground space of $H$ ``connected''?

Recall now that in the classical setting, the solution space of a SAT formula can be either connected or disconnected~\cite{Gopalan2006}. In this work, we ask the analogous fundamental question about the structure of ground spaces of local Hamiltonians:
\begin{question}\label{q1}
    {Can ground spaces of local Hamiltonians be either ``connected'' or ``disconnected''?}
\end{question}
\noindent It is known that if only \emph{poly}-length sequences of local gates are allowed, the answer to this question is YES --- namely, \GSCON with a polynomial sequence of $2$-local unitaries ($m=\poly(n), l=2$) and inverse polynomial spectral gap is\footnote{Quantum-Classical Merlin-Arthur (QCMA) is a quantum analogue of MA, except with a classical proof and quantum verifier.} \QCMA-complete~\cite{Gharibian2018}. However, even in the classical case, in the worst case a connecting path in the hypercube might be \emph{exponentially} long! (Indeed, this is what makes the PSPACE-completeness result of \cite{Gopalan2006} possible.) Thus, to answer \Cref{q1}, we must allow sequences of \emph{exponentially} many local gates, i.e. \GSCON\ with $m=\exp(n)$, denoted $\GSCONexp$.

In addition to this fundamental structural motivation, there are two additional reasons why \Cref{q1} is interesting:
\begin{itemize}
     \item First, from a complexity theory perspective, an instance of $\GSCONexp$ is straightforwardly in \SQCMASPACE---roughly, in step $i$, the prover streams gate $U_i$ to the verifier, who applies it to map its current state from $U_{i-1}\cdots U_1\ket{\psi}$ to $U_i\cdots U_1\ket{\psi}$. Once the proof is fully received, the verifier randomly chooses to check one of the two conditions in the \GSCON definition, and accepts if the condition is met. Thus, \emph{if} a ``quantum Savitch's'' theorem were to hold, i.e. $\SQCMASPACE=\PSPACE$, then we would immediately obtain $\GSCONexp\in \SQCMASPACE=\PSPACE$, resolving an open question of \cite{Gharibian2018}.
     \item Second, and perhaps most interesting, is the connection to quantum error-correcting codes. For example, in a stabilizer code \cite{G97}, the set of valid codewords is the ground space of a local Hamiltonian $H$. In this case, one desires the ground space of $H$ to be ``disconnected'' in the following sense. Let $\ket{\psi}$ be a codeword of $H$. Then, any sufficiently short sequence of local gates (think of these as local errors ``corrupting'' $\ket{\psi}$) should ideally take one \emph{out} of the ground space, so that measuring the Hamiltonian catches the corrupting process with non-negligible probability. Indeed, this is precisely what quantum codes typically achieve. What is much less obvious, however, is what happens with \emph{exponential length} corrupting processes --- by allowing an exponential-length sequences of local gates, can we stealthily map from $\ket{\psi}$ to some other codeword $\ket{\phi}$ while remaining \emph{exponentially close} to the ground space? If so, then a single measurement of the Hamiltonian during this corrupting process is highly unlikely to detect that we are no longer in state $\ket{\psi}$!

    \end{itemize}

\paragraph{Question 2: Exp-length proofs, poly-size history states, and $\QMAt$.} Our next question asks: \emph{Can exponential length proofs be encoded into \emph{poly}-size history state/circuit-to-Hamiltonian constructions?} Here, a circuit-to-Hamiltonian construction is the quantum analogue of the Cook-Levin construction~\cite{Cook1971,Levin1973}, i.e. a map from quantum circuits $V$ to local Hamiltonians $H$, such that the ground space of $H$ encodes the action of $V$. The basic premise is captured by Kitaev's $5$-local construction, which maps a QMA verification circuit $V=V_m\cdots V_1$ (for $1$- and $2$-qubit gates $V_i$) to a local Hamiltonian $H=\Hin+\Hprop+\Hout+\Hstab$. Intuitively, each of $\Hin$, $\Hprop$, and $\Hout$ plays a role analogous to its classical cousin in the Cook-Levin construction--- $\hin$ ensures $V$'s computation is initialized correctly, $\Hprop$ that in time step $t$ the gate $V_t$ is applied, and $\Hout$ that rejecting computations are penalized. Then, the ``ideal'' quantum assignment perfectly satisfying $\hin$ and $\hprop$ is the \emph{history state}
\begin{equation}\label{eqn:hist}
    \ket{\psihist}=\frac{1}{\sqrt{m+1}}\sum_{t=0}^mV_t\cdots V_1\ket{\psi_{\textup{proof}}}_A\ket{0\cdots 0}_B\ket{t}_C
\end{equation}
(the quantum analogue of a ``tableau''), where in the context of QMA, register $A$ starts with the quantum proof $\ket{\psi_{\textup{proof}}}$, $B$ is the ancilla space, and $C$ is the clock keeping track of time.

Returning to the question at hand, the naive approach to encoding an exponentially long proof (given explicitly) into history state $\ket{\psihist}$ would result in an \emph{exponential} size proof register $A$, which is too large for our purposes. However, in our definition of \SQCMASPACE, the proof is not given explicitly, but \emph{streamed} via application of local gates. While this may seem \emph{a priori} more difficult to work with, it has a distinct benefit --- since all gates $V_t$ encoding streamed proof bits (i.e. ``proof gates'') are ``part of'' the verification circuit itself, we can directly encode them into the history state's \emph{superposition/sum} over time steps (requiring only poly-space), thus obviating the need for a separate proof register, $A$! Of course, now we are out of the frying pan into the fire, for there remains a serious problem --- the propagation term $\Hprop=\sum_{t=1}^m H_t$, which explicitly encodes each gate $V_t$ into its corresponding local propagation term, $H_t$, needs to be fully specified in \emph{advance}. However, by definition of streaming proof, the gates $V_t$ which are {proof gates} are \emph{not} known in advance. Can correct propagation still somehow be enforced? To put it more ``dramatically'', can a \emph{history} state be used to encode the \emph{future}?

In and of itself, this seems paradoxical. Yet, there \emph{is} a setting in which special cases of classical proofs can be ``compressed'' into an exponentially smaller number of qubits---$\QMAt$ (\Cref{def:QMAt}). Informally, $\QMAt$ is defined as \QMA, except where the verifier is promised to get a proof in \emph{tensor product} across some prespecified partition $L$ versus $R$ of the qubits, i.e. an ``unentangled'' proof of form $\ket{\psi_1}_L\otimes \ket{\psi_2}_R$. In this setting, Blier and Tapp ~\cite{BT12} first showed that the NP-complete problem $3$-SAT could be verified using just \emph{log-size} ``unentangled'' proofs, log-space quantum verification, and $1/\poly$ promise gap, i.e. in $\PQMAlogt$. Next, Pereszlényi~\cite{Per12} showed a similar result for verifying the \NEXP-complete language SUCCINCT-3-COLORING via poly-size unentangled proofs and $1/\exp$ promise gap, i.e. in $\PreciseQMAt$ (thus obtaining $\PreciseQMAt=\NEXP$). (Further related works in \Cref{sscn:related-work}.) However, these constructions are expressly tailored to the problems being reduced from, and \emph{a priori} have nothing to do with streaming. Moreover, to-date, no systematic constructions were known for embedding ``long'' classical proofs into ``small'' unentangled quantum systems. We thus ask:
\begin{question}\label{q2}
    {Can unentanglement be exploited to compress streaming proofs into exponentially smaller\footnote{For clarity, ``smaller'' refers to the number of qubits in the history state. Thus, if the proof has length $f(n)$, then the history state should be a $O(\log(f(n)))$-qubit state.} history state constructions?}
\end{question}

\subsection{Our Results}\label{sscn:results}

We divide our results into three parts: \SQCMASPACE, ground space traversal, and embedding streaming proofs into unentanglement.

\paragraph{1. The complexity of \SQCMASPACE.} 
We first show that a quantum analogue of Savitch's theorem for \SQCMASPACE is highly unlikely to hold, even in the setting of \emph{constant} promise gap.
\begin{restatable}{theorem}{theoremqcmaspace}\label{thm:qcmaspace}
        $\SQCMASPACE$ with $2^{\poly(n)}$ proof bits, $\poly(n)$ ancilla qubits, completeness $1$, and soundness $1/2$, equals NEXP, i.e. $\SQCMASPACE(\poly,\poly,1)=\NEXP$.
\end{restatable}
\noindent
For completeness, we also define the analogous class $\SQMASPACE$ (\Cref{def:SQMASPACE}), which takes an exponential length streamed \emph{quantum} proof, and show its equality to \QMAEXP (quantum analogue of \NEXP):
\begin{theorem}\label{thm:qmaspace}
    $\SQMASPACE$ with $2^{\poly(n)}$ proof bits, $\poly(n)$ ancilla qubits, completeness $2/3$, and soundness $1/3$, equals \QMAEXP, i.e. $\SQMASPACE(\poly,\poly,1)=\QMAEXP$.
    With $\poly(n)$ proof bits, $O(\log(n))$, ancilla bits, it equals \QMA, i.e. $\SQMASPACE(\log,\log,0) = \QMA$.
\end{theorem}

\paragraph{2. Ground space traversal.} Our second result reveals that \Cref{q1} has an arguably surprising resolution --- in strong contrast to the classical case, in which the solution space for a SAT instance can be connected or disconnected, in the quantum setting, ground spaces of local Hamiltonians are \emph{always} connected.

At the heart of this result is a new technical lemma showing how to simulate any Lipschitz continuous path on the hypersphere by an exponentially long sequence of \emph{local} quantum gates (i.e. gates on a typical gate-based quantum computer). For this, define a \emph{path} between an initial state $\ket{\psi}$ to final state $\ket{\phi}$ as any Lipschitz continuous function on the unit hypersphere, i.e. $f:[0,1]\mapsto S^{d-1}$, with $f(0)=\ket{\psi}$ and $f(1)=\ket{\phi}$ (illustration in \Cref{fig:sphere}; formal definitions in \Cref{scn:discretizing}). We show\footnote{For simplicity in stating the bound on $M$ in \Cref{l:path-interpolation}, we assume $K\in \Theta(1)$, as this suffices for our applications. However, \Cref{l:path-interpolation} also holds for non-constant $K$ with $M\in O(K(\frac{n^2d^2}{\epsilon})^{2n})$ if $0<K\leq 1$ and $M\in O(2^{O(n)}(\frac{K^2n^2d^2}{\epsilon})^{2n})$ if $K>1$.}:

\begin{restatable}[Universal quantum path following lemma]{lemma}{theoremPathInterpolation}\label{l:path-interpolation}
    Set $d:=2^n$, and let $f:[0,1]\to S^{d-1}$ be a $K$-Lipschitz continuous path.
    For every $\epsilon > 0$, there exists a sequence of $M\in O((\frac{n^2d^2}{\epsilon})^{2n})$ $2$-local unitaries $U=U_M\cdots U_1$ which ``$\epsilon$-approximately simulates'' the path $f$ as follows. Define $\ket{\psi_t}=U_t\cdots U_1\ket{\psi_0}$ for $t\in\set{0,\ldots, M}$ and $\ket{\psi_0}:=f(0)$. Then, for all $t$,
    \begin{equation}\label{eqn:pointerror}
        \enorm{\ket{\psi_t}-f(t/M)}\leq \epsilon.
    \end{equation}
\end{restatable}
\begin{figure}[t]
    \centering
    \includegraphics[width=.2\linewidth]{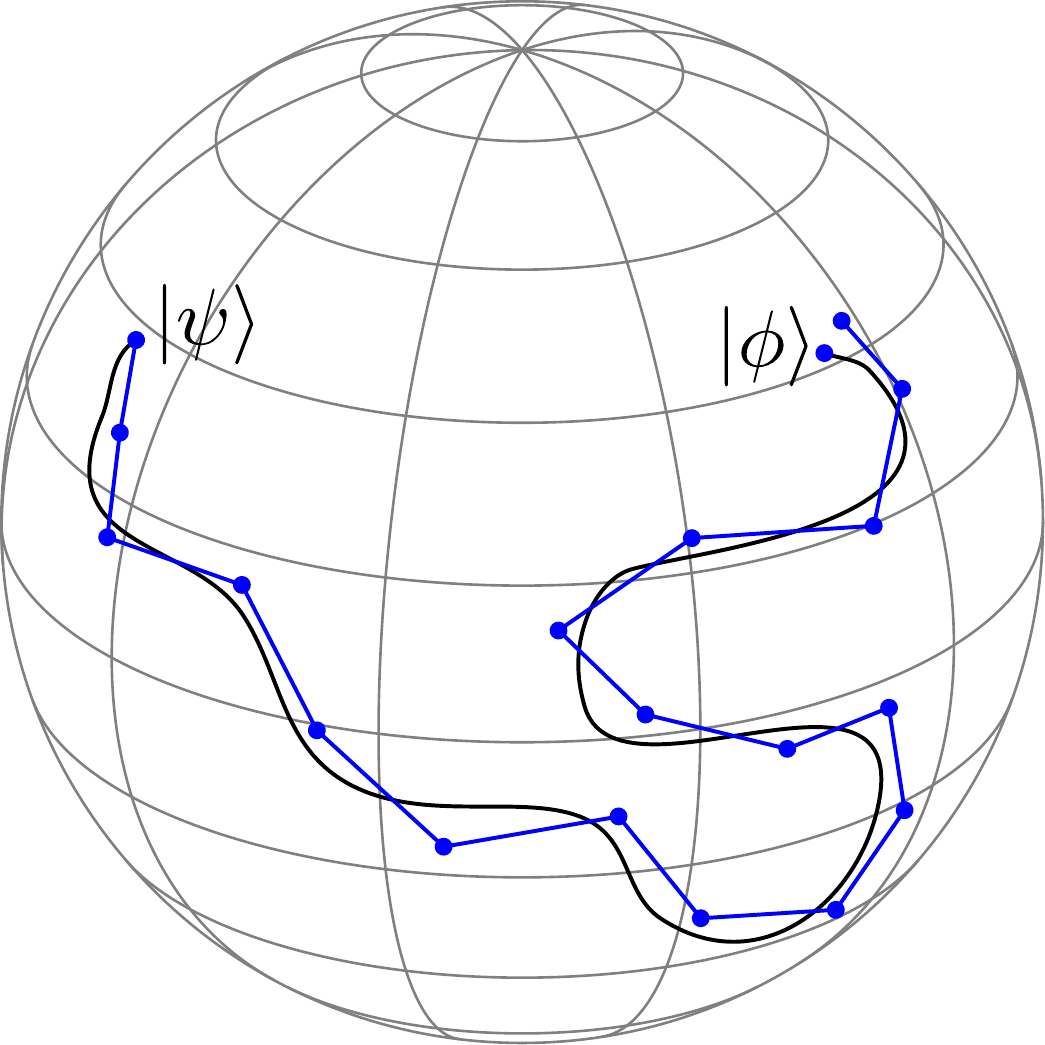}
    \caption{(Color online) Simplified illustration of the Universal quantum path following lemma with $f$ in black (smooth), $\ket{\psi}=f(0),\ket{\phi}=f(1)$, and the path of intermediate states $\ket{\psi_t}$ in blue (piece-wise linear). In the actual construction, each linear segment is itself further subdivided and likewise approximately simulated.}
    \label{fig:sphere}
\end{figure}


With \Cref{l:path-interpolation} in hand, we resolve \Cref{q1} by showing that in the quantum setting, ground spaces of local Hamiltonians are {always} connected in the following sense.
\begin{restatable}{theorem}{theoremtraversal}\label{thm:traversal}
    Let $H\in\hermp{\C^d},\,d=2^n$ with $0\preccurlyeq H \preccurlyeq I$, $\ket{\psi},\ket\phi\in\C^d$ with $\braketb\psi H\le \eta$ and $\braketb\phi H\le\eta$.
    For any $\Delta \ge 2^{-\poly(n)}$, there exists a sequence of $2$-local unitary gates $U=U_m\cdots U_1$ with $m\le 2^{\poly(n)}$ such that
    \begin{enumerate}[label=(\arabic*)]
        \item $\enorm{U\ket\psi - \ket\phi}\le \Delta$, and
        \item for all $i\in [m]$, $\braketb{\psi_i} H \le \eta+\Delta$, where $\ket{\psi_i} := U_{i}\cdots U_1\ket\psi$.
    \end{enumerate}
\end{restatable}
\noindent In words, even if we wish to remain \emph{exponentially} close to the ground space of $H$ throughout the local evolution from $\ket{\psi}$ to $\ket{\phi}$, this can be achieved, at the expense of exponentially blowing up the length of the evolution. Returning to our motivating example of error correcting codes, we conclude: For any $H$, if the ground space of $H$ encodes a quantum error-correcting code, and $\ket{\psi}$ and $\ket{\phi}$ are any pair of code words, then \Cref{thm:traversal} says one can stealthily corrupt $\ket{\psi}$ into $\ket{\phi}$ via a sequence of $2$-qubit gates, so that at any point in the evolution, we are exponentially close to the code space, and thus the corruption is unlikely to be caught via measurement of $H$. The trade-off is that, again, this evolution path ``hugging'' the code space is exponentially long.

As an immediate corollary, we are now able to answer an open question of \cite{Gharibian2018}.
\begin{corollary}[Informal; formally \Cref{cor:GSCONP}]
    \GSCON\ with exponentially many gates and inverse polynomial promise gap is in P.
\end{corollary}
\noindent This follows since by \Cref{thm:traversal}, all \GSCON\ instances in the parameter regime above are YES instances. Thus, allowing an exponentially long proof \emph{trivializes} \GSCON, which is otherwise QCMA-complete in the setting of poly-length proofs~\cite{Gharibian2018}.

As a sanity check, we also strengthen a result of \cite{Gharibian2018} by showing that even an \emph{unbounded} number of $1$-local gates (as opposed to $2$-local gates in \Cref{cor:GSCONP}) with constant promise gap do not suffice to trivialize \GSCON.
\begin{theorem}[Informal; formally \Cref{thm:GSCON1local}]
  \GSCON{} is \PSPACE-complete for $1$-local gates, constant promise gap, and an unbounded number of gates.
\end{theorem}
\noindent The previous PSPACE-hardness result of \cite{Gharibian2018} required inverse exponential promise gap and an exponential bound on the number of gates.

\paragraph{3. Embedding streaming proofs into unentanglement.} We next resolve \Cref{q2} in the positive, showing that streaming proofs can be systematically compressed into exponentially smaller history states.

The formalization of this goes via the Sparse Separable Hamiltonian (SSH) problem (\Cref{def:sepsparse}), which informally is identical to the $k$-local Hamiltonian problem, except for two key differences: (1) $H$ is sparse, rather than local, and (2) proofs are restricted to be in tensor product form. A bit more formally: Given a sparse Hamiltonian $H$ (\Cref{def:sparse}) on $n$ qubits and bipartition $L$ versus $R$ of $[n]$, does there exist $\LRt$ such that
\begin{equation}\label{eqn:qmat}
    \RLt H\LRt
\end{equation}
is ``small'', or does it hold that for all $\LRt$, $\RLt H \LRt$ is ``large''? Note that, in general, optimizations over tensor product states $\LRt\in\C^{d^2}$ are harder than optimizations over \emph{all} $\ket{\psi}\in\C^{d^2}$, i.e. without the tensor product requirement. For example, if $H$ in \Cref{eqn:qmat} had \emph{polynomial} dimension, than maximizing \Cref{eqn:qmat} is NP-hard~\cite{G03}, whereas maximizing $\bra{\psi}H\ket{\psi}$ over {all} $\ket{\psi}\in\C^{d^2}$ is an eigenvalue problem, and thus efficiently solvable in the dimension of $H$. In other words, the optimal solution to a tensor product optimization is \emph{not} necessarily an eigenvector of $H$, and this makes the design and analysis of unentangled proof systems challenging\footnote{For example, Marriott-Watrous~\cite{Marriott2005} \emph{strong} error reduction for QMA (i.e. without increasing the proof size) fails for $\QMAt$, since it crucially leverages the fact that for QMA, the optimal assignment is an eigenvector. The attainment of the ``standard'' notion of \emph{weak} error reduction (i.e. via parallel repetition) by Harrow and Montanaro~\cite{Harrow2013} was considered a breakthrough.}.

We now state our main technical result.
A key parameter is the \emph{promise gap} of the Sparse Separable Hamiltonian problem.
Chailloux and Sattath~\cite{Chailloux2012} show SSH is $\QMAt$-complete (\Cref{def:QMAt}) for inverse polynomial promise gap.
We show:

\begin{lemma}[(Informal) Embedding lemma; formally \Cref{l:embed}]
    Let $p,q,r,m,\alpha,\beta:\R\mapsto\R$, where $p,q,r$ are poly-bounded. Let $Q$ be a quantum circuit consisting of $m$ $2$-qubit gates, taking in (1) input $x\in\Sigma^n$, (2) a classical streaming proof ${y}\in \set{0,1}^{2^{p}}$, and (3) $q$ ancilla qubits initialized to all zeroes. We are promised that either there exists a streaming proof $y$ causing $Q$ to accept with probability at least $\alpha$, or all streaming proofs are accepted with probability at most $\beta$, for $\alpha-\beta\geq 2^{-r}$. Then, there exists a poly-time many-one reduction from $(Q,x)$ to a Separable Sparse Hamiltonian $H$ instance with norm $\snorm{H}\in\poly(m,2^r)$, and with thresholds $\alpha',\beta'$, such that:
    \begin{enumerate}
        \item $H$ acts on $O(q+\log m)$ qubits.
        \item The promise gap scales as $\abs{\alpha'-\beta'}\in\Omega\left(\frac{1}{m2^r}\right)$.
    \end{enumerate}
\end{lemma}
\noindent In words, any quantum verification $Q$ with $q$ qubits as workspace and taking in a classical proof of length $2^p$ can be compressed to a Separable Sparse Hamiltonian instance on $O(q+p)$ qubits and with promise gap scaling\footnote{This statement assumes the verification time $m$, proof length $2^p$, and promise gap $2^r$ are polynomially related, which is a reasonable setting. Of course, in general, these relationships need not hold. What we \emph{can} assume without loss of generality is that $m\geq 2^p$ to allow $Q$ to read the entire proof. This means the two potentially dominating terms are $m$ and $2^r$, which is why these appear in the norm and promise gap of $H$ in \Cref{l:embed}.} as $1/2^p$. Moreover, the mapping (1) preserves the space required up to poly overhead, and (2) embeds the proof of length $2^p$ bits into $\sim p$ qubits. To the best of our knowledge, this is the first such systematic method for compressing arbitrary classical proofs via unentanglement.\\

\noindent \emph{Applications of the Embedding Lemma.} \Cref{l:embed} immediately applies to arbitrary $\SQCMASPACE$ verifiers. Here, however, we focus on the application to MIP:
\begin{corollary}[(Informal) Reducing MIP to unentanglement; formally, \Cref{cor:MIP}]\label{cor:MIPinformal}
    There exists a poly-time many-one reduction from any classical multi-prover interactive protocol ($\MIP$, \Cref{def:mip}) with $p$ provers, $r$ rounds, $u$ space, and $t$ bits of communication per round, to an instance of Separable Sparse Hamiltonian on $\widetilde{O}(u)$ qubits with promise gap scaling dominated by scaling $2^{-tr}$. (The tilde in $\widetilde{O}$ hides polylogarithmic factors.)
\end{corollary}
\noindent For context, recall that $\MIP$ with two provers, one round and polynomially many bits of communication equals NEXP \cite{Babai1990,FL92} (formal restatement in \Cref{thm:mipnexp}). As for NP, it is contained in MIP with 2 provers, 1 round, and logarithmic bits of communication (see \Cref{sscn:MIP}). In words, \Cref{cor:MIPinformal} says that any MIP protocol can be reduced to an SSH instance, with the key parameter being the number of bits $t$ of communication; this is what dictates the promise gap of the SSH instance $H$ we obtain. Note we also preserve the space used by the MIP protocol (which is important for \Cref{cor:NP} for the case of NP, where the MIP uses log-space).

With \Cref{l:embed} in hand, we next show various $\QMAt$-type containments. For this, we first show that the specific Hamiltonian construction $H$ output by the Embedding Lemma can be decided in $\QMAt$ using appropriate parameters:
\begin{lemma}[Informal; see \Cref{l:QMAtcontain}]
    Let $H$ be the Sparse Separable Hamiltonian instance produced by the Embedding Lemma, acting on $n$ qubits and with promise gap $g$. Then, $H$ can be decided by a $\QMAt$ verifier acting on $O(n)$ qubits and with promise gap $O(g)$.
\end{lemma}
\noindent As an aside, at present we are curiously unable to show \Cref{l:QMAtcontain} \emph{without} exploiting the specific structure\footnote{In other words, given an arbitrary sparse Hamiltonian $H$ of potentially exponential norm, it is not clear to us how one would verify it in \QMAt with (say) $1/\exp$ promise gap. (For example, Quantum Phase Estimation (QPE) would seemingly fail --- see \Cref{sscn:containmentqmat} for a brief discussion.)} of $H$ from the Embedding Lemma.

Finally, by combining \Cref{l:embed} and \Cref{l:QMAtcontain}, we obtain the following two main corollaries:

\begin{corollary}[Informal; see \Cref{cor:QSPACEinQMAt}]
    $\SQCMASPACE$ with proof length $2^p$, $q$ ancilla qubits, and promise gap $1/2^r$ is contained in $\QMAt$ with $q+\log p$ proof and ancilla qubits, and promise gap $1/2^{p+r}$.
\end{corollary}
\noindent Above, note that $p$ and $r$ are polynomially \emph{bounded}, i.e. logarithmic $p$ and $r$ are allowed.

\begin{corollary}[Informal; see \Cref{cor:MIPinQMAt}]
    $\MIP$ with $t$ bits of communication per round, space $u$, $v$ random bits, $p$ provers, $r$ rounds, and completeness/soundness $c$ and $s$, respectively, is contained in $\QMAt$ with $u+v+\log(tr\log(pt))$ proof and ancilla qubits, and promise gap $2^{-tr\log(pt)+\log(c-s)}$.
    \end{corollary}
\noindent Thus, we obtain the first systematic $\QMAt$-type bounds on arbitrary multi-prover interactive protocols. Above, the $\QMAt$ verifier requires the same amount of ancilla space as the MIP, and the $\QMAt$ promise gap depends exponentially on the total amount of communication but only polynomially on the MIP promise gap. As a bonus, we also immediately rederive (\Cref{cor:NEXP-QMA(2)}) in a unified fashion the results $\NP=\PQMAlogt$~\cite{BT12} and $\NEXP=\PreciseQMAt$~\cite{Per12}.

Finally, as a last application of the Embedding Lemma, we return to our study of \GSCON by showing \NEXP-completeness for a variant of \GSCON:

 \begin{theorem}[Informal; formally \Cref{thm:SSGSCON-NEXP}]
   \GSCON is \NEXP-complete with a sparse Hamiltonian, an inverse exponential promise gap, and an exponential number of $2$-local gates which may not act across a given $L$ versus $R$ cut of the qubits (i.e. all intermediate states are product across the $L$ versus $R$ cut).
 \end{theorem}

\subsection{Techniques}\label{sscn:techniques}

The proof of $\SQCMASPACE=\NEXP$ (\Cref{thm:qcmaspace}) follows easily from the PCP characterization of $\NEXP$~\cite{Babai1990}; see \Cref{sec:savitch}. As for $\SQMASPACE=\QMAEXP$ (\Cref{thm:qmaspace}), the obstacle is to show that (weak) error reduction holds for $\SQMASPACE$.
This is because with only poly-size ancilla space, the verifier seemingly can only repeat the verification a polynomial number of times, which is not enough to amplify an exponentially small promise gap.
We overcome this by forcing the streamed proof itself to repeatedly replenish the verifier's ancilla, and run a pair of counters
to both ensure the prover sends correctly initialized ancilla qubits all set to zero, along with sufficiently many ``good'' proofs accepted with high probability.

The main technical contributions of this work, however, are the Universal quantum path following lemma (\Cref{l:path-interpolation}) and the Embedding lemma (\Cref{l:embed}), which we now discuss.

\paragraph{1. Universal Quantum Path Following Lemma.} Recall \Cref{l:path-interpolation} shows how to simulate any Lipschitz continuous path on the unit hypersphere via an exponentially long sequence of local gates. To show this, we first ``discretize'' the given path $f$ into a dense enough sequence of points $\ket{\psi_1},\dots,\ket{\psi_N}$ so that each consecutive pair of points $\ket{\psi_j}$ and $\ket{\psi_{j+1}}$ is ``close''. Thus, if \emph{global} (i.e. $n$-local) unitaries were allowed, a ``small rotation'' (i.e. close to identity) would suffice to exactly map $\ket{\psi_j}$ to $\ket{\psi_{j+1}}$.
However, here we are restricted to $2$-local gates, and the typical approach~\cite{NC00} to simulate {global} rotations using {$2$-local} gates would yield intermediate states possibly very far from $\ket{\psi_j}$ and $\ket{\psi_{j+1}}$ (and more generally, from the desired path $f$).
Hence, we devise a general decomposition of global unitaries close to identity into $2$-local gates close to identity.
Specifically, we give an approximate decomposition $e^{itH} \approx \prod_j e^{it_jH_j}$, where $\sum_j\abs{t_j}$ is bounded by $2^{\poly(n)}\abs t^{1/2n}$.
Basically, we can decompose a unitary with a short \emph{pulse time} into many local unitaries with short pulse times, which allows us to map a quantum state along each segment $\ket{\psi_j}$ to $\ket{\psi_{j+1}}$.

For that, we first write $H = \sum_j \alpha_j P_j$ in the Pauli basis (i.e. each $P_j$ is a tensor product of the Pauli matrices and identity) and apply the Suzuki decomposition~\cite{Suzuki1976} (\Cref{lem:custom-suzuki})
$ e^{iH} = \prod_j e^{i\alpha_j P_j} + O(t^2), $
where $\sum_j\abs{\alpha_j}\le t$.
Clinton, Bausch, and Cubitt~\cite{CBC21} give an exact $2$-local decomposition for the $e^{i\alpha_j P_j}$ terms with bounded pulse times.
We provide an alternative construction with a simpler analysis, and which requires a polynomial number of gates to implement a Hamiltonian interaction (compared to exponential in \cite{CBC21}), at the cost of a slightly worse pulse time bound compared to \cite{CBC21}.

In terms of application, recall \GSCON asks whether there exists a sequence of local unitaries mapping ground state $\ket\psi$ of $H$ to orthogonal ground state $\ket\phi$, while remaining in low energy space.
Since we can apply \Cref{l:path-interpolation} to \emph{arbitrary} Lipschitz continuous paths, we can apply it to path
\begin{equation}
    f(t)= \cos(t\pi/2)\ket\psi + \sin(t\pi/2)\ket\phi,
\end{equation}
where note that for all $t$, $f(t)$ is also a ground state\footnote{Note \Cref{thm:traversal} also applies when $\ket\psi$ and $\ket\phi$ are not ground states, but just low energy states.} of $H$.
Thus, \Cref{l:path-interpolation} allows us to ``follow'' this path via $2$-qubit gates, yielding a suitable gate sequence $U_m\cdots U_1$ for \GSCON.
In general, this sequence requires an exponential number of gates, and in return achieves exponential precision.

\paragraph{2. The Embedding Lemma.} \Cref{l:embed} shows how to compress any quantum verification $Q$ with $q$ qubits as workspace and taking in a streaming classical proof of length $2^p$ into a Separable Sparse Hamiltonian instance on $O(q+p)$ qubits and promise gap scaling as $1/2^p$. So, let $Q=V_m\cdots V_1$ be a quantum verifier taking in streaming proof $y$. Recall we formalized ``streaming'' by partitioning the gates $\set{V_i}$ into two sets: ``Proof gates'' indexed by $P\subseteq [m]$, and ``computation gates'' indexed by $[m]\setminus P$. Our goal is to design a Hamiltonian $H$ so that, when there exists proof $y$ accepted by $Q$, then an appropriately defined history state $\ket{\psihist}$ has low energy against $H$. The problem is that we do not know the proof gates $\set{V_i}$ with $i\in P$ while computing the reduction---the verifier $Q$ only learns this information in the \emph{future}. To overcome this, at a very high level, we instead demand an appropriately defined \emph{unentangled} history state of form $\ket{\psihist}_L\otimes\ket{\psihist}_R$. We then exploit this ``unentanglement'' to logically simulate a quadratic Boolean function across the two copies of $\ket{\psihist}$, in turn allowing the history state to decide ``on-the-fly'' whether it wishes to stream proof bit $0$ or $1$ in step $t\in P$.

Formally, we define our Hamiltonian as (details in \Cref{scn:embedding}) $\Ht = \Din\Hint+\Dprop\Hpropt+\Dsym\Hsym+\Houtt$ for some weights $\Din,\Dprop,\Dsym\geq 0$.
Briefly, $\Hint$ and $\Houtt$ ensure that in any candidate proof $\ket{\psi}_L\otimes\ket{\phi}_R$, both $\ket{\psi}_L$ and $\ket{\phi}_R$ are initialized correctly at time $t=0$ and accept at time $t=m$. $\Hsym$ enforces that a low energy state is symmetric under exchange with respect to the cut $L$ versus $R$.
The key ingredient, however, is hiding in $\Hpropt$, and is the \gadget\ gadget\footnote{This gadget will allow the history state to ``encode the future''; the name is thus a reference to the ``flux capacitor'', which makes time travel possible in the film \emph{Back to the Future}.},
\begin{equation}\label{eqn:gad}
    (\HttI)_L\otimes (\HttiX)_R  +(\HttiX)_L\otimes (\HttI)_R,
\end{equation}
used to encode \emph{future} streamed proof gates (i.e. for time steps $t\in P$).

This gadget works as follows. A propagation term $\HttI$ or $\HttiX$ enforces that at time $t\in P$, the local proof $\ket{\psi}$ applies proof gate $I$ (to simulate streaming bit $0$) or proof gate $iX$ (to simulate streaming bit $1$), respectively. Since we do not know in advance which of these two gates will be applied, we run a thought experiment --- imagine we had two parallel universes, where universe $L$ streams bit $0$, or universe $R$ streams bit $1$. This can be simulated via term $(\HttI)_L\otimes (\HttiX)_R$ --- namely, since the tensor product is multiplicative, this constraint is satisfied, i.e.
\begin{equation}
    \left((\HttI)_L\otimes (\HttiX)_R\right) \ket{\psi}_L\otimes\ket{\phi}_R=\HttI\ket{\psi}_L\otimes \HttiX\ket{\phi}_R=0,
\end{equation}
only if either universe $L$ (i.e. $\ket{\psi}_L$) applies gate $I$ \emph{or} universe $R$ (i.e. $\ket{\psi}_R$) applies gate $iX$, or both. The keyword here is ``or'', in that the tensor product allows us to simulate the Boolean OR function between universes. Of course, we have not yet achieved anything, since neither universe has any choice in which bit it streams! This brings us back to the \gadget\ gadget --- observe that the ``+'' in Equation \eqref{eqn:gad} acts as a Boolean ``AND''. In other words, to satisfy the gadget, universe $L$ can apply $I$ (this annihilates the first term, $(\HttI)_L\otimes (\HttiX)_R$) and $R$ can apply $I$ (this annihilates the second term, $(\HttiX)_L\otimes (\HttI)_R)$. Similarly, both can instead choose to apply $iX$ to satisfy the gadget. The conclusion is that both universes can freely decide which proof bit to stream at time $t\in P$, so long as they stream the \emph{same} bit! Indeed, this works because we have exploited unentanglement to simulate the quadratic Boolean function EQUALS:
$
    ({x}\vee \overline{y})\wedge (\overline{x}\vee{y}) \leftrightarrow x=y \quad\text{for }x,y\in\set{0,1}.
$

The next challenge is to prove soundness of the construction, where recall tensor product optimizations are difficult to analyze since optimal solutions do not correspond to eigenvectors (and thus, standard techniques from the study of $\kLH$ cannot be directly employed). Indeed, this step is rather involved (a step-by-step derivation of the construction is in~\Cref{sscn:ingredients}). For example, the careful reader might wonder why we chose $iX$ to stream bit $1$ rather than simply $X$ --- it turns out use of $X$ breaks soundness of the \gadget\ gadget. Even when we use $iX$, without the symmetry constraint $\Hsym$, soundness again breaks via simultaneous cheating across \emph{multiple} \gadget\ gadgets. 

To overcome these obstacles, very briefly, our analysis first exploits the large weight $\Dsym$ to enforce any low energy state to look like $\ket{\psi}_L\otimes \ket{\psi}_R$ for some $\ket{\psi}$. To next force $\ket{\psi}$ to look like an actual history state, two ingredients smoothly fit together. First, since we use $iX$ instead of $X$ in the \gadget\ gadget, it turns out that for any choice of assignment $\ket{\psi_1}_L$ on $L$, a low energy state $\ket{\psi_2}_R$ on system $R$ must implement at time $t\in P$ the operator
    \begin{equation}\label{eqn:U}
        U(a_t,b_t)= \frac{1}{\sqrt{a_t^2+b_t^2}}(a_tiX+b_tI).
    \end{equation}
    for some $a_t,b_t\geq 0$. Now, due to the $i$ in $iX$, $U(a_t,b_t)$ turns out to be \emph{unitary}. Thus, conditioned on any fixed assignment on $L$, we {can} ``invert'' $U(a_t,b_t)$ by applying Kitaev's change-of-basis operator \cite{KSV02}, thus diagonalizing what we call the ``residual propagation term on $R$'',
    \begin{equation}
        \bra{\psi_1}\HttI\ket{\psi_1}(\HttiX)_R  +\bra{\psi_1}\HttiX\ket{\psi_1}(\HttI)_R.
    \end{equation}
    The second ingredient is to show that by setting $\Dprop$ large enough, we can extract a ``proper'' propagation Hamiltonian hiding under this ``residual operator on $R$'' over \emph{all} time steps. This allows us to force any low energy state of $\Ht$ to indeed be of form $\ket{\psihist}_L\otimes\ket{\psihist}_R$ --- which is \emph{almost} what we want.

    The final problem is that for any $t\in P$, $\ket{\psihist}$ is currently forced to apply a unitary of form $U(a_t,b_t)$ from Equation \eqref{eqn:U} for some $a_t,b_t$. What we \emph{actually} want is for the \gadget\ gadget to act like a ``switch''---either $a_t=0$ and $b_t\gg 0$ (streaming proof bit $0$) or $a_t\gg 0$ and $b_t= 0$ (streaming proof bit $1$). By carefully exploiting the structure of $U(a_t,b_t)$ itself, we finally show that any low energy $\ket{\psihist}_L\otimes\ket{\psihist}_R$ can be ``rounded'' to obtain a state closeby which perfectly satisfies this desired ``switch'' behavior for all $t\in P$.

\subsection{Related Work}\label{sscn:related-work}

\paragraph{GSCON.}
In the classical setting, Gopalan, Kolaitis, Maneva, and Papadimitriou~\cite{Gopalan2006} show the problem of determining whether two solutions of a Boolean formula are connected through its solution space is in \p or \PSPACE-complete, depending on the types of constraints allowed in the formula.
The \GSCON problem was introduced by Gharibian and Sikora~\cite{Gharibian2018}, who show that \GSCON with $m=\poly(n)$ $(l=2)$-local unitaries is \QCMA-complete.
For $m=\exp(n)$ and $l=1$, it is \PSPACE-complete.
If the Hamiltonian is given as a succinct circuit description, \GSCON is \NEXP-complete.
Gosset, Mehta, and Vidick~\cite{Gosset2017} show the surprising result that \QCMA-completeness holds even for \emph{commuting} local Hamiltonians (an analogous result for QMA-completeness of $\kLH$ on \emph{commuting} Hamiltonians remains an open question).
Nagaj, Hangleiter, Eisert, and Schwarz~\cite{NHES21} next show \QCMA-completeness for stoquastic Hamiltonians.
Watson, Bausch, and Gharibian~\cite{WBG20} study \GSCON with a \emph{translationally invariant Hamiltonian} on a 1D chain of qudits (i.e. there exists a single 2-local Hamiltonian acting on each pair of neighbors in the chain) and prove $\QCMA_{\EXP}$-completeness ($\QCMA_{\EXP}$ is \QCMA with exponentially large proof and exponential-time quantum verifier).
We remark that the EXP in $\QCMA_{\EXP}$ arises due to the translation-invariance, which causes the encoding size of the problem to be exponentially smaller than the size of the chain.

\paragraph{QMA(2).}
The complexity class $\QMA(k)$ (\QMA with $k$ unentangled proofs) was first introduced by Kobayashi, Matsumoto, and Yamakami~\cite{KMY03}.
Blier and Tapp~\cite{BT12} show that $\NP\subseteq \QMA_{\log}(2)$ ($\QMAt$ but with  log-sized proofs and a log-space verifier) with perfect completeness and $1-1/\poly$ soundness.
Aaronson et al.~\cite{Aaronson2008} give a $\QMA_{\log}(\tilde O(\sqrt{n}))$ protocol for \tSAT with a constant soundness gap (as opposed to $1/\poly$ in \cite{BT12}).
They further argue that assuming a weak version of the Additivity Conjecture from quantum information theory, $\QMA(k)=\QMAt$ for all $k\ge 2$ and $\QMA(2)$ can be amplified to exponentially small error.
Harrow and Montanaro~\cite{Harrow2013} prove this statement by developing a protocol for a \emph{product test} that allows a quantum verifier to check if a state is a product state across $n$ cuts, given two copies.
It also follows that \tSAT has a \QMAt protocol with proof size $\tilde O(\sqrt n)$.
We remark that it remains an open problem whether \QMAt is equal to \NEXP, though an oracle separation to \coNP exists~\cite{KMY03}. Gharibian, Santha, Sikora, Sundaram and Yirka~\cite{GSSSY18} define quantum generalizations of the Polynomial Hierarchy, QCPH and QPH (using classical and quantum proofs, respectively, and quantum verifiers in both cases), and show that (1) if $\textup{QCPH}=\textup{QPH}$, then $\QMAt$ is in the Counting Hierarchy, and (2) unless $\QMAt=\textup{Q}\Sigma_3$ ($\textup{Q}\Sigma_3$ the third level of QPH), $\QMAt$ is strictly contained in NEXP.

Chen and Drucker~\cite{CD10} improve upon~\cite{Aaronson2008} with a $\BellQMA_{\log}(\tilde O(\sqrt{n}))$ protocol for \tSAT, where $\BellQMA(k)$ is defined as $\QMA(k)$ without entangled measurements.
\QMAt permits an inverse polynomial gap, however with an exponentially small gap it is equal to \NEXP as shown by Pereszlényi~\cite{Per12}.
With a linear number of provers and an exponential soundness gap, $\BellQMA$ equals \NEXP as well.
Kinoshita~\cite{kinoshita2018qma2} proves that \QMAt with postselection also equals \NEXP.
Chiesa and Forbes~\cite{CF13} give a tight soundness analysis of the protocol of \cite{BT12}, showing a soundness gap $\Omega(n^{-1})$, notably without using a \PCP.
They further improve upon~\cite{CD10} by providing a smooth trade-off between the number of provers $k$ and the soundness gap $\Omega(k^2/n)$.
Chailloux and Sattath~\cite{Chailloux2012} show the Separable Sparse Hamiltonian problem with $1/\poly$ promise gap is complete for $\QMA(2)$.
Sparsity is crucial here, as \cite{Chailloux2012} shows the Separable \emph{Local} Hamiltonian problem is in $\QMA$.

\paragraph{Space-bounded quantum computation.}
Watrous~\cite{Watrous1999,Watrous2003} initiates the study of space-bounded quantum computation and shows $\BQSPACE(s(n))\subseteq \SPACE(O(s(n)^2))$, where \BQSPACE is the space-bounded variant of \BQP with intermediate measurements.
It follows that $\BQPSPACE = \PSPACE$.
Fefferman and Lin~\cite{Fefferman2018} prove that \QMA with an inverse exponentially small gap, denoted \PreciseQMA, is equal to PSPACE, by showing that $\BQUSPACE(s(n))$ (like \BQSPACE but with only unitary gates) equals \QMA with a poly-time verifier, $O(s(n))$ space and proof size, and $2^{-O(s(n))}$ soundness gap.
Consequently, the precise local Hamiltonian problem (inverse exponential precision) is \PSPACE-complete.
Fefferman and Remscrim~\cite{FR21} improve upon these results by showing $\BQUSPACE(s) = \BQSPACE(s) = \QUMASPACE(s) = \QMASPACE(s)$. (For clarity, recall QMASPACE receives a \emph{poly}-sized \emph{quantum} proof, whereas in this work \SQCMASPACE takes an \emph{exponential} size \emph{classical} proof.)
Notably, they are able to eliminate intermediate measurements, which is nontrivial in the space-bounded setting as deferred measurements require a fresh ancilla for each measurement.

\subsection{Open questions}\label{sscn:open-problems}
First, while we have given characterizations for both $\SQCMASPACE$ and $\SQMASPACE$, our focus has primarily been on \emph{classical} streamed proofs. Discovering further properties of \emph{quantum} streamed proofs is thus left as a natural open question.

Next, via the Universal Quantum Path Following Lemma (\Cref{l:path-interpolation}), we showed that \GSCON\ with exponentially many gates and inverse poly promise gap is in P (\Cref{cor:GSCONP}). However, what remains unclear is the complexity of \GSCON\ with exponentially many gates and inverse \emph{exponential} promise gap. Then, depending on the exact size of the gap and number of unitaries allowed, \Cref{l:path-interpolation} does not necessarily apply, and indeed, in \Cref{thm:PSPACEGSCON} we show that $\GSCON$ in this setting is $\PSPACE$-hard. The only progress we are able to make here is \Cref{thm:SSGSCON-NEXP}, which requires a \emph{sparse} (versus local) Hamiltonian and predefined $L$ versus $R$ cut across which gates may not act (whereas originally $\GSCON$ has no such restriction). Second, whereas the classical analogue of \GSCON, \STCONN, satisfies a dichotomy theorem (i.e. is either in P or PSPACE-complete depending on the constraints allowed)~\cite{Gopalan2006}, a similar result remains unknown for \GSCON.

In terms of unentanglement, the Embedding Lemma (\Cref{l:embed}) recovers the result of \cite{BT12} for NP with \emph{log}-size $\QMAt$ proofs, and in particular, also with an inverse poly promise gap. Whether this gap can be improved to \emph{constant} while maintaining a log-size proof remains open. Next, can an analogue of \Cref{l:QMAtcontain} be shown \emph{without} assuming the structure on $H$ guaranteed by the Embedding Lemma? Recall our proof of \Cref{l:QMAtcontain} crucially leveraged the latter. Finally, the complexity of $\QMAt$ remains frustratingly open --- is $\QMAt=\NEXP$? What other natural complete problems are there for $\QMAt$ beyond the (inverse poly-gapped) Separable Sparse Hamiltonian~\cite{Chailloux2012}?

\subsection{Organization}
\Cref{scn:defs} begins with all relevant definitions. In \Cref{scn:quantum-npspace}, we give our no-go result for a quantum analogue of Savitch's theorem and analyze streaming quantum proofs. \Cref{scn:discretizing} gives our first main result, the Universal Quantum Path Following Lemma. This is then applied in \Cref{sec:application-gscon} to show \Cref{thm:traversal}, i.e. that the low energy space of any Hamiltonian is always ``connected'' in the presence of exponentially long local gate sequences. \Cref{scn:embedding} gives our second main result, the Embedding Lemma, with applications in \Cref{scn:embeddingapps}. Appendices \ref{sec:gscon-exp-pspace} and \ref{app:2} study variants of \GSCON with exponentially many local gates.

\section{Definitions}\label{scn:defs}
We begin by defining \SQCMASPACE. Remarks: First, the definition of \SQCMASPACE\ will allow inverse exponential promise gap, although we show in \Cref{thm:qcmaspace} that without loss of generality, we may assume a constant promise gap. Second, intermediate measurements are not free in our model, and so only a polynomial number of simulated measurements via the principle of deferred measurement~\cite{Nielsen} can be made in polynomial space. To begin, we model a \emph{streaming} classical proof via the quantum circuit model as follows.

\begin{definition}[Streaming classical proof]\label{def:stream}
    Let $U$ be a quantum circuit acting on an $n_1$-qubit input register $R_1$, $n_2$-qubit ancilla register $R_2$, and $1$-qubit proof register $R_3$, for some $n_1,n_2>0$. Registers $R_2$ and $R_3$ are initialized to all zeroes. At a high level, the idea is to stream a classical proof in register $R_3$ one bit at a time. To do so, we view the entire execution of $U$ as a sequence of $1$- and $2$-qubit gates, but where certain $1$-qubit gates on $R_3$ are \emph{a priori} unknown. Formally:
    \begin{enumerate}
        \item There are two main phases in the circuit, which repeat until the circuit completes. In iteration $i$:
        \begin{enumerate}
            \item (Computation phase) A sequence of $1$- and $2$-qubit gates acts solely on registers $R_1$ and $R_2$.
            \item (Proof phase)
            \begin{enumerate}
                \item (Compute) Single-qubit gate $W_i\in\set{I,X}$ is applied to $R_3$, for $X$ the Pauli NOT gate. \label{step:proof1}
                \item (Copy) $R_3$ is classically copied into $R_2$ via CNOT gate (controlled from $R_3$ onto $R_2$).\label{step:proof2}
                \item (Uncompute) $W_i$ is applied to $R_3$ to return $R_3$ to $\ket{0}$.\label{step:proof3}
            \end{enumerate}
        \end{enumerate}
    \end{enumerate}
\end{definition}
\noindent \emph{Remarks.} Above, we view each gate $W_i$ as being applied dynamically by the prover, i.e. each time the computation phase ends, the prover supplies the next bit. In principle, this can be embedded into an interactive proof, although this is possibly overkill, as all communication is one-way (from prover to verifier). Further clarifications: (1) Each time the computation phase is run, the sequence of gates applied need not be the same as in the previous computation phase. (2) For simplicity, we may assume without loss of generality that the computation ends with a proof streaming phase in which $W_i=I$. (3) Without loss of generality, in Step \ref{step:proof2} we assume there is a fixed qubit in $R_2$, say $q$, to which the content of $R_3$ is copied each time. If $U$ does \emph{not} wish to use the next proof bit, it may set $q$ to $\ket{+}$ just before Step \ref{step:proof2}, so that the CNOT gate of \ref{step:proof2} acts invariantly.

\begin{restatable}[Streaming-$\QCMASPACE$ ($\SQCMASPACE(p,q,r)$)]{definition}{definitionQCMASPACE}\label{def:QCMASPACE}
    A promise problem $A=(\ayes,\ano)$ is in $\SQCMASPACE(p,q,r)$ for polynomially-bounded functions $p,q,r$, if there exist thresholds $\alpha(n),\beta(n)$ satisfying $\alpha(n)-\beta(n)\geq 2^{-r(n)}$, and a $q(n)$-space uniform family of quantum circuits $\set{Q_n}$ with properties as follows. $Q_n$ takes as input a string $x\in\Sigma^n$, a classical streaming proof ${y}\in \set{0,1}^{2^{p(n)}}$, and $q(n)$ ancilla qubits in state $\ket{0}^{\otimes q(n)}$. We say $Q_n$ accepts $(x,y)$ with probability $p$ if on input $(x,y)$, measuring $Q_n$'s dedicated output wire in the standard basis yields $1$ with probability $p$. Then:
    \begin{itemize}
    \item (Completeness) If $x\in\ayes$, there exists a streaming proof $y\in\set{0,1}^{2^{p(n)}}$ such that $Q_n$ accepts $(x,y)$ with probability at least $\alpha$.
    \item (Soundness) If $x\in\ano$, for all streaming proofs ${y}\in \set{0,1}^{2^{p(n)}}$, $Q_n$ accepts $(x,{y})$ with probability at most $\beta$.
    \end{itemize}
    Finally, let the input, ancilla, and proof registers be denoted $R_1$, $R_2$, $R_3$ respectively. To enforce that $R_1$ and $R_3$ are not used as ancilla, we require that $Q_n$ only acts on $R_1$ and $R_3$ via CNOTs with the control in $R_1$ or $R_3$ and the target in $R_2$.
\end{restatable}

\noindent Note the use of term ``polynomially-\emph{bounded}''---thus, $r=\log n$ is allowed above. For clarity, a polynomial-space Turing machine is bounded only in its workspace tape length; its output tape is unbounded to allow for outputting the (exponential length) quantum circuit $Q_n$. 

\begin{remark}\label{rem:1}
    Throughout this paper, for $\SQCMASPACE$ and all other complexity classes below, we slightly abuse notation and use $\SQCMASPACE(p,q,r)$ to mean $\SQCMASPACE(O(p),O(q),O(r))$ (i.e. we omit explicitly writing the Big-Oh each time).
\end{remark}

\begin{definition}[Sparse Hamiltonian (e.g.~\cite{Chailloux2012})]\label{def:sparse}
A Hermitian operator $H \in \herm\left((\C^2)^{\otimes n}\right)$ is \emph{row-sparse} if each row of $H$ has at most $\poly(n)$ non-zero entries, and if there exists an efficient classical algorithm mapping row index $i\in[2^n]$ to a sequence of all non-zero entries $H_{ij}$ of $H$.
\end{definition}

\begin{definition}[Separable Sparse Hamiltonian (SSH($g$)) \cite{Chailloux2012}]\label{def:sepsparse}
    Let $g:\N\mapsto\R$ be an efficiently computable function. Given as input a sparse Hamiltonian $H$, a bipartition $L$ versus $R$ of the $n$ qubits $H$ acts on, and threshold parameters $\alpha,\beta$ satisfying $\beta-\alpha\geq 1/g(n)$, decide:
    \begin{itemize}
        \item (YES case) If there exists $\LR$ such that $\RL H \LR\leq \alpha$, output YES.
        \item (NO case) If for all $\LR$, $\RL H \LR\geq \beta$, output NO.
    \end{itemize}
\end{definition}
\noindent Chailloux and Sattath show that the Separable Sparse Hamiltonian problem with inverse polynomial gap, $\SSH(1/\poly)$, is QMA(2)-complete, for QMA(2) defined next.

\begin{definition}[$\QMA(2,p,q,r)$ \cite{KMY03}]\label{def:QMAt}
    A promise problem $A=(\ayes,\ano)$ is in $\QMA(2,p,q,r)$ for polynomially bounded functions $p,q,r$ if there exist thresholds $\alpha(n),\beta(n)$ satisfying $\alpha(n)-\beta(n)\geq 2^{-r(n)}$, and a poly-time uniform family of quantum circuits $\set{Q_n}$ with properties as follows. $Q_n$ takes as input a string $x\in\Sigma^n$, a quantum proof $\ket{\psi_1}_L\otimes\ket{\psi_2}_R\in\C^{2^{p(n)}}\otimes \C^{2^{p(n)}}$, and $q(n)$ ancilla qubits in state $\ket{0}^{\otimes q(n)}$. We say $Q_n$ accepts $(x,y)$ with probability $p_{\textup{acc}}$ if on input $(x,\LR)$, measuring $Q_n$'s dedicated output wire in the standard basis yields $1$ with probability $p_{\textup{acc}}$. Then:
    \begin{itemize}
    \item (Completeness) If $x\in\ayes$, there exists a $\LR$ such that $Q_n$ accepts $(x,\LR)$ with probability at least $\alpha$.
    \item (Soundness) If $x\in\ano$, for all $\LR$, $Q_n$ accepts $(x,\LR)$ with probability at most $\beta$.
    \end{itemize}
\end{definition}
\noindent \emph{Caution:} We are using slightly non-standard notation above, in that the promise gap scales as $2^{-r}$, whereas typically in the literature the parameter $r$ would define a $1/r$ gap. This is to align with our definition of (e.g.) \SQCMASPACE, which can have an exponentially small promise gap. Next, by setting $p,q,r$ appropriately, \Cref{def:QMAt} captures the variants of $\QMAt$ studied thus far in the literature (as far as we are aware): When $p,q\in\poly(n)$ and $r\in \log n$, we recover $\QMAt$ \cite{KMY03},  $p,q,r\in\poly(n)$ yields $\PreciseQMAt$ \cite{Per12}, and $p,q,r\in \log(n)$ gives $\pQMAlogt$ \cite{BT12} (for $\pQMAlogt$, only $c=1$ versus $s=1-1/\poly(n)$ is known, i.e. error reduction to arbitrary $c$ and $s$ remains open without blowing up the proof size superlogarithmically). Note that even when $q\in \log(n)$, the circuit $Q_n$ may still consist of $\poly(n)$ gates. Harrow and Montanaro~\cite{Harrow2013} have shown that error reduction holds for $\QMAt$, i.e. we may assume $\alpha$ and $\beta$ are exponentially close to $1$ and $0$, respectively.

\subsection{Streaming-QMASPACE}\label{sscn:SQMASPACE}

For our streaming version of $\QMASPACE$, the previous setup of $\SQCMASPACE$ does not suffice, as (e.g.) single-qubit gates do not suffice to generate arbitrary quantum proofs.
Hence, we define $\SQCMASPACE$ with an exponentially long proof that is swapped into the proof register bit-by-bit.

\begin{restatable}[Streaming-$\QMASPACE$ ($\SQMASPACE(p,q,r)$)]{definition}{definitionSQMASPACE}\label{def:SQMASPACE}
  A promise problem $A=(\ayes,\allowbreak\ano)$ is in $\SQMASPACE(p,q,r)$ for polynomially-bounded functions $p,q,r$, if there exist thresholds $\alpha(n),\beta(n)$ satisfying $\alpha(n)-\beta(n)\geq 2^{-r(n)}$, and a $q(n)$-space uniform family of quantum circuits $\{Q_n\}$
  with properties as follows.
  $Q_n$ takes as input a string $x\in\Sigma^n$, a $2^{p(n)}$-qubit proof $\ket{\psi}$ in register $\Y$, a $q(n)$-bit ancilla register $\X$ initialized to $\ket{0}^{\otimes q(n)}$, and is of form (see \Cref{fig:sqma})
  \begin{equation}
    Q_n= \prod_{i=m}^1 \bigl((V_i)_\X\cdot\SWAP_{\X_1,\Y_i}\bigr)\cdot(V_0)_\X.
  \end{equation}
  Then,
  \begin{itemize}
    \item (Completeness) If $x\in\ayes$: $\exists\ket\psi\in\Y:\bra{0^q}\bra{\psi}\bigl(Q_n^\dagger\ketbra11_{\X_1}Q_n\bigr)\ket{0^q}\ket{\psi}\ge \alpha(n)$.
    \item (Soundness) If $x\in\ano$: $\forall\ket\psi\in\Y:\bra{0^q}\bra{\psi}\bigl(Q_n^\dagger\ketbra11_{\X_1}Q_n\bigr)\ket{0^q}\ket{\psi}\le \beta(n)$.
  \end{itemize}
  As in \Cref{def:QCMASPACE}, we do not allow $Q_n$ to alter the contents of its input register (to avoid using said register as additional ancilla space).
\end{restatable}
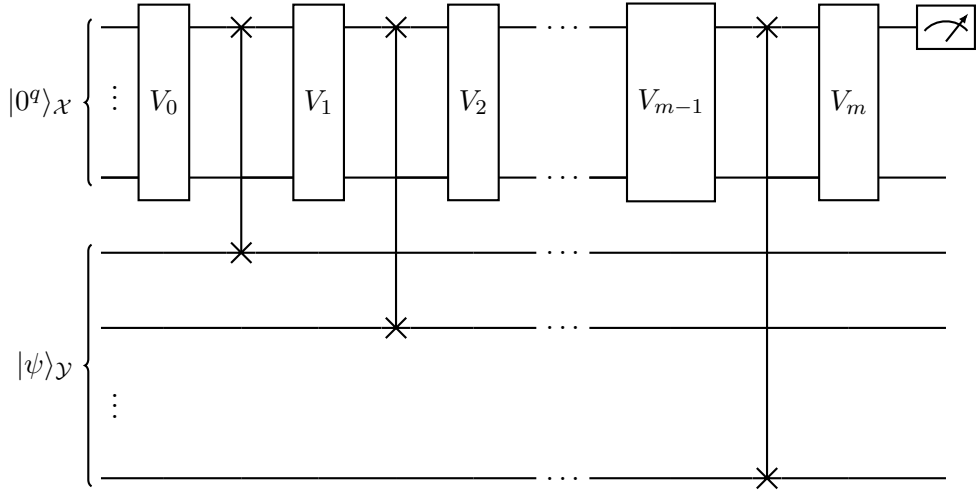
\begin{figure}[t]
  \centering
  \begin{quantikz}[row sep={1cm,between origins}]
    \lstick[2]{$\ket{0^q}_\X$\makebox(0,0)[lb]{\hspace{5mm}$\vdots$}}
    &\gate[wires=2]{V_0} & \swap{2} & \gate[wires=2]{V_1} & \swap{3} & \gate[wires=2]{V_2} & \qw \ \ldots\ & \gate[wires=2]{V_{m-1}} & \swap{4} & \gate[wires=2]{V_m} & \meter{}\\[1cm]
    &\qw & \qw & \qw & \qw &\qw & \qw \ \ldots\ & \qw & \qw & \qw & \qw \\
    \lstick[3]{$\ket{\psi}_\Y$\makebox(0,0)[lb]{\vspace{-6mm}\hspace{5mm}$\vdots$}}
    &\qw & \targX{} & \qw & \qw & \qw & \qw \ \ldots\ & \qw & \qw & \qw & \qw   \\
    &\qw & \qw & \qw & \targX{} & \qw& \qw\ \ldots\ & \qw & \qw & \qw & \qw \\[1cm]
    &\qw & \qw & \qw & \qw & \qw & \qw\ \ldots\ & \qw & \targX{} & \qw & \qw
  \end{quantikz}
  \caption{\SQMASPACE circuit. The gates from $\X$ to $\Y$ are SWAP gates. For simplicity, we do not depict the input register, whose contents without loss of generality are not altered by $Q_n$.}\label{fig:sqma}
\end{figure}
\noindent As in $\SQCMASPACE$, we allow an inverse exponential promise gap.
Later in \Cref{thm:amplify-sqmaspace}, we show that weak error reduction holds for \SQMASPACE.
Note that proof streaming is modeled here by swapping the proof bits one by one into a designated ancilla register of the verifier.


Since we will reduce $\SQMASPACE$ to $\QMAEXP$, we also define the latter now, as well as its complete problem.
\QMAEXP is defined the same as \QMA but with an exponential-time uniform circuit family.

\begin{definition}[$\QMAEXP(p,q)$]\label{def:QMAEXP}
    A promise problem $A=(\ayes,\ano)$ is in $\QMAEXP(p(n),q(n))$ for polynomially bounded functions $p, q$ if there exists an exponential time uniform family of quantum circuits $\set{Q_n}$ with properties as follows. $Q_n$ takes as input a string $x\in\Sigma^n$, a quantum proof $\ket{\psi_1}\in (\C^2)^{\otimes 2^{p(n)}}$, and ancilla qubits in state $\ket{0}^{\otimes 2^{q(n)}}$. We say $Q_n$ accepts $(x,\ket\psi)$ with probability $p_{\textup{acc}}$ if on input $(x,\ket\psi)$, measuring $Q_n$'s dedicated output wire in the standard basis yields $1$ with probability $p_{\textup{acc}}$. Then:
    \begin{itemize}
    \item (Completeness) If $x\in\ayes$, there exists a $\ket{\psi}$ such that $Q_n$ accepts $(x,\ket\psi)$ with probability at least $2/3$.
    \item (Soundness) If $x\in\ano$, for all $\ket\psi$, $Q_n$ accepts $(x,\ket\psi)$ with probability at most $1/3$.
    \end{itemize}
\end{definition}
\noindent
We write $\QMAEXP$ to mean $\QMAEXP(\poly,\poly)$.
The 1D translationally invariant Hamiltonian problem is complete for $\QMAEXP$ \cite{GI09}.
Here, ``1D translationally invariant'' means the \emph{same} local constraint $H_{i,i+1}$ is repeated on all consective qubits $i$ on the chain.
Formally:
\begin{definition}[\IDTIH \cite{GI09}]\label{def:1D-TIH}
    Fix a constant $d$. The input is the length of the chain $N$ encoded in binary, a single Hamiltonian constraint $H\in\Herm(\C^{d}\otimes \C^{d})$ specified with a constant number of bits, and polynomials $\alpha,\beta$ such that $\beta-\alpha\ge 1/\poly(N)$.
    The full Hamiltonian is thus $H^{(N)} := \sum_{i=1}^{N-1} H_{i,i+1}$ (i.e. $H$ acts on each pair of qudits on a line). Decide:
    \begin{itemize}
        \item (YES case) If $\lmin(H^{(N)})\le \alpha(N)$, output YES.
        \item (NO case)  If $\lmin(H^{(N)})\ge \beta(N)$, output NO.
    \end{itemize}
\end{definition}
Note that, crucially, the length of the chain $N$ is exponential in the encoding size of the input.

\subsection{Multi-prover interactive proofs}\label{sscn:MIP}

\begin{definition}[$\MIP(t(n),u(n),v(n),p(n),r(n),c(n),s(n))$ (introduced in \cite{BGKW88}, as stated in \cite{FV14})]\label{def:mip}
    A promise problem $A=(\ayes,\ano)$ is in $\MIP(t,p,r,c,s)$ if there exist polynomial $t$ and polynomially bounded functions $u$ and $v$, and a classical verifier $V$ using $\poly(n)$ time, $u(n)$ space, $v(n)$ bits of randomness, and interacting with $p$ non-communicating provers via $r$ rounds of interaction, where each round consists of $t(n)$ bits of communication between verifier and provers, and where $n=\abs{x}$ is the size of input $x$, such that
    \begin{itemize}
        \item If $x\in \ayes$, then there exists a strategy for the provers that is accepted by the verifier with probability at least $c$.
        \item If $x\in \ano$, any strategy of the provers is accepted by the verifier with probability at most $s$.
    \end{itemize}
\end{definition}
\noindent

\begin{theorem}[\cite{Babai1990,FL92}]\label{thm:mipnexp}
    For any polynomial $r$,
    \begin{equation}
        \MIP(\poly, \poly, \poly, \poly, \poly, 2/3, 1/3)=\MIP(\poly,\poly, \poly, 2,1,1,2^{-r})=\NEXP.
    \end{equation}
\end{theorem}
\noindent
Next, $\NP\subseteq \MIP(\log,\log, \log, 2,1,1,1-1/\poly(n))$---this follows from the standard $3$-SAT multi-prover protocol, in which the verifier asks prover $1$ for the assignment $(x,y,z)$ to a random clause, and prover $2$ for the assignment to one of the bits $x$, $y$, or $z$ uniformly at random. By applying the PCP theorem, one immediately strengthens this inclusion to the case of \emph{constant} soundness for $\MIP$. More generally, $\MA=\MIP(\log,\poly, \poly, 2,1,1,1-1/\poly(n))$. (The forward inclusion follows by applying the Cook-Levin theorem to the MA verifier to obtain a $3$-SAT formula $\phi$, and using \MIP's randomness to choose $\poly(n)$ random bits for $\phi$ (this is why the space requirement increases from $\log$ to $\poly$ for NP versus MA, respectively) and then run the interactive protocol for NP above. The reverse inclusion simply has the MA verifier receive a brute force concatenation of all answers the provers would send to the verifier possible questions, which will have $\poly(n)$ total length.)

\subsection{Probabilistically checkable proofs}

\begin{definition}[{$\PCP[r(n),q(n)]$} \cite{Arora1998}]\label{def:pcp}
    A language $L$ is in $\PCP[r(n),q(n)]$ if there exists verifier Turing machine $M$ that behaves as follows:
    \begin{enumerate}
        \item $M$ receives input $x$, a proof $y\in\bin^{*}$, and a random string $z\in\bin^{r(n)}$ on separate tapes.
        \item $M$ computes indices $i_1,\dots,i_{q(n)}$ without accessing $y$ (it may access $z$) in polynomial time.
        \item $M$ copies proof bits $y_{i_1},\dots,y_{i_{q(n)}}$ to its work tape.
        \item $M$ accepts or rejects in polynomial time without accessing $y$.
    \end{enumerate}
    $M$ must also satisfy the following conditions for all $x\in\bin^n$:
    \begin{itemize}
        \item If $x\in L$, $\exists y:\Pr_z[M(x,y,z)=1]=1$, where $z\in\bin^{r(n)}$ is chosen uniformly at random.
        \item If $x\notin L$, $\forall y:\Pr_z[M(x,y,z)=1]\le 1/2$.
    \end{itemize}
\end{definition}

\begin{theorem}[PCP Theorem \cite{Arora1998a}]
    $\NP = \PCP[O(\log(n)), O(1)]$.
\end{theorem}

\begin{theorem}[\cite{Babai1990}]\label{thm:pcp-nexp}
    $\NEXP=\PCP[O(\poly(n)),O(\poly(n))]$.
\end{theorem}

\subsection{Ground State Connectivity Problem}

The \emph{ground state connectivity} problem ($\GSCON$) introduced by Gharibian and Sikora \cite{Gharibian2018} intuitively asks the following question:
Given a Hamiltonian $H$ and ground states $\ket\psi$ and $\ket\phi$, does there exist a sequence of local gates that maps $\ket\psi$ to $\ket\phi$, such that all intermediate states have low energy with respect to $H$?
Formally, it is defined as follows.

\begin{definition}[$\GSCON(H,k,\eta_1,\eta_2,\eta_3,\eta_4,\Delta,l,m,U_\psi,U_\phi)$ \cite{Gharibian2018}]\label{def:gscon}
    \hfill
\begin{enumerate}[label={},leftmargin=1.5em]
    \item Input:
    \begin{itemize}
        \item A $k$-local Hamiltonian $H\in\hermp{\B^{\otimes n}}$, where $\B:=\C^2$.
        \item $\eta_1,\eta_2,\eta_3,\eta_4,\Delta\in\R$ and integer $m\ge 0$, such that $\eta_2-\eta_1\ge\Delta$ and $\eta_4-\eta_3\ge \Delta$.
        \item Polynomial size quantum circuits $U_\phi,U_\psi$ generating \enquote{starting} and \enquote{target} states $\ket\phi$ and $\ket\psi$ (on input $\ket{0^n}$), respectively, satisfying $\braketb\psi H\le \eta_1$ and $\braketb\phi H\le \eta_1$.
    \end{itemize}
    \item Output:
    \begin{enumerate}[align=left,labelwidth=\widthof{YES:},leftmargin=4em,labelsep*=1.1em]
        \item[YES:] There exists a sequence of $l$-local unitaries $U_1,\dots,U_m$ such that:
        \begin{enumerate}[label=(\alph*)]
            \item (Intermediate states remain in low energy space) For all $i\in[m]$ and intermediate states $\ket{\psi_i} := U_i\cdots U_1\ket\psi$, it holds that $\braketb{\psi_i} H \le \eta_1$, and
            \item (Final state is close to target state) $\enorm{\ket{\psi_m}-\ket\phi} \le \eta_3$.
        \end{enumerate}
        \item[NO:] For all $l$-local sequences of unitaries $U_1,\dots,U_m$, either:
        \begin{enumerate}[label=(\alph*)]
            \item (Intermediate state obtains high energy) There exists $i\in[m]$ and an intermediate state $\ket{\psi_i}$ such that $\braketb{\psi_i}H\ge \eta_2$, or
            \item (Final state far from target state) $\enorm{\ket{\psi_m}-\ket\phi} \ge\eta_4$.
        \end{enumerate}
    \end{enumerate}
\end{enumerate}
\end{definition}

We assume $U_\psi$ and $U_\phi$ to be given as sequences of gates from a universal gate set.
The numeric parameters are specified with rational entries using $O(\poly(n))$ bits of precision.
Note that $\ket\psi$ and $\ket\phi$ are not necessarily required to be ground states.

This definition is quite flexible as it allows all parameters to be specified.
For $2$-local unitaries, a $5$-local Hamiltonian, polynomial $m$ and $\Delta$, $\GSCON$ is $\QCMA$-complete.

\begin{theorem}[\cite{Gharibian2018}]\label{thm:gscon-qcma}
    There exists a polynomial $p$ such that $\GSCON$ is $\QCMA$-complete for $m=O(p(n))$, $\Delta = \Theta(1/m^5)$, $l=2$, and $k\ge5$, where $n$ denotes the number of qubits $H$ acts on.
\end{theorem}
\noindent
Choosing different parameters leads to $\PSPACE$-completeness.

\begin{theorem}[\cite{Gharibian2018}]\label{thm:gscon-pspace}
    $\GSCON$ is $\PSPACE$-complete for $m=2^n$, $\Delta=2^{-(2n+4)}$, $l=1$, $k=3$, where $n$ denotes the number of qubits $H$ acts on.
\end{theorem}
\noindent
This result is a consequence of the fact that $\STCONN$ is $\PSPACE$-complete \cite{Gopalan2006}.

\begin{definition}[$\STCONN{}$]
    Given a $3\mhyphen\CNF$ formula $\phi$ and solutions $x,y\in\bin^n$ to $\phi$, does there exist a sequence of strings $x_1,\dots,x_m$, such that
    \begin{itemize}
        \item $x_1=x$, and $x_m=y$, and
        \item for all $i\in[m]$, the Hamming distance between $x_i$ and $x_{i+1}$ is at most $1$, and
        \item for all $i\in[m]$, $x_i$ is a solution to $\phi$?
    \end{itemize}
\end{definition}
\noindent
Observe the similarity between $\STCONN$ and $\GSCON$:
$\phi$ corresponds to $H$, $x$ to $\ket\psi$, $y$ to $\ket\phi$, and $x_i$ to $\ket{\psi_i}$.
We are interested in the power of $\GSCON$ with $l=2$ and $m=2^{\poly(n)}$, and denote this class $\GSCONexp{}$.

\begin{definition}[$\GSCONexp{}$]\label{def:gscon-exp}
    $\GSCONexp{}$ is the union over all $GSCON(\,\dotsi)$, where $l=2$, $m = O(2^{p(n)})$ and $\Delta=\Omega(2^{-p(n)})$ for some polynomial $p$.
\end{definition}
\noindent

\section{Quantum analogues of \texorpdfstring{\NPSPACE}{NPSPACE}}\label{scn:quantum-npspace}

In this paper, we investigate the power of \GSCONexp{}.
To show the containment $\GSCON{}\in\QCMA$ for polynomial $m$ and $\Delta$, \cite{Gharibian2018} construct a $\QCMA$-verifier that receives classical approximations of the unitaries $U_1,\dots,U_m$ as proof.
That technique no longer works for $\GSCONexp{}$.
The paths through the Hamiltonian's low energy space can be of exponential length and therefore intermediate states can no longer be expressed succinctly.
For that reason, we conjecture that $\GSCONexp{}$ may not even by contained in $\PSPACE{}$.

It holds that \GSCONexp is \PSPACE-hard (see \Cref{sec:gscon-exp-pspace}) (under polynomial-time reductions) and contained in \NEXP{}.
The containment is trivial, since a $\NEXP$-verifier can choose the unitary sequence nondeterministically.
Also, hardness is not implied by \Cref{thm:gscon-pspace}, since $2$-local unitaries can flip two bits at the same time.

A natural question is whether there exists a better upper bound than $\NEXP$.
We do not know the answer to that question.
One intuitive candidate would be a quantum analogue of $\NPSPACE{}$, which we model as a variant of $\QCMA$ with an exponentially long proof and polynomially many qubits.
However, we argue in \Cref{sec:savitch} that any such construction equals $\NEXP$.

Lastly, we show in \Cref{sec:application-gscon}, that for sufficiently large $m=2^{\poly(n)}$, we can map $\ket\psi$ to $\ket\phi$ while remaining close to the span of $\ket\psi$ and $\ket\phi$.
We can make the distance to the span arbitrarily small by increasing $m$.
We conclude that $\GSCONexp{}$ does not have any NO-instances with $m=2^{\poly(n)}$ and $\Delta = 2^{-o(n)}$ for sufficiently large $n$.


\subsection{Streaming-\texorpdfstring{$\QCMASPACE$}{QCMASPACE} vs. \texorpdfstring{\NEXP}{NEXP} and Savitch's theorem}\label{sec:savitch}

We now show our no-go theorem for a quantum analogue of Savitch's theorem.
For this, recall the definition of $\SQCMASPACE$:

\definitionQCMASPACE*
\noindent
We observe that a $\PCP$ verifier (see \Cref{def:pcp}) can easily be simulated in $\SQCMASPACE$.
\theoremqcmaspace*

\begin{proof}
    The containment $\SQCMASPACE \subseteq \NEXP$ is trivial.
    To show $\NEXP\subseteq\SQCMASPACE{}$, we use the fact that $\NEXP=\PCP[\poly,\poly]$ (see \Cref{thm:pcp-nexp}).
    We construct a $\SQCMASPACE{}$ verifier $Q$ that simulates a $\PCP[r,q]$ verifier.
    Let $T = 2^{\poly(n)}$ be an upper bound on the largest proof bit index accessed by $M$.
    \begin{enumerate}
        \item $Q$ generates the random string $z\in\bin^{r(n)}$ by constructing a state $\ket{+}^{\otimes r(n)}$ and then measuring it in standard basis.
        This can be done with the usual deferred measurement technique (e.g., \cite{Nielsen}) since $r(n)$ is polynomial (it is nontrivial to simulate an exponential number of measurements).
        \item $Q$ simulates the index computation of $M$ and stores the indices $i_1,\dots,i_{q(n)}$ in ancilla space.
        \item For $j=1,\dots,T$, $Q$ applies $W_j$ to an ancilla $\ket0_c$, which maps it to $\ket{y_j}_c$.
        If $j=i_k$ for some $k$, copy $y_{j}$ to a fresh ancilla.
        Afterwards, $W_j$ is applied again to reset the ancilla $c$ back to $\ket{0}_c$.
        \item Simulate $M$ with the stored proof bits to accept or reject.
    \end{enumerate}
    Since the measured string $z\in\bin^{r(n)}$ is distributed uniformly at random, we have \[ \Pr[Q^y_n\text{ accepts } \ket{x}] = \Pr_z[M(x,y,z)=1]. \]
    Note that mapping $\ket{y_j}_c$ back to $\ket0_c$ is no issue because the circuit is entirely classical after generating the random string.
\end{proof}
\noindent
We remark that above theorem really only uses quantum computation to generate randomness.
It follows that the soundness gap in $\SQCMASPACE(\poly,\poly,\poly)$ can be reduced to a constant.
\begin{corollary}\label{cor:qcmaspace-error-reduction}
    $\SQCMASPACE(\poly,\poly, 1) = \SQCMASPACE(\poly,\poly,\poly)$.
\end{corollary}
\noindent
We leave as open problem the question whether direction error reduction is possible, i.e. without the detour via \PCP and \NEXP.

An alternative interpretation of \Cref{thm:qcmaspace} is that Savitch's theorem \cite{Savitch1970}, which implies $\PSPACE = \NPSPACE$, has likely no quantum analogue because the space-bounded variant of $\BQP$, denoted $\BQUPSPACE$, equals $\PSPACE$, as shown by Fefferman and Lin~\cite{Fefferman2018}.
$\BQUPSPACE{}$ is defined as $\BQP$ with polynomial-space uniformly generated quantum circuits (i.e. like $\SQCMASPACE$ without a proof).
Watrous~\cite{Watrous1999,Watrous2003,Watrous2008} gave an earlier definition of $\BQPSPACE{}$ based on quantum Turing machines.
The main difference between these definitions is that the quantum Turing machines may perform an exponential number of intermediate measurements, whereas that is not possible with a $\BQUPSPACE{}$ verifier (the subscript `U' indicates the verifier may only perform unitary operations).
The usual deferred measurement approach does not work because it requires fresh ancillae for each measurement.
Both definitions nevertheless equal $\PSPACE{}$.
Recently, Fefferman and Remscrim~\cite{FR21} proved that even $\QMASPACE = \PSPACE$, where the $\QMASPACE{}$ verifier is an exponentially long quantum circuit that receives a polynomially-sized proof and is allowed to perform an unrestricted number of intermediate measurements.
Hence, a variant of $\SQCMASPACE$ with exponentially long circuit, but only polynomially sized proof, would also equal $\PSPACE{}$.

\subsection{Streaming-\texorpdfstring{$\QMASPACE$}{QMASPACE} vs. \texorpdfstring{\QMAEXP}{QMA\_EXP}}\label{ssec:savitch}

Next, we characterize the power of $\SQMASPACE$, which recall is defined as:

\definitionSQMASPACE*
\noindent
We first show that \SQMASPACE can be amplified to a constant promise gap\footnote{In fact, the completeness/soundness thresholds can be made to be exponentially close to 1/0 in the length of the proof.}.

\begin{theorem}\label{thm:amplify-sqmaspace}
  $\SQMASPACE(p,q,r) \subseteq \SQMASPACE(p',q',1)$, where $q'(n) = q(n) + O(r(n))$ and $p'(n) = O(p(n) + r(n) + \log q(n))$.
\end{theorem}
\begin{proof}
  Let $Q_n$ be an $\SQMASPACE(p,q,r)$ verifier with $m$ gates, which is given some input $x$.
  Our goal is to amplify its completeness/soundness via repetition.
  The main challenge is that there are not enough ancillas for the usual deferred measurement approach as we require $2^q$ repetitions with $O(q)$ space.
  Our solution is to use the proof itself as ancillas for each repetition and count the number of accepting computations.

  We construct a verifier $V$ with the following registers:
  Counter $\Cacc$ to count accepting computations, $\Cinit$ to count correct initializations of ancillas, and $\X$ to simulate $Q_n$. Formally, $V$ is given as (explanation in words below):
  \begin{equation}
    V := \prod_{i=R}^1 \left[\CUp_{\X_1=1,\Cacc}\cdot\prod_{j=m}^1 \left(V_iS_{\X_1,\Y_{i(q+m)+q+j}}\right) \cdot V_0 \cdot \CUp_{\X=0^q,\Cinit} \cdot \prod_{j=q}^1\left(S_{\X_1,\X_j}S_{\X_1,\Y_{i(q+m)+j}}\right)\right]
  \end{equation}
  Here, $\CUp$ indicates a controlled increment operation (e.g. $\CUp_{\X_1=1,\Cacc}$ increments the counter in $\Cacc$ if $\X_1=1$) and $S$ a SWAP gate.
  The indices are upside down to indicate that the leftmost term has the largest index.
  In words, $V$ operates in $R$ rounds as follows:
  \begin{enumerate}
    \item (Refresh ancilla qubits) Swap $q$ streamed proof bits into $\X$.
    \item (Condition on ancilla being properly initialized) Increment $\Cinit$ conditioned on $\X=0^q$, i.e. apply unitary
    \begin{equation}
        \ketbra{0^q}{0^q}_{\X}\otimes U^+_{\Cinit} + (I-\ketbra{0^q}{0^q})_\X\otimes I_{\Cinit}.
    \end{equation}
    \item (Run the actual verification) Simulate $Q_n$.
    \item Increment $\Cacc$ conditioned on $\X_1 = 1$ (output qubit of $Q_n$).
  \end{enumerate}
  Letting $\delta := (\alpha-\beta)/2$ and $t := (\alpha-\delta) R$, define the accepting projector as
  \begin{equation}
    \Piacc := \sum_{i\ge t}\ketbra{i}{i}_{\Cacc}\otimes \ketbra{R}{R}_{\Cinit}.
  \end{equation}
  Note that $\lceil\log_2(R)\rceil$ bits are required for each counter.
  Hence, $V$ acts on $q'= q+2\lceil\log_2(R)\rceil$ ancillas.\\

  \noindent\emph{Completeness:}
  In the YES-case, we assume the honest prover sends $\ket\phi=(\ket{0^q}\ket\psi)^{\otimes R}$, where $\ket\psi$ is a proof accepted by $Q_n$ with probability $\ge \alpha$.
  We can view $\Cacc$ as the sum of independent random variables corresponding to the outcomes of each round.
  Then by Hoeffding's inequality,
  \begin{equation}
    \Pr[-\Cacc +\alpha R  \ge \delta R] \le e^{-2\delta^2R} \le 1/3
  \end{equation}
  for $R\ge \ln(1/3)/2\delta^2$ with $\log(R) = O(r)$.\\

  \noindent \emph{Soundness:}
  In the NO-case, we first argue that we can assume the proof's ancilla bits are initialized to 0.
  We split the proof register $\Y$ into $A$ for the ancillas, $B$ for the actual proof, and write $\ket{\phi} = \sum_{z\in\bin^{qR}} a_z\ket{z}_A\ket{\phi_z}_B$.
  Then $\Piacc V\ket\phi = a_{0^{qR}}\ket{0^{qR}}_A\ket{\phi_{0^{qR}}}_B$ as only $A=\ket{0^{qR}}$ causes $\Cinit=r$.
  Hence, we can assume $A=\ket{0^{qR}}$ and get the POVM
  \begin{align}
    \Pacc &= \braketa{\bra{0^{q'+qR}}_{A,\Cinit,\Cacc}}{\bigl(V^\dagger \Piacc V\bigr)}\\
    &= \sum_{z\in\bin^{R},\abs{z}\ge t} \bigotimes_{i=1}^{R} P_{z_i},
  \end{align}
  where $\abs{z}$ denotes the Hamming weight and $P_1,P_0$ the accepting/rejecting POVM of $Q_n$.
  Since $P_0$ and $P_1 = I-P_0$ commute, $\Pacc$ has an eigenbasis $\{\ket{\phi_{i_1},\dots,\phi_{i_R}}\}_{i_1,\dots,i_R}$, where $\{\ket{\phi_i}\}_i$ is an eigenbasis of $P_0$ (and $P_1$).
  Therefore, $\Pacc$ has an eigenvector with maximum eigenvalue of the form $\ket{\psi_1}\otimes\cdots\otimes\ket{\psi_R}$, where $Q_n$ accepts each $\ket{\psi_i}$ with probability $\le\beta$ by assumption.
  Since we have projected $A$ onto all-zeroes, the optimal proof is of form $\bigotimes_{i=1}^R\ket{0^q}\ket{\psi_i}$, and we can apply Hoeffding's inequality again:
  \begin{equation}
    \Pr[\Cacc - \beta R \ge \delta R ] \le e^{-2\delta R} \le 1/3.
  \end{equation}
\end{proof}

\begin{corollary}\label{cor:SQCMASPACE=QMAEXP}
  $\SQMASPACE(\poly,\poly,1)=\QMAEXP$.
\end{corollary}
\begin{proof}
  Containment is trivial.
  The other direction follows by amplifying the verification circuit $V$ of a \IDTIH instance (1D translationally invariant Hamiltonian, \Cref{def:1D-TIH}), which is complete for \QMAEXP with a promise gap of $1/\exp$~\cite{GI13}.
  We use the standard verifier from Kitaev's ``quantum Cook-Levin theorem'' \cite{KSV02} that picks a random index $i$, and then measures using a (potentially rescaled) $2$-local constraint $H$ as POVM (recall in TIH all terms on the chain are identical).
  $V$ then has a promise gap of $1/\exp(n)$, where $n$ is the input size.
  Note it is straightforward to implement $V$ as an $\SQMASPACE(\poly,\poly,\poly)$ circuit because we only need to measure one Hamiltonian term (selected at random).
  Hence, the required qubits can be swapped into the ancilla space when streamed and measured at the end of the computation.
  \Cref{thm:amplify-sqmaspace} completes the proof.
\end{proof}

\begin{corollary}\label{cor:SQCMASPACE=QMA}
  $\SQMASPACE(\log,\log,1)=\QMA$.
\end{corollary}
\begin{proof}
    The proof is analogous to \Cref{cor:SQCMASPACE=QMAEXP}, but instead using QMA-completeness of the (non-translationally invariant) local Hamiltonian problem~\cite{KSV02}.
\end{proof}

\section{Universal Quantum Path Following Lemma}\label{scn:discretizing}

In this section, we give a general construction for simulating \emph{any} Lipschitz continuous path $f$ on the unit hypersphere via a sequence of $2$-local gates. We begin with definitions.
\begin{definition}[Paths and Lipschitz continuity]\label{def:lipschitz-path}
    For any $d\geq 2$, consider unit hypersphere $S^{d-1}:=\set{\ket{\psi}\in\C^d\mid \enorm{\ket{\psi}}=1}$. A \emph{path} is any function $f:[0,1]\rightarrow S^{d-1}$. We say $f$ is \emph{$K$-Lipschitz continuous} if for all $a,b\in[0,1]$, $\enorm{f(a)-f(b)} \le K\abs{a-b}$.
\end{definition}
\noindent We measure the distance between two paths by the metric $d(f, g) := \max_{t\in[0,1]} \enorm{f(t) - g(t)}$ for $\enorm{\cdot}$ the Euclidean norm.

The main result of this section is the following.
\theoremPathInterpolation*
\noindent
In words, any Lipschitz continuous path $f$ on the unit hypersphere can be approximately simulated to any desired precision $\epsilon$ by applying a sequence of $M$ $2$-local unitaries (see \Cref{fig:sphere} for an illustration).
The main idea behind the proof is to first discretize $f$ sufficiently finely, and then to locally simulate $f$ between each consecutive pair of discrete points via a sequence of ``small rotations''.
Here, by ``small rotations'', we mean unitaries close to identity, which can also be written as $U=e^{iH}$ with small $\snorm{H}$.
We can write $H = \sum_j \alpha_j H_j$ in Pauli basis (i.e. each $H_j$ is a tensor product of $I,X,Y,Z$) with small $\alpha_j$.
Applying a result due to Suzuki (\Cref{lem:custom-suzuki}), we have $e^{iH} \approx \prod_j e^{\alpha_jH_j}$.
Next, a construction of Clinton, Bausch, and Cubitt \cite{CBC21} is used to decompose the $e^{\alpha_jH_j}$ into $2$-local unitaries $U_1,\dots U_m$ ($m$ independent of $\alpha_j$) such that $U_k\to I$ as $\alpha_j\to 0$ (\Cref{sec:decomposition}).%
\footnote{Decompositions of arbitrary unitaries into $2$-local gates are well known (e.g., \cite{Nielsen}), but to the best of our knowledge, they do not provide bounds on the distance from $I$.}
\Cref{sec:general-decomposition} combines the Suzuki and CBC decompositions and \Cref{sec:approx-path} applies that result to complete the proof of \Cref{l:path-interpolation}.

\subsection{Technical Lemmas}

We state a collection of technical results used in the proof of \Cref{l:path-interpolation}.

\subsubsection{Norms}

\begin{lemma}[{\cite[Equation 2.2.5]{Golub1996}}] \label{lem:12-norm}
    For all $v\in\C^d$,
        $\enorm{v} \le \norm{v}_1 \le \sqrt{d} \enorm{v}$.
\end{lemma}

\begin{lemma}\label{lem:enorm-diff}
    For all $\ket\psi,\ket\phi\in\C^d$,
        $\enorm{\ket\psi-\ket\phi} = \sqrt{2-2\real(\braket\phi\psi)}$.
\end{lemma}
\noindent
For operators $M\in\L(\C^d)$, the corresponding \emph{operator norm}, usually called the \emph{spectral norm}, is defined as
\begin{equation}
    \snorm{M} := \max_{v\in\C^d} \frac{\enorm{Mv}}{\enorm{v}}.
\end{equation}
It holds that $\snorm{M} = \sqrt{\lmax(M^\dagger M)}$.
For $M\succcurlyeq0$, we have $\snorm{M} =\lmax(M)$, where $\lmax(M)$ denotes the largest eigenvalue of $M$.
We write $M = O(f(d))$ if $\snorm{M} = O(f(d))$ for some function $f$.

The \emph{Frobenius norm} is defined as
\begin{equation}
    \fnorm{M} := \sqrt{\Tr(M^\dagger M)}=\sqrt{\sum_{i=1}^d\sum_{j=1}^d \abs{m_{ij}}^2},
\end{equation}
where $m_{ij}$ denote the entries of $M$.
Note that the Frobenius norm is the same as the Euclidean norm of $M$ viewed as a $d^2$-dimensional vector.

\begin{lemma}[{\cite[Equation 2.3.7]{Golub1996}}] \label{lem:frobenius-norm-inequality}
    $ \snorm{M} \le \fnorm{M} \le \sqrt{d} \snorm{M} $
\end{lemma}
\noindent
We define the \emph{trace norm} for operators $M\in\L(\C^d)$ as
$ \trnorm{M} := \Tr(\abs M) = \Tr\sqrt{M^\dagger M}, $
where $\abs \cdot$ and $\sqrt{\cdot}$ are applied as operator functions.

\begin{lemma}\label{lem:trace-dist}
    Let $\ket\psi,\ket\phi\in\C^d$.
    Then,
        $\trnorm{\ketbra\psi\psi - \ketbra\phi\phi} = 2\sqrt{1-\abs{\braket\psi\phi}^2}$.
\end{lemma}

\subsubsection{Unitaries and Hamiltonians}

\begin{lemma}\label{lem:1-e^ix}
  For all $x\in\R$, it holds that $\abs{1-e^{ix}} \le \abs{x}$.
\end{lemma}

\begin{lemma}\label{lem:e^iH-I}
  Let $H\in\herm(\C^d)$. For $U=e^{iH}$, $\snorm{U - I} \le \snorm{H}$.
\end{lemma}
\begin{proof}
  Let $H = \sum_{j=1}^d\lambda_j\ketbra{\psi_j}{\psi_j}$ be the spectral decomposition of $H$.
  Then, the spectral decomposition of $U-I$ is given by
  \begin{align}
    I - e^{iH} &= \sum_{j=1}^d\ketbra{\psi_j}{\psi_j} - \sum_{j=1}^de^{i\lambda_j}\ketbra{\psi_j}{\psi_j} = \sum_{j=1}^d\left(1-e^{i\lambda_j}\right)\ketbra{\psi_j}{\psi_j}.
  \end{align}
  Therefore,
  \begin{equation}
    \snorm{I-U} = \max_j{\abs*{1-e^{i\lambda_j}}} \le \max_j{\abs{\lambda_j}} = \snorm{H},
  \end{equation}
  where the inequality follows from \Cref{lem:1-e^ix}.
\end{proof}

\begin{lemma}\label{lem:approx-circuit}
  Let $U = U_m\cdots U_1$ and $V=V_m\cdots V_1$ be unitary matrices. For a submultiplicative norm $\norm{\cdot}$, it holds that
  $\norm{U - V} \le \sum_{i=1}^m \norm{U_i - V_i}$.
\end{lemma}

\begin{lemma}[Suzuki]\label{lem:custom-suzuki}
    Let $H=\sum_{j=1}^m H_j$ be a sum of Hermitian operators such that $ \sum_{j=1}^m\snorm{H_j} \le t \le 1$ and $s\in\N$. Then
    \begin{equation}\label{eq:custom-suzuki}
         e^{iH} = \left(\prod_{j=1}^m e^{iH_j/s}\right)^s + O\left(\frac{t^2}s\right).
    \end{equation}
\end{lemma}
\begin{proof}
Follows directly from \cite[Theorem 3]{Suzuki1976}.
\end{proof}

This lemma can also be used for Hamiltonian simulation, for if $H$ is a $k$-local Hamiltonian, then the $e^{iH_j/n}$ terms are $k$-local gates.
Therefore, we can simulate the evolution $e^{iH}$ with only local gates.
The well-known Lie-Trotter product formula follows directly from Equation \eqref{eq:custom-suzuki}
\begin{equation}
    e^{iH_1+iH_2} = \lim_{n\to\infty} \left(e^{iH_1/n}e^{iH_2/n}\right)^n.
\end{equation}

\subsection{Decomposition of Pauli Interactions}\label{sec:decomposition}

Next, we show how to decompose operators $e^{itH}$ for $H\in\{I,X,Y,Z\}^{\otimes n}$ into $2$-local gates of the form $e^{it_jH_j}$, such that the total evolution time $\sum_j|t_j|$ is bounded by $O(t^{1/n})$.
This result is originally due to Clinton, Bausch, and Cubitt \cite{CBC21}. For clarity, here we give an alternative construction of their decomposition (still using Lemmas from~\cite{CBC21}), with a simpler analysis of pulse time bounds, and with an exponential improvement in the number of gates required --- see \Cref{rem:CBC} below for details.
The main insight we use in the decomposition is as follows.

\begin{lemma}[{\cite[Lemmas 7 and 9]{CBC21}}]\label{lem:depth4}
    Let $U = e^{itH}$ for a Hamiltonian $H = \frac{1}{2i}[h_1,h_2]$, where $h_1$ and $h_2$ anti-commute and square to identity.
    For $0\le t\le \pi/2$, there exist $t_1,t_2\in \R$ with
    \begin{equation}
        \abs{t_1} + \abs{t_2} \le \sqrt{2t},
    \end{equation}
    and
    \begin{equation}
        U  = e^{it_1h_1}e^{it_2h_2}e^{it_2h_1}e^{it_1h_2}.
    \end{equation}
\end{lemma}
\noindent
We can also use \Cref{lem:depth4} with negative $t \ge -\pi/2$ by applying the lemma to $-t$ and then using the inverse of the resulting decomposition ($(e^{itH})^\dagger = e^{-itH}$).

To apply this to Pauli interactions, we observe that $X,Y,Z$ pairwise anti-commute, square to identity, and
\begin{equation}\label{eq:pauli-commute}
    [X,Y] = 2iZ,\qquad [X,Z] = 2iY,\qquad [Y,Z] = 2iX.
\end{equation}
Hence, we can apply \Cref{lem:depth4} to decompose $e^{itH}$ for $n=2^k+1$ and
\begin{equation}
    H=P_1\otimes\cdots\otimes P_n\in\{I,X,Y,Z\}^{\otimes n}
\end{equation}
into two $2^{k-1}+1$ local evolutions as follows.
Let $j = 2^{n-1}+1$, assume $P_j = Z$, and set
\begin{align}
    h_1 &= P_1 \otimes \cdots\otimes P_{j-1} \otimes X_j \otimes I_{j+1,\dots,n},\\
    h_2 &= I_{1,\dots,j-1}\otimes Y_j\otimes P_{j+1}\otimes\cdots\otimes P_n.
\end{align}
Then,
\begin{align}
    [h_1,h_2] &= P_1 \otimes \cdots\otimes P_{j-1} \otimes XY_j \otimes P_{j+1}\otimes\cdots\otimes P_n \\
    &\qquad- P_1 \otimes \cdots\otimes P_{j-1} \otimes YX_j \otimes P_{j+1}\otimes\cdots\otimes P_n\\
    &= P_1 \otimes \cdots\otimes P_{j-1} \otimes [X,Y]_j \otimes P_{j+1}\otimes\cdots\otimes P_n,\\
    &= 2i H.
\end{align}
The cases $P_j=X$ or $P_j=Y$ are analogous due to Equation \eqref{eq:pauli-commute}.
The decomposition is depicted in \Cref{fig:decomp} (tensor products between the Pauli operators are omitted for conciseness).

\begin{figure}[htb]
        \begin{center}
            \includegraphics[width=\linewidth]{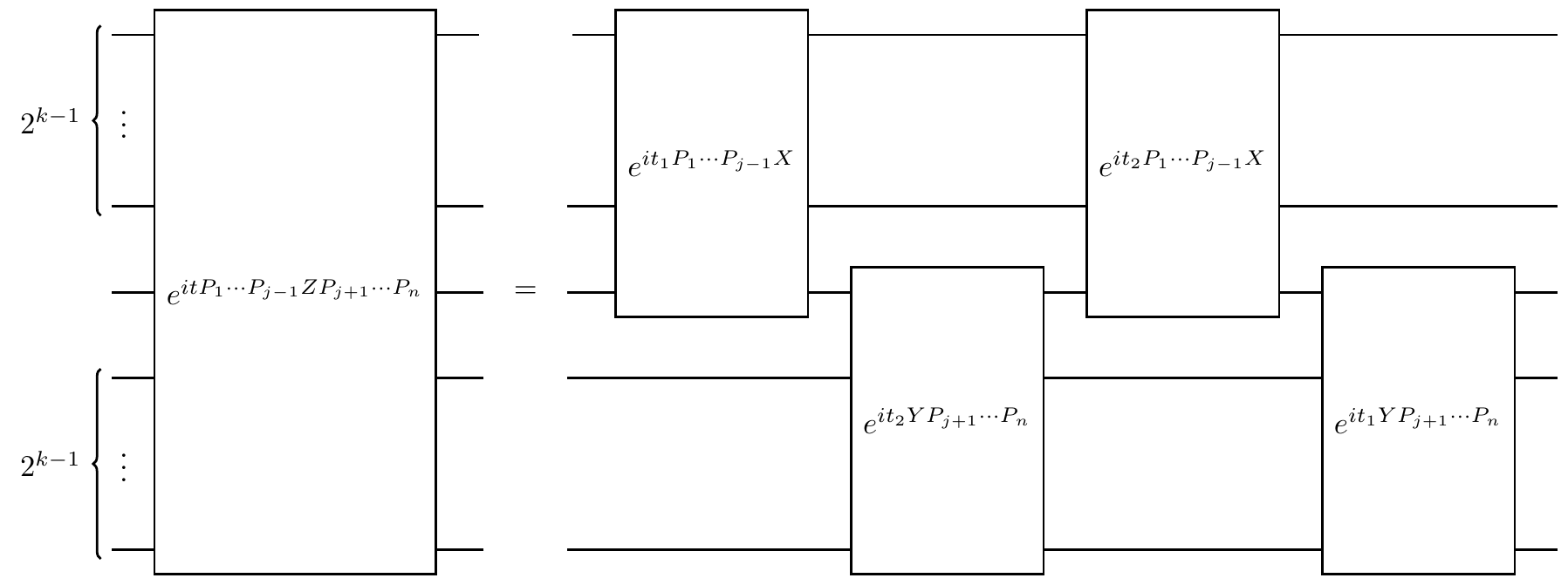}
        \end{center}
    \caption{Decomposition of Pauli interactions.}\label{fig:decomp}
\end{figure}

For $n\ne2^k+1$, we cannot always choose the split $j$ to be exactly in the middle, but the resulting interactions will still be at most $(2^{k-1}+1)$-local, provided $n\le2^{k}+1$.
Applying this decomposition recursively we show:

\begin{lemma}\label{lem:exact-decomposition}
  Let $H \in\{I,X,Y,Z\}^{\otimes n}$ with $n\in(2^{k-1}+1,\,2^k+1]$, and $t\in\R$ with $8\abs t^{1/2^{k}}\le\pi/2$.
  There exists a decomposition of $e^{it H} = \Pi_{j=1}^m e^{i t_j H_j}$, where the $H_j$ are $2$-local Pauli matrices, $m\le 4^k = O(n^2)$, and $\sum_{i=1}^m |t_i|=O(n^2\abs t^{1/2^k}) = O(n^2|t|^{1/2n})$.
\end{lemma}
\noindent For clarity, the last equality of the claim, $O(n^2\abs t^{1/2^k}) = O(n^2|t|^{1/2n})$, holds since $8\abs t^{2^{-k}}\le\pi/2$ implies $\abs{t}<1$.
\begin{proof}[Proof of \Cref{lem:exact-decomposition}]
    We construct the $2$-local $H_1,\dots,H_m$ by applying \Cref{lem:depth4} recursively as outlined above.
    $e^{it H} = \Pi_{j=1}^m e^{i t_j H_j}$ follows from the correctness of \Cref{lem:depth4}.
    After each recursion, we have interactions that are at most $(2^{k-1}+1)$-local.
    Hence, after $k$ recursions, only $2$-local interactions remain and we are done.
    Since each recursion increases the number of interactions by a factor of $4$, we have $m\le 4^k = O(n^2)$.

    By \Cref{lem:depth4}, a recursion constructs new interactions, each with a pulse time of at most
    \begin{equation}
        \abs{t_1}+\abs{t_2} \le \sqrt{2t}.
    \end{equation}
    This gives us a recurrence for an upper bound on the individual pulse times after $r$ recursions:
    \begin{align}
        T(r) = \begin{cases}
            \abs t, & \text{if } r = 0\\
            \sqrt{2T(r-1)}, & \text{if } r > 0
        \end{cases}.
    \end{align}
    Using induction on $r$, we show that
    \begin{equation}
        T(r) = 2^{1-2^{-r}}\abs t^{2^{-r}}.
    \end{equation}
    For $r=0$, we have $T(0)=\abs t$.
    For $r>0$, we have
    \begin{align}
        T(r) &= \sqrt{2T(r-1)}
        = \left(2\cdot 2^{1-2^{-r+1}}\abs t^{2^{-r+1}}\right)^{1/2}
        = 2^{1/2}\cdot 2^{1/2-2^{-r}}\abs t^{2^{-r}}
        = 2^{1-2^{-r}}\abs t^{2^{-r}}.
    \end{align}
    Hence, the individual pulse times after $k$ recursions are bounded by $T(k) \le 8\abs t^{2^{-k}}$.
    The total pulse time is then bounded by $m T(k) = O(n^2\abs t^{1/2n})$.
\end{proof}
\begin{remark}\label{rem:CBC}
    Our decomposition only uses polynomially many gates, whereas it appears to us that the construction given in \cite{CBC21} uses exponentially many.
    This might be of interest for physical applications.
    We also only require their Lemmas~7 and~9, without having to use the more complex Lemmas 8 and 10.
    Their decomposition has a total pulse time of $O(\abs t^{1/n})$.
\end{remark}

\subsection{General Decomposition}\label{sec:general-decomposition}

Next, we show how to use \Cref{lem:exact-decomposition} to decompose general unitaries of small norm.

\begin{lemma}\label{lem:pauli-based-decomp}
  Let $U = e^{iH}$ for Hermitian $H \in \herm(\C^d),\,d=2^n$ with $\snorm{H} =: \epsilon<(\pi/16)^{2n}$.
  There exists an approximate decomposition
  $ U=U_m\cdots U_1 + O\left(d^2\epsilon^2\right)$
  into $m\le2^{O(n)}$ $2$-local unitaries, such that
  \begin{equation}\label{lem:pauli-based-decomp:bound}
    \sum_{j=1}^m \snorm{I-U_j} = O\left(n^2d^2\epsilon^{1/2n}\right)
  \end{equation}
\end{lemma}
\begin{proof}
  We write $H$ in the Pauli basis:
  \begin{equation}
    H = \sum_{j=1}^{d^2}\alpha_j P_j,
  \end{equation}
  where $\alpha_j\in\R$ and $P_j \in\{I,X,Y,Z\}^{\otimes n}$ for all $j\in[d^2]$.
  Then, by \Cref{lem:frobenius-norm-inequality}
  \begin{equation}\label{eq:fnorm-snorm}
    \sqrt{d}\snorm{H} \ge \fnorm{H} =\sqrt{\Tr(H^\dagger H)}=\sqrt{d\sum_{j=1}^{d^2} \alpha^2_j}.
  \end{equation}
  Therefore, $|\alpha_j| \le \epsilon$ for all $j\in[d^2]$.
  It holds that
  \begin{align}
    \sum_{j=1}^{d^2} \snorm{\alpha_jP_j} &= \sum_{j=1}^{d^2}\abs{\alpha_j}
    \le d\sqrt{\sum_{j=1}^{d^2} \alpha^2_j}
    \le d\snorm{H}
    \le 1,
  \end{align}
  where the first inequality follows by \Cref{lem:frobenius-norm-inequality} (note $\sum_{j=1}^{d^2}\abs{\alpha_j}$ is a sum of $d^2$ terms), and the second by Equation \eqref{eq:fnorm-snorm}. Applying \Cref{lem:custom-suzuki} with $s=1$, we have
  \begin{equation}
    U = e^{iH} = \prod_{j=1}^{d^2} e^{i\alpha_jP_j} + O\left(d^2\epsilon^2\right).
  \end{equation}
  \Cref{lem:exact-decomposition} allows us to decompose each term $e^{i\alpha_jP_j}$ into $m_j=O(n^2)$ $2$-local  unitaries
  \begin{equation}
    e^{i\alpha_jP_j} = \prod_{k=1}^{m_j} e^{it_{j,k}H_{j,k}}
  \end{equation}
  with an evolution time of $\sum_{k=1}^{m_j}\abs{t_{j,k}} = O(n^2\abs{\alpha_j}^{1/2n})$.
  We get the complete decomposition
  \begin{equation}
    U = e^{iH} = \prod_{j=1}^{d^2} \prod_{k=1}^{m_j} e^{it_{j,k}H_{j,k}} + O(d^2\epsilon^2),
  \end{equation}
  with a total evolution time of
   \begin{equation}
    O\left(n^2\sum_{j=1}^{d^2} \abs{\alpha_j}^{1/2n} \right) = O\left(n^2d^2 \epsilon^{1/2n}\right) .
    \end{equation}
  Equation \eqref{lem:pauli-based-decomp:bound} follows from \Cref{lem:e^iH-I} as
  \begin{equation}
    \snorm{I - e^{it_{j,k}H_{j,k}}} \le \snorm{t_{j,k}H_{j,k}} = t_{j,k}.
  \end{equation}
\end{proof}

\noindent We remark, that we usually choose $\epsilon \ll (\pi/16)^{2n}$ in order to make Equation \eqref{lem:pauli-based-decomp:bound} small.
Furthermore, the above decomposition is approximate.
It appears to be an open question whether a similar result is achievable with an exact decomposition.

\subsection{Approximating paths via local unitaries}\label{sec:approx-path}

We are almost ready to prove \ref{l:path-interpolation}; the last ingredient we require is the following lemma. It uses the above decomposition to (approximately) map between two close vectors $\ket\psi$ and $\ket\phi$ while bounding the distance of intermediate states from $\ket\psi$.

\begin{lemma}\label{lem:map-close-vectors}
    Let $\ket\psi,\ket\phi\in\C^d$ be unit vectors with $d=2^n$.
    Let $\enorm{\ket\psi-\ket\phi} \le \epsilon < (\pi/16)^{2n}$.
    There exists a sequence of $2$-local unitaries $U = U_m\cdots U_1$ with $m\le2^{O(n)}$, such that
    \begin{enumerate}[label=(\arabic*)]
        \item $\enorm{\ket\phi-U\ket\psi} = O(d^2\epsilon^2)$, and
        \item for all $i\in[m]$, $\enorm{\ket\psi - U_i\cdots U_1\ket\psi} = O(n^2d^2\epsilon^{1/2n})$.
    \end{enumerate}
\end{lemma}
\begin{proof}
    Let $\theta = \cos^{-1}(\real(\braket\psi\phi))$ be the angle between $\ket\psi$ and $\ket\phi$.
    After a suitable change of basis $W$, we have
    \begin{align}
        W\ket\psi &= \ket0,\\
        W\ket\phi &= \cos(\theta)\ket0+\sin(\theta)\ket1.
    \end{align}
    Hence, we only need to apply the rotation matrix (extended to $d$ dimensions)
\begin{equation}\label{eq:bigrot}
\sbox0{$\begin{matrix}\cos\theta & -\sin\theta\\\sin\theta & \cos\theta\end{matrix}$}
R(\theta)=\left(
\begin{array}{c|c}
\usebox{0}&\makebox[\wd0]{\Large $0$}\\
\hline
  \vphantom{\usebox{0}}\makebox[\wd0]{\Large $0$}&\makebox[\wd0]{\Large $I$}
\end{array}
\right)
\end{equation}
    to map from $W\ket\psi$ to $W\ket\phi$.
    Let $V = W^\dagger R(\theta) W$.
    $V$ has the same eigenvalues as $R(\theta)$, namely $e^{i\theta},e^{-i\theta},1$.
    Hence, $V = e^{iH}$ for $\snorm{H} = \abs\theta$.

    To apply \Cref{lem:pauli-based-decomp}, we need to bound $\theta$.
    We have for $\abs\theta<1$,
    \begin{align}
        \enorm{\ket\psi-\ket\phi} &= \sqrt{2-2\real(\braket\phi\psi)}
        = \sqrt{2-2\cos\theta}
        \ge \sqrt{\theta^2 - \theta^4/12}
        \ge \abs\theta/2,
    \end{align}
    where the first equality follows from \Cref{lem:enorm-diff}, and the first inequality via Taylor expansion. Hence, $\abs\theta \le 2\epsilon$.
    Let $U = U_m\cdots U_1$ from \Cref{lem:pauli-based-decomp}.
    Properties (1) and (2) of the claim follow from \Cref{lem:approx-circuit}.
\end{proof}

We now restate and prove~\ref{l:path-interpolation}.
\theoremPathInterpolation*
\begin{proof}
    The idea is to first discretize $f$ into a sufficiently large number $N'+1$ of points, and subsequently apply \Cref{lem:map-close-vectors} to simulate $f$ along each interval $[i/N',(i+1)/N']$. To begin, \Cref{def:lipschitz-path} says that for any $i\in\set{0,\ldots, N'-1}$,
    \begin{equation}
        \enorm{f(i/N')-f((i+1)/N')}\leq K/N'=: \delta.
    \end{equation}
    We will shortly set $N'$ as needed, but it will be sufficiently large so that $\delta<(\pi/16)^{2n}$.
    Thus, to the $i$th interval $[i/N',(i+1)/N']$ we can apply \Cref{lem:map-close-vectors} to obtain a sequence of $2$-local unitaries $U_i=U_{i,m_i}\cdots U_{i,1}$ with $m_i\leq 2^{O(n)}$ such that for all $i$,
    \begin{eqnarray}
      \enorm{f((i+1)/N')-U_i\cdot f(i/N')}&=&O(d^2\delta^2)\text{, and}\label{eqn:e1}\\
      \forall j\in \set{1\ldots, m_i}\text{, }\enorm{f(i/N')-U_{i,j}\cdots U_{i,1} f(i/N')}&=&O(n^2d^2\delta^{1/2n})\label{eqn:e2}.
    \end{eqnarray}
    Letting $U=U_{N'}\cdots U_1$, we have $M=\sum_{i=0}^{N'-1}m_{i}\leq N'2^{O(n)}$. It remains to choose $N'$ so as to bound the point-wise error to $\epsilon$ as in Equation \eqref{eqn:pointerror}.

    The analysis proceeds in two stages. First, Equation \eqref{eqn:e1} and \Cref{lem:approx-circuit} imply that for any $t\in\set{0,\ldots, M}$ such that $t/M=i/N'$ for some $i\in\set{0,\ldots, N'-1}$ (these are the $N'+1$ points obtained after our first round of discretizing $f$),
    \begin{equation}
        \enorm{\ket{\psi_t}-f(t/M)}\in O(N'd^2\delta^2).
    \end{equation}
    Second, we ``subdivided'' each interval $[i/N',(i+1)/N']$ into intermediate points $U_{i,j}\cdots U_{i,1} f(i/N')$ via \Cref{lem:map-close-vectors}. Equation \eqref{eqn:e2} says each of these intermediate points is at most $O(n^2d^2\delta^{1/2n})$-far from the start point of that interval, $f(i/N')$. Combining the two errors, we have that for any $t\in\set{0,\ldots, M}$,
    \begin{equation}\label{eqn:finalerror}
         \enorm{\ket{\psi_t}-f(t/M)}\in O(N'd^2\delta^2+n^2d^2\delta^{1/2n}).
    \end{equation}
    (Note $N'$ does not appear on the $n^2d^2\delta^{1/2n}$ term, as this error does not accumulate from one interval $[i/N',(i+1)/N']$ to the next.) To bound this by $\epsilon>0$, it suffices to set
    \begin{align}
        N' \in \begin{cases}
            \Omega \left(K\left(\frac{n^2d^2}{\epsilon}\right)^{2n}\right) & \text{if }0<K\leq 1\\
            \Omega\left(\left(\frac{K^2n^2d^2}{\epsilon}\right)^{2n}\right) & \text{if }K > 1,
            \end{cases}
    \end{align}
    and thus $M\leq N'2^{O(n)}\in O(K(\frac{n^2d^2}{\epsilon})^{2n})$ if $0<K\leq 1$ and $M\in O(2^{O(n)}(\frac{K^2n^2d^2}{\epsilon})^{2n})$ if $K>1$.
\end{proof}

\section{Applying Quantum Path Following to \texorpdfstring{$\GSCONexp{}$}{GSCON\_exp}}\label{sec:application-gscon}

In the previous section, we show how to implement general paths with $2$-local unitaries.
We apply this approach to construct unitary sequences for $\GSCON$ instances.
Note, however, that these sequences have exponential length.
Suppose we are given a $\GSCON{}$ instance, where $l=2$, and for simplicity the starting state $\ket\psi$ and the target state $\ket\phi$ are orthogonal ground states of $H$ (as opposed to just low energy states).
To determine whether we have a YES-instance, we need to check whether there exists a sequence of $2$-local unitaries that maps $\ket\psi$ to $\ket\phi$ but keeps the energy of intermediate states low.
Certainly, states in the span of $\ket\psi$ and $\ket\phi$ are also ground states.
Hence, we can apply \Cref{l:path-interpolation} to the path $f(t) := \cos(t\pi/2)\ket\psi + \sin(t\pi/2)\ket\phi$ to obtain a suitable unitary sequence.
This approach also works if $\ket\psi$ and $\ket\phi$ are \emph{not} orthonormal ground states, which is proven in the theorem below.

\theoremtraversal*
\begin{proof}
    The idea is to first rotate $\ket\psi$ onto a ground state $\ket{\mu}$ and then use the same method to rotate $\ket{\mu}$ to $\ket\phi$.
    This will leave all intermediate states at an energy of at most $\eta$ (in practice it might exceed $\eta$ since the construction is approximate).
    Let $\ket{\lambda_1},\dots,\ket{\lambda_d}$ be an orthonormal eigenbasis of $H$.
    We assume $\ket\mu = \ket{\lambda_1}$ is a ground state.
    We can write
    \[ \ket{\psi} = \cos(\theta)\ket{\mu} + \sin(\theta)\ket{\nu}, \]
    where $\ket{\nu} \in \Span\set{\ket{\lambda_2},\dots,\ket{\lambda_d}}$.

    We argue that increasing the amplitude on $\ket\mu$ and decreasing the amplitude on $\ket\nu$ cannot increase the energy.
    Let $\ket{\psi'} = \cos(\theta')\ket{\mu} + \sin(\theta')\ket{\nu}$ with $\abs{\cos\theta'} > \abs{\cos{\theta}}$.
    Then, for $\lambda = \lmin(H)$,
    \begin{align}
        \braketb{\psi} H &= \cos^2(\theta)\lambda + \sin^2(\theta) \braketb\nu H\\
        &= \cos^2(\theta)\lambda + (1-\cos^2(\theta)) \braketb\nu H\\
        &< \cos^2(\theta')\lambda + (1-\cos^2(\theta')) \braketb\nu H\\
        &= \braketb{\psi'}H,
    \end{align}
    where the inequality follows since $\lambda\le \braketb\nu H $. Define path $f(t) := \cos((1-t)\theta)\ket\mu + \sin((1-t)\theta)\ket\nu$.
    The path $f$ is Lipschitz continuous for $K=\pi/2$, and by the above argument, $f(t)^\dagger H f(t) \le \eta$ (assuming $\theta \in [0,\pi/2]$ without loss of generality).
    Define a path $g$ from $\ket\mu$ to $\ket\phi$ in the same way and concatenate both to path $h$.
    The proof is complete by applying \Cref{l:path-interpolation} to $h$ with $\epsilon = \Delta$.
\end{proof}
\noindent
Note that the theorem easily generalizes to any $0\preccurlyeq H \preccurlyeq 2^{\poly(n)}I$.

We can now partially answer an open question of Gharibian and Sikora~\cite{Gharibian2018}, which asked about the complexity of $\GSCON$ with exponential $m$ and $l=2$. Specifically, Reference \cite{Gharibian2018} showed that for exponential $m$, $l=1$, and inverse exponential gap $\Delta$, GSCON is PSPACE-complete, and conjectured that the analogous $l=2$ case is NEXP-complete. Here, on the one hand, we will later show (\Cref{thm:GSCON1local} below) that the $l=1$ PSPACE-completeness result holds even for \emph{constant} gap $\Delta$. However, in strong contrast, in the $l=2$ case \Cref{thm:traversal} says that for any \emph{subexponential} $\Delta$, GSCON is poly-time decidable:

\begin{corollary}\label{cor:GSCONP}
    $\GSCON$ with $m=2^{\poly(n)}$, $l=2$, and subexponential $\Delta$ does not have NO-instances for sufficiently large $n$.
\end{corollary}
\noindent
In the following, we investigate how the above result relates to the $1$-local case (\Cref{sec:1-local-relation}), the classical case (\Cref{sec:locality-in-st-conn}), and the Traversal Lemma from~\cite{Gharibian2018} (\Cref{sec:relation-traversal-lemma}).

\subsection{Relation to the 1-Local Case}\label{sec:1-local-relation}

We have seen that any $\GSCONexp{}$ instance with $2$-local gates ($l=2$) becomes either a YES-instance or an invalid instance if we make $m$ sufficiently large.
This raises the question of what happens in the $1$-local case (i.e. $l=1$).
We show that $\GSCON$ with $l=1$ and $m=\infty$ (i.e. gate sequences may be arbitrarily long) is still $\PSPACE$-complete.
Therefore, arbitrarily long sequences does not change the hardness of $1$-local \GSCON{}, which is also $\PSPACE$-complete for bounded $m$.

\begin{theorem}\label{thm:GSCON1local}
  \GSCON{} is \PSPACE-complete for $l=1,\,k=3,\,\eta_1=\eta_3=0,\,\eta_2=1/8,\eta_4=1/2,\Delta=1/8$ and unbounded $m$.
\end{theorem}
\begin{proof}
    \Cref{lem:gscon-unbounded-hard,lem:gscon-unbounded-contain} below show hardness and containment.
\end{proof}

\subsubsection{Hardness}

\begin{lemma}[{\cite[Lemma 6.2]{Gharibian2018}}]\label{lem:l=1-orig}
  \GSCON{} is \PSPACE-hard for $k = 3,\, \eta_1 = \eta_3 = 0,\, \eta_2 = 2^{-(2n+4)},\, \eta_4 = 1/4,\,\Delta = 2^{-(2n+4)},\, l = 1$, and $m = 2^n$, where $n$ denotes the number of qubits $H$ acts on.
\end{lemma}

We can strengthen this statement as follows.

\begin{lemma}\label{lem:gscon-unbounded-hard}
  \GSCON{} is \PSPACE-hard for $l=1,\,k=3,\,\eta_1=\eta_3=0,\,\eta_2=1/8,\eta_4=1/2,\Delta=1/8$ and unbounded $m$.
\end{lemma}
\begin{proof}
  As in \cite{Gharibian2018}, we reduce $L\in\PSPACE$ to $\STCONN(1)$ (see \Cref{def:st-conn}).
  Let $(\phi,x,y)$ be an $\STCONN{}$ instance.
  We set $\ket\psi=\ket x$, $\ket\phi = \ket y$, $H=\sum_i H_i$, where $H_i = \ketbra{z_i}{z_i}\otimes I$ for the unsatisfying assignment $z_i$ of clause $c_i$ in $\phi$.
  Completeness follows from \Cref{lem:l=1-orig}.

  To prove soundness, let $(\phi,x,y)$ be a no-instance and consider a sequence of $1$-local unitaries $U_1,\dots, U_m$.
  We show that $(H,\ket x,\ket y)$ is a no instance for $\GSCON$.
  All intermediate states are of the form
  \begin{equation}
    \ket{\psi_i} = \bigotimes_{j=1}^n\left(\alpha^0_{i,j}\ket0+\alpha^1_{i,j}\ket1\right).
  \end{equation}
  For $i=1,\dots,m$ define $x_i = x_{i,1}\dotsm x_{i,n}$ with
  \begin{equation}
    x_{i,j} = \begin{cases}
    0 & \text{if } \abs{\alpha^0_{i,j}}^2 \ge \frac12\\
    1 & \text{else}
  \end{cases}.
  \end{equation}
  Hence, $\abs{\braket{x_i}{\psi_i}}\ge 2^{-n/2}$.
  Since $\ket{\psi_i}$ and $\ket{\psi_{i+1}}$ differ in only one qubit, we have $h(x_i,x_{i+1})\le1$.
  Assume for contradiction that $\enorm{\ket{\psi_m} - \ket\phi}\le\eta_4$.
  Then,
  \begin{equation}
    \enorm{\ket{\psi_m} - \ket{y}} = \sqrt{2-2\real(\braket{\psi_m}{y})} \le \eta_4.
  \end{equation}
  Hence,
   \begin{equation}
    \abs{\braket{\psi_m}y} \ge \real(\braket{\psi_m}{y}) \ge 1-\frac{\eta_4^2}2 > \sqrt{\frac12}.
   \end{equation}
  Therefore, we have $\abs{\alpha_{m,j}^{y_j}}>\sqrt{1/2}$ for all $j\in[n]$ and thus $y=x_m$.
  Hence, there exists an $x_i$ such that $x_i$ does not satisfy a clause $c_j$ of $\phi$.
  It follows that $\braketb{x_i}{H_j}=1$.
  \Wlog{}, $H_j$ operates on qubits $1,2,3$.
  Then, $x_{i,1}x_{i,2}x_{i,3} = z_j$ and $\braketb{\psi_i}{H_j} = \abs{\alpha_{m,1}^{x_{i,1}} \alpha_{m,2}^{x_{i,2}} \alpha_{m,3}^{x_{i,3}}}^2\ge 1/8 = \eta_2$.
  Thus, $(H,\ket\psi,\ket\phi)$ is a no-instance.
\end{proof}

\subsubsection{Containment}

\begin{lemma}[{\cite[Lemma 6.3]{Gharibian2018}}]\label{lem:l=1-containment}
  For all nonnegative constants $c_1$ and $c_2$, $\GSCON$ with $l=1$ is in $\PSPACE$ for $m\le2^{n^{c_1}}$ and $\Delta\ge1/2^{n^{c_2}}$, where $n$ denotes the number of qubits $H$ acts on.
\end{lemma}
\noindent
\GSCON{} is also contained in \PSPACE{} for unbounded $m$.

\begin{lemma}\label{lem:gscon-unbounded-contain}
  $\GSCON \in\PSPACE$ with $l=1$ for $\Delta\ge2^{-\poly(n)}$ and unbounded $m$.
\end{lemma}
\begin{proof}
  We use the same $\NPSPACE$ algorithm given in \cite[Algorithm 2]{Gharibian2018} used to prove \Cref{lem:l=1-containment} with some sufficiently large $m'$, except here we must account for the fact that $m$ can be \emph{unbounded}, unlike in \cite{Gharibian2018}.
  That $m$ is unbounded does not affect the soundness analysis, which follows from \Cref{lem:l=1-containment}.
  For completeness, however, we show that the sequence $U_1,\dots,U_m$ can be shortened to some $U_1',\dots,U'_{m'}$ with $m'\le2^{\poly(n)}$.
  For this, note that every intermediate state is of the form
  \begin{equation}
    \ket{\psi_i} = \left(\bigotimes_{j=1}^n V_{i,j}\right)\ket\psi =: V_i\ket\psi.
  \end{equation}
  For all $i$, there exists at most one $j$ for which $V_{i,j}\ne V_{i+1,j}$.
  The \NPSPACE\ algorithm, however, stores approximations $V_{i,j}'$ in $2^{p(n)}$ bits such that for all $j\in[n]$, $\snorm{V_{i,j}'-V_{i,j}}\le\epsilon/n$ for some $\epsilon\ge2^{-\poly(n)}$ (i.e. the $V'$ are taken from an $\epsilon$-net; see \cite[Lemma 3.1]{Gharibian2018}).
  Therefore, $\snorm{V_i-V'_i} \le \epsilon$, where $V_i' = \bigotimes_{i=1}^nV'_{i,j}$.
  Thus, $\enorm{\ket{\psi_i}-\ket{\psi'_i}} \le \epsilon$ and $\braketb{\psi_i'}H \le \eta_1+2\epsilon$ if $\snorm{H}\le1$.
  There are $m':=2^{n\cdot p(n)}$ possibilities for each $V_i$.
  Hence, the sequence can be shortened to $V_1'',\dots,V_{m'}''$.
  We can assume that $V'_i$ and $V'_{i+1}$ only differ on a single qubit, which allows us to construct the sequence $U_1',\dots,U'_{m'}$ of $1$-local unitaries with $V''_i = U'_iV''_{i-1}$.
  Choosing sufficiently small $\epsilon$, the \NPSPACE algorithm accepts when given the sequence $U_1',\dots,U'_{m'}$.
\end{proof}

\subsection{Locality in \texorpdfstring{\STCONN}{S,T-CONN}}\label{sec:locality-in-st-conn}

We have shown that $\GSCON$ with $m=\infty$ becomes trivial for $l=2$, but remains \PSPACE-complete for $l=1$.
Does a similar result hold classically, i.e. for \STCONN?
We define $l$-local \STCONN{} and show that a classical analogue of \Cref{thm:traversal} does not hold.

\begin{definition}[$\STCONN(l)$]\label{def:st-conn}
  Given a $3\mhyphen\CNF$ formula $\phi$ and solutions $x,y\in\bin^n$ to $\phi$, does there exist a sequence of strings $x_1,\dots,x_m$, such that
  \begin{itemize}
    \item $x_1=x$, and $x_m=y$, and
    \item for all $i\in[m]$, the Hamming distance between $x_i$ and $x_{i+1}$ is at most $l$, and
    \item for all $i\in[m]$, $x_i$ is a solution to $\phi$?
  \end{itemize}
\end{definition}

\begin{theorem}
  $\STCONN(l)$ is \PSPACE-complete for all $l\ge1$.
\end{theorem}
\begin{proof}
  $\STCONN(l)\subseteq\PSPACE$ follows from Savitch's theorem.
  $\STCONN(1)$ is \PSPACE-complete \cite{Gopalan2006}.
  Let $L\in\PSPACE$ and reduce $L$ to a $\STCONN(1)$-instance $(\phi,x,y)$.
  Construct a $3\mhyphen\CNF$ formula $\phi'$ from $\phi$ by creating $l-1$ copies of each variable with equality constraints with corresponding solutions $x'$ and $y'$.
  We show that $(\phi,x,y) \in \STCONN(1)$ iff $(\phi',x',y')\in\STCONN(l)$.

  If $(\phi,x,y) \in \STCONN(1)$, then there exists a sequence of solutions $x=x_1,\dots,x_m=y$ with $h(x_i,{x_{i+1}}) = 1$.
  By adding $l-1$ copies of each variable, we get a sequence of valid solutions $x'=x'_1,\dots,x'_m=y'$ with $h(x'_i,{x'_{i+1}}) = l$.
  Hence, $(\phi',x',y')\in\STCONN(l)$.

  If $(\phi',x',y') \in \STCONN(l)$, then there exists a sequence of solutions $x'=x'_1,\dots,x'_m=y'$ with $h(x'_i,{x'_{i+1}}) \le l$.
  All $l$ copies of each variable must be equal in each solution.
  Hence, between $x'_i$ and $x'_{i+1}$ all copies of exactly one variable are changed.
  Hence, we can convert $x_1',\dots,x_m'$ to solutions of $\phi$ $x=x_1,\dots,x_m=y$.
  Thus, $(\phi,x,y) \in \STCONN(1)$.
\end{proof}

\subsection{Relation to the Traversal Lemma}\label{sec:relation-traversal-lemma}

Reference~\cite{Gharibian2018} uses the \emph{Traversal Lemma} as an important tool to show that $\GSCON{}$ is $\QCMA{}$-complete.
Two states $\ket u,\ket w\in\B^{\otimes n}$ are said to be \emph{$k$-orthogonal} if for all $k$-local unitaries $U$, we have $\bra w U \ket v = 0$.
Two subspaces $S,T\subseteq\B^{\otimes n}$ are called $k$-orthogonal if any pair of vectors $\ket v\in S,\ket w\in T$ is $k$-orthogonal.

\begin{lemma}[Traversal Lemma \cite{Gharibian2018}]\label{lem:traversal-lemma}
    Let $S,T\subseteq\B^{\otimes n}$ be $k$-orthogonal subspaces.
    Let $\ket v\in S,\ket w\in T$ and consider a sequence of unitaries $U_1,\dots,U_m$ with
    \[ \enorm{\ket w-U_m\cdots U_1\ket v} \le \epsilon < 1/2. \]
    Let $\ket{v_i} := U_i\cdots U_1\ket v$ and $P:= I-\Pi_S - \Pi_T$.
    Then, there exists an $i\in[m]$ such that
    \[ \braketb{v_i}P \ge \left(\frac{1-2\epsilon}{2m}\right)^2. \]
\end{lemma}
\noindent
Reference~\cite{Gharibian2018} provides an example for which the Traversal Lemma is tight by explicitly constructing a gate sequence to map $\ket{000}$ to $\ket{111}$ with $\braketb{v_i}P\le \Delta$ for $m =O(1/\Delta^2)$.
Our \Cref{thm:traversal} constructs such a sequence in general, although it is not as tight, because $m$ is only polynomial in $\Delta^{-1}$ and exponential in $n$.

\section{Embedding streaming proofs into unentanglement}\label{scn:embedding}

In this section, we state and prove the Embedding Lemma (\Cref{l:embed}), which shows how to embed any quantum circuit verifying a streaming proof into unentanglement (more accurately, into a Sparse Separable Hamiltonian problem (\Cref{def:sepsparse}).

\begin{restatable}[Embedding lemma]{lemma}{embed}\label{l:embed}
    Let $p,q,r,m,\alpha,\beta:\R\mapsto\R$ be efficiently computable functions, where $p,q,r$ are polynomially bounded. Let $Q_n$ be a quantum circuit consisting of $m(n)$ $1$-and $2$-qubit gates, taking in (1) input $x\in\Sigma^n$, (2) a classical streaming proof ${y}\in \set{0,1}^{2^{p(n)}}$, and (3) $q(n)$ ancilla qubits in state $\ket{0}^{\otimes q(n)}$, such that $m(n)\geq 2^{p(n)}$ and $q(n)\geq p(n)$ for all sufficiently large $n$. Define thresholds $\alpha(n),\beta(n)$ satisfying $\alpha(n)-\beta(n)\geq 2^{-r(n)}$. We are promised that either:
    \begin{itemize}
        \item (YES) There exists\footnote{\label{footnote:embedding-lemma}To clarify, we are slightly abusing notation here for simplicity. Formally, \Cref{def:stream} defines $y$ as being ``part of the circuit'' $Q_n$. Section \ref{sscn:ingredients} will reflect this by using notation $Q_n(y)$ (i.e. the circuit $Q_n$ with proof gates according to $y$). In the statement of the lemma, however, we say for simplicity, in the usual wording, ``there exists a proof $y$ such that$\ldots$''.} a streaming proof $y\in\set{0,1}^{2^{p(n)}}$ such that $Q_n$ accepts $(x,y)$ with probability at least $\alpha$.
        \item (NO) For all streaming proofs ${y}\in \set{0,1}^{2^{p(n)}}$, $Q_n$ accepts $(x,{y})$ with probability at most $\beta$.
    \end{itemize}
    There exists a $\poly(n)$-time mapping from $(Q_n,x)$ to a sparse Hamiltonian $H$ on $O(q(n)+\log(m(n)))$ qubits, partition $(L,R)$ of the qubits $H$ acts on, and threshold parameters $\alpha'(n)$ and $\beta'(n)$ satisfying $\alpha(n)'-\beta(n)'\geq ((m(n)+1)2^{r(n)})^{-1}$ such that:
    \begin{itemize}
        \item If $(Q_n,x)$ is a YES case, there exists $\LR$ such that
        $
            \RL H \LR \leq \alpha'.
        $
        \item If $(Q_n,x)$ is a NO case, then for all $\LR$,
        $
            \RL H \LR \geq \beta'.
        $
    \end{itemize}
    The norm of $H$ scales as $\snorm{H}\in\poly(m(n),2^{r(n)})$.
\end{restatable}
\noindent Note that verification circuits $Q_n$ in which the classical proof $y$ is fully specified (as opposed to streamed), such as for \NP or \QCMA, are also covered by \Cref{l:embed} \emph{so long as} the ancilla space is large enough to store the entire proof $y$. (In this case, as each bit of $y$ in \Cref{l:embed} is streamed, we save it to a fresh ancilla qubit. Once the entire proof is recorded, we run the (say) \QCMA circuit $Q_n$ on $y$. Thus, there is no loss of generality in streaming the proof.)

\paragraph{Organization of section.} \Cref{sscn:ingredients} first sets up the proof ingredients. For pedagogical purposes, an effort is made to derive each of the ingredients as a response to a roadblock which arises when using a simpler construction. The full formal proof combining all ingredients is in \Cref{sscn:finalproof}.

\subsection{Proof setup and ingredients}\label{sscn:ingredients}
Let $Q_n(y)=V_m\cdots V_1$ be the quantum circuit in \Cref{l:embed} for input size $n$ given streaming proof $y$, which recall acts on registers $R_1$ (input of size $n$), $R_2$ (ancilla of $q(n)\in \poly(n)$ qubits), $R_3$ (streaming classical proof, single qubit). We write $Q_n(y)$, as opposed to simply $Q_n$, because the set $\set{V_i}$ includes both computation and proof unitaries (cf. \Cref{def:stream}), of which the latter are \emph{a priori} unknown. This is in contrast to, say, QMA verification, where $Q_n$ is fixed given just $n$.

\paragraph{Setup.} Next, we recall and slightly adapt the definitions of history state and the Feynman-Kitaev circuit-to-Hamiltonian construction~\cite{KSV02} to our setting. As is common in the study of circuit-to-Hamiltonian mappings, without loss of generality\footnote{This is because, as per \Cref{def:stream}, $R_1$ is treated as a read-only register, and thus as classical control. Since we will be designing a Hamiltonian whose terms depend on the gates in $Q_n$, the poly-time Turing machine computing the reduction can simply hardcode the gates on-the-fly conditional on the bits of $x$.} we do not need to explicitly encode the input register, $R_1$. We will, however, keep the naming conventions for $R_2,R_3,R_4$ for consistency.

We define the history state as
\begin{equation}\label{eqn:psihist}
    \ket{\psihist(y)}=\frac{1}{\sqrt{m+1}}\sum_{t=0}^mV_t\cdots V_1\ket{0\cdots 0}_{R_2}\ket{0}_{R_3}\ket{t}_{R_4},
\end{equation}
where $R_4$ denotes the clock register. As with $Q_n(y)$, we write $\ket{\psihist(y)}$ to stress the proof $y$ is now embedded into the circuit $Q_n$, rather than given directly via a separate proof register (as it would be in the setting of QMA). Also, since $m(n)\in \Omega(2^{p(n)})$ necessarily (otherwise the circuit does not have time to see each bit of proof $y$), the clock register $R_4$ is encoded in binary (as opposed to unary, as in~\cite{KSV02}) to potentially handle $p(n)$ polynomial in $n$. This makes the Hamiltonian terms defined below $\log(m(n))$-local.

Next, we define the Feynman-Kitaev circuit-to-Hamiltonian construction elements as
\begin{eqnarray}
    \hin&:=&
    \left(I-\ketbra{0\cdots 0}{0\cdots 0}\right)_{R_2}\otimes\ketbra{1}{1}_{R_3}\otimes \ketbra{0}{ 0}_{R_4}\\
    \hout&:=&\ketbra{0}{0}_{\textup{out}}\otimes
     \ketbra{m}{m}_{R_4}\\
    \hprop &:=& \sum_{t=1}^{m} H_t {\rm,~where~ }H_t{\rm ~is ~defined~ as}\label{eqn:hpropproper}\\
    H_t&:=&-V_t\otimes\ketbra{t}{ {t-1}}_{R_4} -V_t^\dagger\otimes\ketbra{{t-1}}{t}_{R_4} +I\otimes(\ketbra{ t}{t}+\ketbra{ {t-1}}{ {t-1}})_{R_4},\label{eqn:Ht}
\end{eqnarray}
 where in $\hout$, $\ketbra{0}{0}_{\textup{out}}$ projects onto the dedicated output wire of $Q_n$ (say, the first qubit of $R_2$).

 Finally, define for $1$- or $2$-qubit unitary $U$ the operator $H_t^U$ as $H_t$, but with $V_t$ replaced with $U$. Let $P\subseteq[m]$ denote the set of time steps for which $V_i=W_i$ or $V_i=W_i^\dagger$ (corresponding to Steps \ref{step:proof1} and \ref{step:proof3} of \Cref{def:stream}, respectively), i.e. in which a proof bit is written or uncomputed. We shall refer to such $V_i$ as \emph{proof gates}. Let $y^*$ denote an optimal streamed proof, i.e. accepted by $Q_n$ with the maximum probability $p^*$ possible.

\paragraph{The construction.} The basic goal of our construction is simple---design Hamiltonian $H$ so that $\ket{\psihist(y^*)}$ is its ground state. Unfortunately, since we do not know the proof gates in advance, we cannot embed the action of $Q_n(y)$ into $\hprop$. To overcome this, we weaken our optimism---we instead design $H$ to so that $\ket{\psihist(y^*)}_L\otimes\ket{\psihist(y^*)}_R$ is a low-energy state (in the sense of \Cref{def:sepsparse}) of $H$. We then use unentanglement across the two copies to logically simulate Boolean functions, allowing the history state to decide ``on-the-fly'' whether it wishes to stream proof bit $0$ or $1$ in the next round. We proceed in a sequence of attempts, each time pushing the current setup as far as possible before it breaks down, and then adding the next work-around. The full final construction is stated succinctly in \Cref{sscn:finalproof}. For clarity, throughout we assume the Hamiltonians we design act on bipartition $L$ versus $R$ of the Hilbert space.\\

\noindent \emph{Attempt 1: The foundation.} Define:
\begin{eqnarray}
    \Hint &=& (\hin)_L\otimes I_R + I_L\otimes (\hin)_R\label{eqn:hin}\\
    \Hpropt &=& \sum_{t=1}^m \Htt, \quad \text{where $\Htt$ is defined as}\label{eqn:hprop}\\
     \Htt&=&
    \begin{cases}
         (\HttI)_L\otimes (\HttX)_R  +(\HttX)_L\otimes (\HttI)_R & \text{ if }t\in P\label{eqn:gadget}\\
         (H_t)_L\otimes I_R + I_L\otimes (H_t)_R & \text{ if } t\not\in P
    \end{cases}\\
    \Houtt &=& (\hout)_L\otimes I_R + I_L\otimes (\hout)_R\label{eqn:hout}\\
    \Ht&=&\Hint+\Hpropt+\Houtt.\label{eqn:full1}
\end{eqnarray}
Completeness will hold straightforwardly for this and all subsequent iterations of the construction (see proof of \Cref{l:embed} in \Cref{sscn:finalproof}), but the intuition is as follows.
Recall our goal is for $\ket{\psihist(y^*)}_L\otimes\ket{\psihist(y^*)}_R$ to be a low-energy state of $\Ht$.
Then, the ``$+$'' in $\Hint$ and $\Houtt$ simulates a logical ``AND'', forcing {both} $L$ and $R$ registers to be correctly initialized and to accept in the final time step, respectively.
$\Htt$ is split into two cases: When $t\not\in P$, we know $V_t$ and hence can directly force both $L$ and $R$ to implement it via the ``$+$''. When $t\in P$, however, we do not know $V_t$, but only that $V_t\in\set{I,X}$ acting on $R_3$.
In this case, the ``$\otimes$'' in $\Htt$ simulates a logical ``OR'', and $\Htt$ itself simulates identity $({x}\vee \overline{y})\wedge (\overline{x}\vee{y}) \leftrightarrow x=y$ for $x,y\in\set{0,1}$; denote this construction of $\Htt$ as the \gadget\ gadget.
Intuitively, if (say) $\ket{\psihist}_L$ chooses to apply $V_t=I$ to annihilate $(H_t^I)_L$, then $\ket{\psihist}_R$ must also apply $V_t=I$ to annihilate $(H_t^I)_R$.

We now address the various shortcomings of this construction, beginning with the fact that the \gadget\ gadget itself is not sound.

\paragraph{Obstacle 1: Fooling the \gadget\ gadget.} Let $\LR$ be an arbitrary state. To force a dishonest prover to simulate an honest one, ideally, $\Htt$ with $t\in P$ should act approximately like a ``switch'', meaning
\begin{equation}
    \bra{\psi_1}_L\bra{\psi_2}_L\Htt\ket{\psi_1}_R\ket{\psi_2}_R\approx 0 \text{ iff } \left(\bra{\psi_1}\HttI\ket{\psi_1}\approx 0\text{ and }\bra{\psi_1}\HttX\ket{\psi_1}\approx 1 \text{(or vice versa)}\right) .
\end{equation}
To formally study this idea, define for $a,b\in\R$ the operator-valued function
\begin{equation}\label{eqn:G1}
            G(a,b):=a\HttX+b\HttI,
\end{equation}
so that
\begin{equation}\label{eqn:G3}
    \bra{\psi_2}\,G\left(\bra{\psi_1}\HttI\ket{\psi_1},\bra{\psi_1}\HttX\ket{\psi_1}\right)\,\ket{\psi_2}=\pto\left(\HttI\otimes\HttX +\HttX\otimes\HttI\right)\pot.
\end{equation}
The problem is that for \emph{any} $a,b\in\R$, $a\HttX+b\HttI$ has null vector $\ket{\phi}_{R_1R_2}\ket{+}_{R_3}(\ket{t-1}+\ket{t})_{R_4}$ for any $\ket{\phi}$, clearly violating the intended behavior of applying either $I$ or $X$ to $\ket{0}_{R_3}$ in step $t$. Moreover, we cannot simply force $R_3$ set to $\ket{0}$ or $\ket{1}$, as the projector onto the latter space is simply identity.\\

\noindent\emph{Attempt 2: Make it complex.} Suppose instead of using $I$ and $X$ to encode proof bit $0$ and $1$, we instead use more general unitaries $U,V\in\unitary{\C^2}$ applied to some initial state $\ket{\phi}$ (generalizing the use of $\ket{0}$ in $R_3$). In other words, an honest prover prepares $U\ket{\phi}$ to encode logical proof bit $0$, and $V\ket{\phi}$ for proof bit $1$. 
The \gadget\ gadget is thus generalized to (for $t\in P$)
\begin{equation}
    \Htt(U,V):= (\HttU)_L\otimes (\HttV)_R  +(\HttV)_L\otimes (\HttU)_R.
\end{equation}
(Observe the initial state $\ket{\phi}$ is not explicitly encoded here; this would instead be enforced by setting $R_3$ to $\ket{\phi}$ at time step $0$ of the history state. We will shortly choose $\ket{\phi}=\ket{0}$ anyway, which is enforced by our present choice of $\Hint$.)
Next, Equation \eqref{eqn:G1} is generalized to
    \begin{equation}\label{eqn:G2}
        G(a,b):=a\HttV+b\HttU.
    \end{equation}
\noindent The reason soundness breaks following Equation \eqref{eqn:G1} is captured by the following sufficient condition.
    \begin{lemma}\label{l:right}
        Let $U, V, G$ be defined as in Equation \eqref{eqn:G2}. If there exist unit vectors $\ket{\gamma_1},\ket{\gamma_2}\in\C^2$ such that
        \begin{enumerate}
            \item $V^\dagger U\ket{\gamma_1}=\ket{\gamma_1}$,\label{item:break1}
            \item $VU^\dagger\ket{\gamma_2}=\ket{\gamma_2}$, and
            \item $U\ket{\gamma_1}=\ket{\gamma_2}$,
        \end{enumerate}
        then there exists non-zero $\ket{\eta}$ acting on $R_2R_3R_4$ such that for all $a,b\in\R$, $G(a,b)\ket{\eta}=0$.
    \end{lemma}
    \begin{proof}
        Assume such $\ket{\gamma_1},\ket{\gamma_2}$ exist. Then, $U\ket{\gamma_1}=V\ket{\gamma_1}$ and $U^\dagger\ket{\gamma_2}=V^\dagger\ket{\gamma_2}$. For any $\ket{v}$ acting on $R_2$, define
        \begin{equation}
            \ket{\eta}_{R_2R_3R_4}:=\ket{v}_{R_2}\left(\ket{\gamma_1}_{R_3}\ket{t-1}_{R_4}+\ket{\gamma_2}_{R_3}\ket{t}_{R_4}\right).
        \end{equation}
        Then,
        \begin{equation}
            G(a,b)\ket{\eta}\propto (a+b)\ket{v}_{R_2}\otimes\left( -U\ket{\gamma_1}\ket{t}-U^\dagger\ket{\gamma_2}\ket{t-1}+\ket{\gamma_2}\ket{t}+\ket{\gamma_1}\ket{t-1}\right)_{R_3R_4}=0,
        \end{equation}
        where the last equality uses $U\ket{\gamma_1}=\ket{\gamma_2}$.
    \end{proof}
\noindent As a sanity check, we may apply this to Equation \eqref{eqn:G1} by setting $\ket{\gamma_1}=\ket{\gamma_2}=\ket{+}$, $U=X$ and $V=I$, for which the preconditions of Lemma~\ref{l:right} hold.

Now that we understand the bottleneck, we can work around it. Call $(\ket{\phi},U,V)$ a \emph{valid} encoding if $\bra{\phi}V^\dagger U\ket{\phi}=0$, i.e. $U$ and $V$ map $\ket{\phi}$ to orthogonal states, which may be viewed as logical $0$ and $1$. For simplicity, pick $\ket{\phi}=\ket{0}$. First, by \Cref{item:break1} of \Cref{l:right}, $V^\dagger U$ should not have a $1$-eigenvector, and second, condition $\bra{0}V^\dagger U\ket{0}=0$ and the fact that $V^\dagger U$ is unitary enforce
\begin{equation}
    V^\dagger U=\left(
                  \begin{array}{cc}
                    0 & e^{i\theta_1} \\
                    e^{i\theta_2} & 0 \\
                  \end{array}
                \right)
\end{equation}
for some $\theta_1,\theta_2\in\R$. Thus, set $\ket{\phi}=\ket{0}$, $U=I$, and $V=iX$. We now have that, restricted to $\Span(\ket{t-1},\ket{t})$ on $R_4$,
\begin{equation}\label{eqn:evals}
    \lmin(G(a,b))=a+b-\sqrt{a^2+b^2}.
\end{equation}
\begin{figure}[t]
\centering
    \includegraphics[height=30mm]{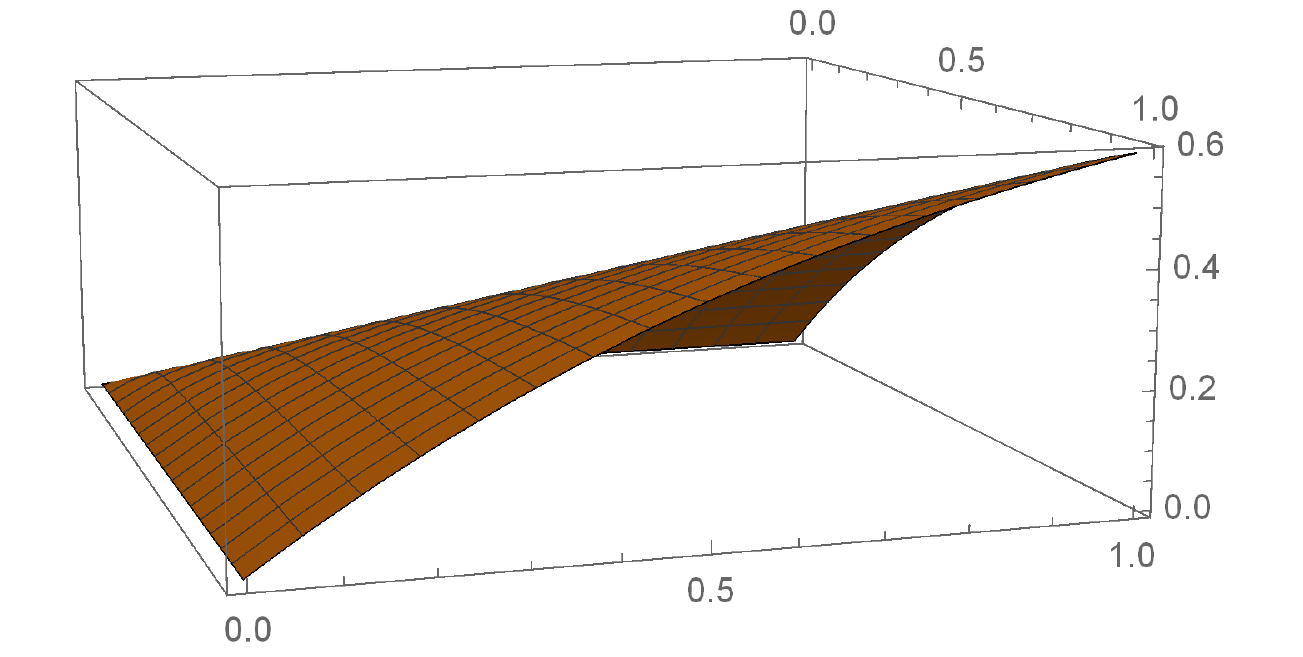}
    \caption{Above, the horizontal axes correspond to $a$ and $b$, and the vertical axis to the minimum eigenvalue of $G(a,b)$.}
    \label{fig:plot}
\end{figure}
The behavior of $\lmin(G(a,b))$ is depicted graphically in \Cref{fig:plot}, from which one immediately sees that $\lmin(G(a,b))=0$ only if at least one of $a=0$ or $b=0$. By Equation \eqref{eqn:G3}, this almost gets us what we want---in order to have a chance at annihilating the residual operator $G(a,b)$, the first proof must correctly encode either $I$ or $iX$ as the $t$th gate.\footnote{Note the use of $iX$ at a time step $t$ in \Cref{def:QCMASPACE} only produces a global phase $i$, and so does not affect the distribution obtained when measuring the output of the circuit.} We just need to check that \emph{both} $a$ and $b$ are not zero (otherwise, the second proof becomes unconstrained at time $t$).

\begin{lemma}\label{l:complement}
    For $U=I$ and $V=iX$, restricted to $\Span(\ket{t-1},\ket{t})$ on $R_4$ we have for any unit $\ket{\psi}$
    \begin{equation}
        \bra{\psi}\HttU\ket{\psi}+\bra{\psi}\HttV\ket{\psi}\geq 2-\sqrt{2}\approx0.586.
    \end{equation}
\end{lemma}
\begin{proof}
    The claim follows by plugging Equation~(\ref{eqn:evals}) into the sequence of observations,
    \begin{equation}
        \min_{\text{unit }\ket{\psi}}\bra{\psi}(\HttU+\HttV)\ket{\psi}=\min_{\text{unit }\ket{\psi}}\bra{\psi}G(1,1)\ket{\psi}=\lmin(G(1,1))=2-\sqrt{2}.
    \end{equation}
\end{proof}
\noindent In words, defining $a=\bra{\psi_1}\HttU\ket{\psi_1}$ and $b=\bra{\psi_1}\HttV\ket{\psi_1}$ (cf. Equation \eqref{eqn:G3}), \Cref{l:complement} says that $a+b\geq 2-\sqrt{2} \gg 0$, i.e. we cannot have $a=b=0$.\\

\noindent\emph{The fix.} For all $t\in P$, update our current construction so that
\begin{equation}\label{eqn:update}
    \Htt:=\Htt(I,iX)= (\HttI)_L\otimes (\HttiX)_R  +(\HttiX)_L\otimes (\HttI)_R
\end{equation}
To recap, to annihilate $\Htt$ for $t\in P$, the first proof $\ket{\psi_1}$ must simulate either $I$ ($a=0$) or $iX$ ($b=0$) at time $t$. If $a=0$ (resp. $b=0$), then $\ket{\psi_2}$ must annihilate $G(0,c)$ (resp. $G(c,0)$) for $c\geq 2-\sqrt{2}$, meaning the history state must simulate application of $iX$ (resp. $I$) at time $t$. (These pieces will be formally combined in \Cref{sscn:finalproof}.)

\paragraph{Obstacle 2: Skipping time steps.} Thus far, we have characterized how a \emph{single} \gadget\ gadget $\Htt$ for $t\in P$ acts in isolation. In particular, when there is a single \gadget\ gadget, we have shown that it is sound, forcing $\ket{\psi_1}\otimes\ket{\psi_2}$ to correctly act like a ``switch'' at time $t$.

The next step is to analyze whether soundness holds in the presence of \emph{multiple} \gadget\ gadgets, which requires analysis of $\Hpropt$ \emph{as a whole}. To do so, define $M(a,b)=a i X + b I$, and rewrite
\begin{equation}
    G(a,b)= -M(a,b)\otimes\ketbra{t}{ {t-1}} -M^\dagger(a,b)\otimes\ketbra{{t-1}}{t} + (a+b)I\otimes(\ketbra{ t}{t}+\ketbra{ {t-1}}{ {t-1}}).\label{eqn:preG}
\end{equation}
It will be helpful to view this as a ``dynamic'' choice of propagation Hamiltonian, where $M(a,b)$ is applied in step $t$. For convenience, we often omit the $a,b$ term and write $M$ henceforth.

The standard approach~\cite{KSV02} for analyzing a propagation Hamiltonian $\hprop$ is to apply a change of basis that maps $\hprop$ to a tri-diagonal matrix encoding a 1D random walk. Unfortunately, this change of basis requires a unitary gate to be applied at each step $t$, and $M$ above is not unitary. However, since we chose $V=iX$ (as opposed to $V=X$),
\begin{equation}\label{eqn:Uab}
    U(a,b):= \frac{1}{\sqrt{a^2+b^2}}M
\end{equation}
\emph{is} unitary. (Aside: It is not necessarily true that $a^2+b^2=1$.) Plugging this into Equation \eqref{eqn:preG}, we have
\begin{equation}
    \begin{aligned}
        G(a,b)=& -\sqrt{a^2+b^2}U(a,b)\otimes\ketbra{t}{ {t-1}} -\sqrt{a^2+b^2}U^\dagger(a,b)\otimes\ketbra{{t-1}}{t} \label{eqn:Gdef}\\
        &+(a+b)I\otimes(\ketbra{ t}{t}+\ketbra{ {t-1}}{ {t-1}}).
    \end{aligned}
\end{equation}
Recall now Kitaev's~\cite{KSV02} change of basis unitary $W$, which for circuit $V_m\cdots V_1$ acting on $R_2R_3$ as in Equation \eqref{eqn:psihist}, is defined as $W=\sum_{t=1}^m V^\dagger_1\cdots V_t^\dagger\otimes\ketbra{t}{t}_{R_4}$, except where for $t\in P$, we now replace $V_t$ with $U_t(a,b)$. Then, restricted to $\Span(\ket{t-1},\ket{t})$ on $R_4$,
\begin{equation}
    WG(a,b)W^\dagger = \left(
      \begin{array}{cc}
        a+b & -\sqrt{a^2+b^2} \\
        -\sqrt{a^2+b^2} & a+b \\
      \end{array}
    \right).
\end{equation}
So, for example, if we chain together time steps $t\in P$ (compute proof bit), $t+1$ (copy proof bit), $t+2\in P$ (uncompute proof bit), the joint propagation Hamiltonian under conjugation by $W$ is:\footnote{
    $W$ only acts on space $R$.
    Recall that $G(a_t,b_t)$ is the residual operator on timestep $t$ after applying $\ket{\psi_1}$ (see Equation \eqref{eqn:G3}), where $a_t,b_t$ are functions in $\ket{\psi_1}$.
}
\begin{equation}\label{eqn:matrix}
\left(
  \begin{array}{cccccccc}
    \ddots & -1 & 0 & 0 & 0 & 0 & 0 &0 \\
    -1 & 2 & -1 & 0 & 0& 0 & 0 & 0 \\
    0 & -1 & 1+a_i+b_i & -\sqrt{a_i^2+b_i^2} & 0 & 0 & 0 & 0 \\
    0 & 0 & -\sqrt{a_i^2+b_i^2} & 1+a_i+b_i & -1 & 0 & 0 & 0 \\
    0 & 0 & 0 & -1 & 1+a_j+b_j & -\sqrt{a_j^2+b_j^2} & 0 & 0 \\
    0 & 0 & 0 & 0 & -\sqrt{a_j^2+b_j^2} & 1+a_j+b_j & -1 & 0 \\
    0 & 0 & 0 & 0 & 0 &-1 & 2 & -1 \\
    0 & 0 & 0 & 0 & 0 & 0 & -1 & \ddots \\
  \end{array}
\right)
\end{equation}
Two remarks: First, $a_i,b_i\geq 0$ since $\HttU\succeq 0$ for all unitary $U$. Second, as a sanity check, when the first prover is honest, we have $a_i=0$ and $b_i=1$ or vice versa for all $i$, reducing us to Kitaev's original 1D random walk matrix on the second prover's space, as expected.\\

\noindent\emph{Reading the fine print: How to break soundness, again.} Equation \eqref{eqn:evals} and \Cref{l:complement} together imply that in isolation, the \gadget\ gadget is sound. However, that analysis was done \emph{restricted to $\Span(\ket{t-1},\ket{t})$} in $R_4$. A cheating prover, on the other hand, \emph{a priori} is under no obligation to place any reasonable weight on time steps $\ket{t-1}$ and $\ket{t}$ in the history state. Indeed, we now sketch a cheating strategy which breaks soundness when multiple \gadget\ gadgets are chained together.

Roughly, the intuition is as follows. Normally, Kitaev's propagation Hamiltonian acts logically as follows: $\Htt$ (Equation (\ref{eqn:Ht})) ensures that the weight on consecutive time steps $t-1$ and $t$ is identical. By chaining together all $\Htt$, we thus obtain that all time steps must have equal amplitude, and moreover this must be non-zero (otherwise we cannot have a unit vector). The reason this breaks down in our current setting is that each $\Htt$ has \emph{two} ways of being satisfied---either $\ket{\psi_L}$ has equal amplitude on steps $t-1$ and $t$, or $\ket{\psi_R}$ does (or possibly both). So let us give a simple example of how to exploit this. Suppose there are $m$ time steps, with time step $1$ and $m$ being proof bit computation steps (i.e. so that $\widetilde{H}_1$ and $\widetilde{H}_m$ encode the \gadget\ gadget). We claim that any unit vector of form
\begin{equation}
    \ket{\psi_L}\otimes\ket{\psi_R}:=\left(\ket{\phi_L}_{R_2R_3}\ket{m}_{R_4}\right)\otimes \left(\ket{\phi_R}_{R_2R_3}\ket{0}_{R_4}\right)
\end{equation}
is in the null space of $\Hpropt$, violating soundness. To see why, note that $\ket{\psi_L}$ trivially annihilates all terms $\Htt$ except $t=m$, since it only has support on $\ket{m}_{R_4}$. As for $\widetilde{H}_m$, while this is not annihilated by $\ket{\psi_L}$, it \emph{is} annihilated by $\ket{\psi_R}$, since the latter only has support on $\ket{0}_{R_4}$. Note that $\ket{\psi_L}$ reciprocates this favor for $\ket{\psi_R}$ at time $t=0$, in that the former annihilates $\widetilde{H}_0$, allowing $\ket{\psi_R}$ to ``hide'' all its amplitude on $\ket{0}_{R_4}$.\\

\noindent\emph{The fix.} The silver lining is that this loophole is highly asymmetric---in our simple example, $\ket{\psi_L}$ and $\ket{\psi_R}$ had their supports on disjoint sets of time steps in $R_4$. To close this loophole, we thus force $\ket{\psi_1}\approx\ket{\psi_2}$ by adding the projector onto the complement of the symmetric subspace (with respect to the $L$ versus $R$ cut) to our Hamiltonian:
\begin{equation}\label{eqn:hsym}
    \Hsym := I-\Psym_{LR}\quad\text{ for } \quad \Psym_{LR}=\frac{1}{2}\left(I_{LR}+\sum_{xy}\ketbra{xy}{yx}_{LR}\right)
\end{equation}
Note $\Hsym$ is sparse (\Cref{def:sparse}); this is the second of two places necessitating our construction to be sparse. Moreover, any $\ket{\psi_L}\otimes\ket{\psi_R}$ satisfying $\Hsym\ket{\psi_L}\otimes\ket{\psi_R}=0$ must have $\ket{\psi_L}=\ket{\psi_R}$ by definition of the symmetric subspace.

\paragraph{The final ingredient.} With symmetry in hand, we give the final ingredient, \Cref{l:nozero}. For this, define for any $t\in P$ (cf. Equation \eqref{eqn:G1} and Equation \eqref{eqn:G3})
\begin{equation}\label{eqn:G}
    a_t:=\bra{\psi_1}\Delta\HttI\ket{\psi_1}\qquad\qquad
    b_t:=\bra{\psi_1}\Delta\HttiX\ket{\psi_1}
\end{equation}
for $\Delta$ defined as needed. The following lemma shows that in the case of perfect symmetry, we may compute a polynomial $\Delta$ in $m$ such that, for all $t\in P$, $a_t+b_t$ cannot be ``too small'', \emph{even if} the adversary can cheat across multiple \gadget\ gadgets.

\begin{lemma}[Full support lemma]\label{l:nozero}
     Define $\Hpropt$ as in Equation (\ref{eqn:full1}), except with the update of Equation \eqref{eqn:update}. Assume perfect symmetry, i.e. $\ket{\psi_1}=\ket{\psi_2}$, and the notation of Equation~(\ref{eqn:G}). For any\footnote{In our use of \Cref{l:nozero}, $\delta$ and $\Delta$ will be functions in parameters such as $m$, i.e. $\delta\in O(1/\poly (m))$ and $\Delta\in\Omega(\poly(m)).$} $\delta\geq 0$ and $\Delta\geq 1$ satisfying $\Delta>\max(36\delta,(8m^4)/c)$ (for $c\in O(1)$ from \Cref{eqn:c}), the following holds: If
    \begin{equation}
        \ptt\Delta\Hpropt\poo\leq 2,
    \end{equation}
    then for all $t\in P$, $a_t+b_t\geq \delta$.
\end{lemma}
\begin{proof}
    As in the claim, assume $\ptt\Delta\Hpropt\poo\leq 2=:\mu$. Next, for sake of contradiction, assume there exists $t^*\in P$ with $0\leq a_{t^*}+b_{t^*}\leq \delta$ (recall $\HttI,\HttiX\succeq 0$ for all $t$).
    To highlight\footnote{Recall that a theme of the present work is to highlight precisely which parts of our construction require unentanglement.} the single place in which the symmetry assumption is used, we run as much as of the proof as possible in full generality (i.e. not requiring $\ket{\psi_1}=\ket{\psi_2}$). We aim to set $\Delta$ so as to achieve a contradiction.

    \paragraph{Step 1: Bounding the weight on time steps $t^*-1$ and $t^*$.} Recalling that $R_4$ is the clock register, let
    \[
        S_t:=\Span \left(I_{R_2,R_3}\otimes\ket{t-1}_{R_4},I_{R_2,R_3}\otimes \ket{t}_{R_4}\right).
    \]
    Then, for all unit vectors $\ket{\phi}\in S_{t^*}$, Lemma~\ref{l:complement} says $\bra{\phi}(\HttsI+\HttsiX)\ket{\phi}\geq 2-\sqrt{2}$. Writing $\ket{\psi_1}=a\ket{\phi_1}+b\ket{\phi_2}$, for $\abs{a}^2+\abs{b}^2=1$ and unit vectors $\ket{\phi_1}\in S_{t^*},\ket{\phi_2}\in S_{t^*}^\perp$, observe
    \[
        \delta\geq a_{t^*}+b_{t^*}= \Delta\bra{\psi_1}(\HttsI+\HttsiX)\ket{\psi_1}=\Delta\abs{a}^2\bra{\phi_1}(\HttsI+\HttsiX)\ket{\phi_1}\geq \Delta\abs{a}^2(2-\sqrt{2}),
    \]
    where the second equality follows since $\HttsI$ and $\HttsiX$ are only supported on $S_{t^*}$ by definition, and the last inequality by  Lemma~\ref{l:complement}. We conclude that
    \begin{equation}\label{eqn:bound}
        \abs{a}\leq\sqrt{\frac{\delta}{\Delta(2-\sqrt{2})}}=:\delta',
    \end{equation}
    implying the weights of $\ket{\psi_1}$ on time steps $t^*-1$ and $t^*$ are each at most $\delta'$ as well, i.e. writing
    \begin{equation}\label{eqn:terms}
        \ket{\psi_1}=\sum_{t=0}^m\ket{\eta_t}_{R_2,R_3}\ket{t}_{R_4}
    \end{equation}
    for vectors $\ket{\eta_t}$, we have $\enorm{\ket{\eta_{t^*-1}}},\enorm{\ket{\eta_{t^*}}}\leq \delta'$.

    \paragraph{Step 2: Decomposing the analysis into computation and proof phases.} We next decompose the analysis into proof and computation phases. By Definition~\ref{def:stream}, we may partition the set of time steps $\set{1,\ldots, m}$ into sets of contiguous time steps $T_1,T_2,T_3,\ldots, T_{m'}$ for $m'\leq m$ as follows. To begin, $T_1$ is set of time steps corresponding to the first time the first computation phase is run (Step 1(a) of Definition~\ref{def:stream}), $T_2$ the gate $W_1$ (proof phase, Step 1(b)i), $T_3$ the single CNOT gate (proof phase, Step 1(b)ii), $T_4$ the second $W_1$ gate (proof phase, Step 1(b)iii). The pattern now repeats itself until we have accounted for all time steps. For simplicity\footnote{This keeps the notation simpler; the proof approach applies analogously even without this assumption.}, we assume $T_{m'}=\set{m}$ is the final time step, which corresponds to an execution of Step 1(b)iii (proof phase, uncompute). Consider now
    \begin{equation}
        \bpot\Delta\Hpropt\kpot=\bra{\psi_1}\Delta\sum_{t\not\in P} H_t\ket{\psi_1}+\bra{\psi_2}\Delta\sum_{t\not\in P} H_t\ket{\psi_2}+
        \bpot\sum_{t\in P}\Delta\Htt\kpot.\label{eqn:long}
    \end{equation}
    We focus on the terms involving $\ket{\psi_1}$. Define $\Todd := \set{T_i\mid i \text{ is odd}}$ (Steps 1(a) and 1(b)ii of \Cref{def:stream}) and $\Teven:=\set{T_i\mid i \text{ is even}}$ (Steps 1(b)i and 1(b)iii of \Cref{def:stream}), and for any $T\subseteq\set{0,\ldots, m}$, define shorthand $H_T:=\sum_{t\in T}H_t$ for $T\in\Todd$ and $\HtT:=\sum_{t\in T}\Htt$ for $T\in\Teven$ (note the former acts on $L$ or $R$, the latter on both $L$ and $R$). Then, by definition
    \begin{equation}\label{eqn:decouple}
        \bra{\psi_1}\Delta\sum_{t\not\in P}H_t\ket{\psi_1} + \bpot\sum_{t\in P}\Delta\Htt\kpot =\sum_{T\in\Todd} \bra{\psi_1}\Delta H_T\ket{\psi_1} + \sum_{T\in \Teven}\Delta\bpot \HtT\kpot.
    \end{equation}

    \noindent As an aside, note that for any distinct sets $A,B\in\Todd$, $H_A$ and $H_B$ have support on disjoint sets of time steps. (This is because $A$ and $B$ must have at least one proof phase $C\in\Teven$ between them.) Moreover, although $\bigcup_{T\in\Todd}\neq[L]$ (since we are missing all proof time steps $P$), nevertheless the \emph{Hamiltonian} $\sum_{T\in\Todd}H_T$ has {support} on all time steps in $\set{0,\ldots, L}$. (This is because each $C\in \Teven$ is a singleton, and each $H_t$ has support on both $\ket{t}$ and $\ket{t-1}$.)

\paragraph{Step 3: Triggering a chain reaction.} With the decomposition of Step 2 in mind, we can now sketch the remaining proof approach at a high level.
\begin{enumerate}
    \item Recall $t^*\in P$, i.e. is in a proof phase, and that from Equations~(\ref{eqn:bound}) and~(\ref{eqn:terms}) that $\enorm{\ket{\eta_{t^*-1}}}\leq \delta'$ and $\enorm{\ket{\eta_{t^*}}}\leq \delta'$.
    \item Since $\enorm{\ket{\eta_{t^*}}}$ is small, we claim this triggers a ``chain reaction'' causing \emph{all} $\enorm{\ket{\eta_{t}}}$ for $t\geq t^*$ to be small. An identical argument also applies to $t^*-1$ and all $t\leq t^*-1$. (For brevity, we show the claim only for the former case; the latter case follows analogously.)
    \item Thus, all amplitudes of $\ket{\psi_1}$ are small, contradicting the fact that $\ket{\psi_1}$ is a unit vector.
\end{enumerate}
To make this formal, and in particular to show the claim in the second point above, we treat proof and computation phases separately.\footnote{In Kitaev's original circuit-to-Hamiltonian construction~\cite{KSV02}, this claim is achieved in one elegant stroke by analyzing the eigenvalues of a random walk matrix which is unitarily equivalent to the propagation Hamiltonian. In our setting, however, we also have operators $\Htt$ acting on both $\ket{\psi_1},\ket{\psi_2}$, i.e. we are not looking at the spectral properties of a propagation Hamiltonian acting solely on $\ket{\psi_1}$.} Consider any $T\in \Todd\cup\Teven$, and suppose $t^*+1$ is the smallest time step in $T$. We show that if $\enorm{\ket{\eta_{t^*}}}$ is small, so is $\enorm{\ket{\eta_{t}}}$ for all $t\in T$. For brevity, define for any $T\in\Todd\cup \Teven$ the projector $\Pi_T:=\sum_{t\in T}\ketbra{t}{t}_{R_4}$, and $\ket{\psi_T}:=\Pi_T\ket{\psi_1}$.\\

\noindent \emph{Case 1: $T\in \Todd$.} Suppose $T=\set{t^*+1,t^*+2,\ldots, t^*+s}$ for some $s$. Then, $H_T$ has support on time steps $\set{t^*,\ldots, t^*+s}$. Now suppose $\enorm{\ket{\eta_{t^*}}}\leq \epsilon$ for arbitrary $\epsilon \geq0$. Since all time steps in $T$ are computation steps, we may use the well-known facts~\cite{KSV02} that:
\begin{enumerate}
     \item (Fact 1) The null space of $H_T$ is the span of all states of form
    \begin{equation}
        \sum_{t={t^*}}^{t^*+s}U_t\cdots U_{t^*+2}U_{t^*+1}\ket{\phi_{\textup{init}}}\ket{t},
    \end{equation}
where $U_t$ is the $t$th computation gate applied in computation phase $T$, and for any initial unit vector $\ket{\phi_{\textup{init}}}$.

    \item (Fact 2) The eigenvalues of $H_T$ are precisely $\lambda_t=2(1-\cos[\pi t/(s+1)])$ for $0\leq t\leq s$, and so the smallest non-zero eigenvalue is
        \begin{equation}\label{eqn:c}
            2(1-\cos(\pi/(s+1)))\geq c/s^2\text{ for some }c\in \Theta(1).
        \end{equation}
\end{enumerate}
Defining $T':=T\cup\set{t^*}$, consider now $\ket{\psi_{T'}}=\sum_{t={t^*}}^{t^*+s}\ket{\eta_t}\ket{t}$ (recall $\ket{\psi_{T'}}$ is $\ket{\psi_1}$ projected onto time steps in ${T'}$) for vectors $\set{\ket{\eta_t}}$ of possibly \emph{differing} norms. We claim that $\enorm{\ket{\eta_t}}$ is small for all $t\in\set{t^*,\ldots, t^*+s}$.

To see this, by assumption, $\enorm{\ket{\eta_{t^*}}}\leq \epsilon$ and $\bra{\psi_{T'}}\Delta H_T\ket{\psi_{T'}}\leq \mu$. Writing $\ket{\psi_{T'}}$ in terms of its components in the null space ($\ket{\psi_{{T'},0}}$) and support ($\ket{\psi_{{T'},+}}$) of $H_T$, respectively, i.e.
\begin{equation}
    \ket{\psi_{T'}}=\sum_{t={t^*}}^{t^*+s}\ket{\eta_t}\ket{t}=\sum_{t={t^*}}^{t^*+s}\ket{\eta_{t,0}}\ket{t}+\sum_{t={t^*}}^{t^*+s}\ket{\eta_{t,+}}\ket{t}=:\ket{\psi_{{T'},0}}+\ket{\psi_{{T'},+}},
\end{equation}
we have
\begin{equation}
    \mu\geq \bra{\psi_{T'}}\Delta H_T\ket{\psi_{T'}}=\bra{\psi_{{T'},+}}\Delta H_T\ket{\psi_{{T'},+}}\geq \frac{c\Delta\enorm{\ket{\psi_{{T'},+}}}^2}{s^2},
\end{equation}
where the last inequality follows from Fact 2. Thus, $\enorm{\ket{\psi_{{T'},+}}}^2\leq s^2\mu/(c\Delta)$. But by Fact 1, all $\ket{\eta_{t,0}}$ have the same norm with $\enorm{\ket{\eta_{t^*,0}}}\leq \epsilon$ (since $\enorm{\ket{\eta_{t^*}}}\leq \epsilon$ by assumption), and each $\ket{\eta_{t,+}}$ has norm $\enorm{\ket{\eta_{t,+}}}\leq s\sqrt{\mu/(c\Delta)}$. By the triangle inequality, we conclude that for all $t\in\set{t^*,\ldots, t^*+s}$,
\begin{equation}\label{eqn:casecomp}
    \enorm{\ket{\eta_t}}\leq \epsilon + s\sqrt{\frac{\mu}{c\Delta}}.
\end{equation}

\noindent \emph{Case 2: $T\in \Teven$.} For concreteness, suppose $T=\set{t^*+1}$, so that $H_T$ has support on time steps $T':=\set{t^*, t^*+1}$. Now suppose $\enorm{\ket{\eta_{t^*}}}\leq \epsilon$ for arbitrary $\epsilon \geq0$. Letting $F_{t^*+1}$ denote an arbitrary Feynman-Kitaev propagation term (Equation (\ref{eqn:Ht})) for arbitrary unitary $U_{t^*+1}$ at time $t^*+1$,
\begin{eqnarray}
    \bra{\psi_1}F_{t^*+1}\ket{\psi_1}=\bra{\psi_{T'}}F_{t^*+1}\ket{\psi_{T'}}&=&
    \frac{1}{2}\enorm{\ket{\eta_{t^*}}}^2+\frac{1}{2}\enorm{\ket{\eta_{t^*+1}}}^2-\operatorname{Re}(\bra{\eta_{t^*+1}}U_{t^*+1}\ket{\eta_{t^*}})\\
    &\geq&\frac{1}{2}(\enorm{\ket{\eta_{t^*}}}-\enorm{\ket{\eta_{t^*+1}}})^2,\label{eqn:LB}
\end{eqnarray}
where the inequality follows by the Cauchy-Schwarz inequality and unitary invariance of the Euclidean norm. Suppose $\enorm{\ket{\eta_{t^*+1}}}=\enorm{\ket{\eta_{t^*}}}+\zeta$ for $\zeta\in\R$.
By Equation~(\ref{eqn:LB}),
$\bra{\psi_1}F_{t^*+1}\ket{\psi_1}\geq \zeta^2/2$. And now we use the assumption that $\ket{\psi_1}=\ket{\psi_2}$ to obtain that
\begin{equation}
    \mu\geq\Delta \bpot\left( H^I_{t^*+1}\otimes  H^{iX}_{t^*+1}  + H^{iX}_{t^*+1}\otimes  H^I_{t^*+1} \right)\kpot\geq\frac{\Delta\zeta^4}{2},
\end{equation}
where we have substituted $H^I_{t^*+1}$ or $H^{iX}_{t^*+1}$ for $F_{t^*+1}$, as appropriate.
We conclude that
\begin{equation}\label{eqn:casepf}
    \enorm{\ket{\eta_{t^*+1}}}\leq \epsilon + \zeta\leq \epsilon + \left(\frac{2\mu}{\Delta}\right)^{\frac{1}{4}}.
\end{equation}

\paragraph{Step 4: Combining all bounds.} By Equations~(\ref{eqn:bound}) and~(\ref{eqn:terms}), there exists a $t^*\in P$ with $\enorm{\ket{\eta_{t^*-1}}}\leq\delta'$ and $\enorm{\ket{\eta_{t^*}}}\leq \delta'$ for $\delta':=\sqrt{\delta/(\Delta(2-\sqrt{2}))}$. Running the chain reaction upwards from $t^*$ (respectively, downwards from $t^*-1$):
\begin{itemize}
    \item Each time we encounter a computation phase $T\in\Todd$, Equation~(\ref{eqn:casecomp}) says we increase our norm by at most an additive factor of $s\sqrt{\mu/(c\Delta)}$.
    \item Each time we encounter a proof phase $T\in\Teven$, Equation~(\ref{eqn:casepf}) says we increase our norm by at most an additive factor of $(2\mu/\Delta)^{\frac{1}{4}}$.
\end{itemize}
We hence have the (naive) upper bound
\begin{equation}\label{eqn:contra}
    1=\enorm{\ket{\psi_1}}\leq \sum_{t=0}^m\enorm{\ket{\eta_t}}\leq 2\sqrt{\frac{\delta}{\Delta(2-\sqrt{2})}}+(m-2)\left(m\sqrt{\frac{\mu}{c\Delta}}+\left(\frac{2\mu}{\Delta}\right)^{\frac{1}{4}}\right),
\end{equation}
where the first inequality follows by the triangle inequality, and the second\footnote{This is a naive bound, since for each phase we are charging both $s\sqrt{\mu/(c\Delta)}$ and $(2\mu/\Delta)^{\frac{1}{4}}$ for simplicity, rather than introducing additional notation to carefully account for each type of phase.} since $s\leq m$ in Equation~(\ref{eqn:casecomp}). Since $\delta\geq 0$, $\Delta\geq 1$, and $\mu=2\in\Theta(1)$, choosing $\Delta>\max(36\delta,(8m^4)/c)$ yields a contradiction, completing the proof.
\end{proof}

\subsection{Final proof: Combining all ingredients}\label{sscn:finalproof}

With the ingredients of \Cref{sscn:ingredients} in hand, we are ready to restate and prove the main lemma of this section.
{\renewcommand\footnote[1]{\footref{footnote:embedding-lemma}}\embed*}
\begin{proof}
    To reduce clutter, we omit the dependence on $n$ when referring to functions $p,q,r,m,\alpha,\beta$. We assume all notation and definitions of \Cref{sscn:ingredients}. Define
    \begin{equation}\label{eqn:H}
        \Ht = \Din\Hint+\Dprop\Hpropt+\Dsym\Hsym+\Houtt,
    \end{equation}
     where for convenience we restate definitions
    \begin{eqnarray}
    \Hint &=& (\hin)_L\otimes I_R + I_L\otimes (\hin)_R\\
    \Hpropt &=& \sum_{t=1}^m \Htt,\quad \text{where $\Htt$ is defined as}\\
     \Htt&=&
    \begin{cases}
         (\HttI)_L\otimes (\HttiX)_R  +(\HttiX)_L\otimes (\HttI)_R & \text{ if }t\in P\\
         (H_t)_L\otimes I_R + I_L\otimes (H_t)_R & \text{ if } t\not\in P\label{eqn:noP}
    \end{cases}\\
    \Houtt &=& (\hout)_L\otimes I_R + I_L\otimes (\hout)_R\label{eqn:houtb}\\
    \Hsym &=& I-\Psym_{LR}\quad\text{ for } \quad \Psym_{LR}=\frac{1}{2}\left(I_{LR}+\sum_{xy}\ketbra{xy}{yx}_{LR}\right),\label{eqn:Hsym-def}
\end{eqnarray}
and $\Din,\Dprop,\Dsym$ are set as follows. Set $M:=(m+1)2^{r}$.
Then, define\footnote{We have not attempted to optimize these parameters.} $\Din=M^{31}$, $\Dprop=72M^{31}$, and $\Dsym=M^{66+2k}$, where $q(n)\in O(n^k)$ for some $k\in O(1)$ (recall $q$ is the poly-bounded number of ancilla qubits in circuit $Q_n$).
Next, set $\alpha'=2\frac{1-\alpha}{m+1}$ and $\beta'=2\frac{1-\beta}{m+1}-\frac{1}{M}$, where recall $\alpha-\beta\geq 2^{-r}$ by assumption. Observe $\Ht$ acts on $O(q(n)+\log(m(n)))$ qubits (workspace and clock register encoded in binary, respectively). Importantly, $\Ht$ is sparse (in the sense of \Cref{def:sparse}; here we use the fact that although $\Hprop$ has $m$ terms, which may be exponential, each such term has support on only $2$ basis states in the clock register in \Cref{eqn:hpropproper}). For clarity, this means our reduction does \emph{not} output the explicit Hamiltonian $\Ht$, but rather the classical algorithm of \Cref{def:sparse} which produces entries of $\Ht$ on demand. Finally, the norm of $\Ht$ is $\snorm{\Ht}\in\poly(m,2^r)$, as claimed.

\paragraph{Correctness.} Assume $(Q_n,x)$ is a YES case. Let $Q_n=V'_m\cdots V'_2V'_1$. For each $t\in P$ with $V'_t=X$ (i.e. a proof bit of $1$ is streamed at time $t$), define $V_t:=iX$, and for all $t\not\in P$, define $V_t:=V'_t$. Recall the history state of Equation \eqref{eqn:psihist}, i.e.
\begin{equation}
    \ket{\psihist(y)}=\frac{1}{\sqrt{m+1}}\sum_{t=0}^mV_t\cdots V_1\ket{0\cdots 0}_{R_2}\ket{0}_{R_3}\ket{t}_{R_4},
\end{equation}
where $y$ is implicitly encoded by the choice of gates $V_t'$ for $t\in P$. It is straightforward to verify
\begin{eqnarray}
   \Hint\ket{\psihist}\otimes\ket{\psihist}=\Hpropt\ket{\psihist}\otimes\ket{\psihist}=\Hsym\ket{\psihist}\otimes\ket{\psihist}&=&0,\text{ and}\\
    \bra{\psihist}\otimes\bra{\psihist}\Houtt\ket{\psihist}\otimes\ket{\psihist}\leq \frac{2(1-\alpha)}{m+1}&=&\alpha',
\end{eqnarray}
where the factor $2$ appears since $\Houtt$ contains two copies of $\hout$. Thus, completeness holds.

Assume next that $(Q_n,x)$ is a NO case. Assume, for sake of contradiction, there exists $\LR$ such that $\RL \Ht \LR \leq \beta'$. The soundness analysis proceeds in steps. Throughout, recall $\Hint,\Hpropt,\Hsym,\Houtt\succeq 0$.\\

\noindent\emph{Step 1: Enforcing symmetry.} We first show that, up to small additive error, we may assume $\ket{\psi_1}=\ket{\psi_2}$. By assumption,
\[
    \RL \Delta\Hsym \LR \leq \RL \Ht \LR \leq \beta',
\]
from which we have $\pto \Hsym\pot \leq \beta'/\Dsym$. Since the spaces $L$ and $R$ have the same dimension, we may write $\ket{\psi_2}=a\ket{\psi_1}+b\ket{\psi_1^\perp}$ for $\abs{a}^2+\abs{b}^2=1$ and $\ket{\psi_1^\perp}$ orthogonal to $\ket{\psi_1}$. We thus have
\begin{eqnarray}
    \frac{\beta'}{\Dsym}&\geq& \pto \Hsym\pot\\
    &=&\abs{b}^2\bra{\psi_1}\bra{\psi_1^\perp}\Hsym\ket{\psi_1}\ket{\psi_1^\perp}\\
    &=&\frac{1}{2}\abs{b}^2\bra{\psi_1}\bra{\psi_1^\perp}\left(I_{LR}-\sum_{xy}\ketbra{xy}{yx}_{LR}\right)\ket{\psi_1}\ket{\psi_1^\perp}\\
    &=&\frac{1}{2}\abs{b}^2,
\end{eqnarray}
    where the third statement follows by substituting the definition of $\Hsym$, and the fourth since $\sum_{xy}\ketbra{xy}{yx}$ is the SWAP operator (and so $(\sum_{xy}\ketbra{xy}{yx})\ket{\psi_1}\ket{\psi_1^\perp}=\ket{\psi_1^\perp}\ket{\psi_1}$). Applying identity $\trnorm{\ketbra{u}{u}-\ketbra{v}{v}}=2\sqrt{1-\abs{\braket{u}{v}}^2}$ (\Cref{lem:trace-dist}), we conclude
\begin{equation}\label{eqn:gamma1}
    \trnorm{\ketbra{\psi_1}{\psi_1}\otimes \ketbra{\psi_2}{\psi_2}-\ketbra{\psi_1}{\psi_1}\otimes \ketbra{\psi_1}{\psi_1}}\leq 2\sqrt{\frac{2\beta'}{\Dsym}}\leq \frac{2\sqrt{2}}{M^{33}}=:\gamma_1.
\end{equation}

\noindent\emph{Step 2: Extracting a history state which is ``good enough''.} We next treat $\Hint$ and $\Hpropt$ simultaneously. Similar to Step 1, $\pto (\Din\Hint+\Dprop\Hpropt)\pot\leq \beta'$ by assumption. Combining this with Equation \eqref{eqn:gamma1}, the Hölder inequality, and the triangle inequality yields
    \begin{eqnarray}
        \ptt\left(\Din\Hint+\Dprop\Hpropt\right)\poo
        &\leq&\beta'+2\sqrt{\frac{2\beta'}{\Dsym}}\left(\Din q+ 2\Dprop m\right)\\
        &\leq& 1+O\left(\frac{1}{M}\right)\\
        &=:&\zeta, \label{eqn:zeta}
    \end{eqnarray}
    where (1) we are implicitly writing $\Hin$ as a sum of $1$-local terms as is standard, e.g. via trick $I-\ketbra{00}{00}\preceq \ketbra{1}{1}\otimes I + I\otimes \ketbra{1}{1}$, and so $\snorm{\Hint}\leq q$, (2) since for any $t$, $\snorm{H_t}=1$, implying $\snorm{\Hpropt}\leq 2m$ by the triangle inequality, and (3) we use that $q(n)\in O(n^k)$. 
    
    Our strategy for this step is now as follows. We first exploit \Cref{l:nozero} to extract from $\Hpropt$ a ``proper'' Feynman-Kitaev propagation Hamiltonian (i.e. in the form of Equation \eqref{eqn:hpropproper}). We then couple the latter with $\Hint$ and Equation \eqref{eqn:zeta} to argue that $\ket{\psi_1}$ must be close to a history state. This history state will not be exactly what we need, but we will show in the next step that it is ``good enough''.

    To begin, recall
    \begin{eqnarray}
        \ptt \Dprop \Hpropt \poo&=&\bra{\psi_1}\left(2\sum_{t\not\in P}\Dprop H_t+\sum_{t\in P}G(a_t,b_t)\right)\ket{\psi_1}\label{eqn:85}\\
        a_t&=&\bra{\psi_1}\Dprop\HttI\ket{\psi_1}\geq 0\\
        b_t&=&\bra{\psi_1}\Dprop\HttiX\ket{\psi_1}\geq 0
    \end{eqnarray}
    for $G(a_t,b_t)$ from Equation \eqref{eqn:Gdef}. We now show how to extract a ``proper'' Feynman-Kitaev propagation Hamiltonian from the right hand side of Equation \eqref{eqn:85}.

    \begin{lemma}\label{l:decompose}
        Assume the notation above, and that $\ptt\Dprop\Hpropt\poo\leq 2$. Suppose that  $\delta'\geq 0$ and $\Dprop\geq 1$ satisfy $\Dprop>\max(36\sqrt{2}\delta',(8m^4)/c)$ (for $c\in O(1)$ from \Cref{eqn:c}). For all $t\in P$, define $F_t$ to be the Feynman-Kitaev propagation term (Equation (\ref{eqn:Ht})) for unitary $U(a_t,b_t)$ from Equation \eqref{eqn:Uab}.
        Then, 
        \begin{equation}
            2\Dprop\sum_{t\not\in P}H_t+\sum_{t\in P} G(a_t,b_t) \succeq \delta'\left(\sum_{t\not\in P}H_t+\sum_{t\in P}F_t\right).
        \end{equation}
    \end{lemma}
    \begin{proof}
        Consider first the case of $t\in P$. Recall
        \begin{equation}
            \begin{aligned}
            G(a_t,b_t)=& -\sqrt{a^2+b^2}U(a_t,b_t)\otimes\ketbra{t}{ {t-1}} -\sqrt{a^2+b^2}U^\dagger(a_t,b_t)\otimes\ketbra{t-1}{t}\\
            &+(a_t+b_t)I\otimes(\ketbra{ t}{t}+\ketbra{ {t-1}}{ {t-1}}).
            \end{aligned}
        \end{equation}
        Set $\delta=\sqrt{2}\delta'$. Then, we have by \Cref{l:nozero} that $a_t+b_t\geq\sqrt{a_t^2+b_t^2}\geq \delta'$ (here we use $\norm{\cdot}_1\geq\norm{\cdot}_2\geq \norm{\cdot}_1/\sqrt{2}$ for $\C^2$). Thus, defining $s_1:=\sqrt{a^2+b^2}-\delta'$ and $s_2:=a+b-\delta'$, we may rewrite
        \begin{equation}
            \begin{aligned}
            G(a_t,b_t)=\delta' F_t + [&-s_1U(a_t,b_t)\otimes\ketbra{t}{{t-1}} - s_1U^\dagger(a_t,b_t)\otimes\ketbra{t-1}{t}\\
            &+s_2I\otimes(\ketbra{ t}{t}+\ketbra{t-1}{t-1})].
            \end{aligned}
        \end{equation}
        Since $a+b\geq \sqrt{a^2+b^2}$ for all $a,b\geq 0$, we have $s_2\geq s_1\geq 0$, implying $G(a_t,b_t)-\delta'F_t\succeq 0$. Similarly for $t\not\in P$, since $\Dprop\in\omega(\delta')$ by assumption, we have
        $
            (2\Dprop-\delta')H_t\succeq 0,
        $
        from which the claim follows.
    \end{proof}
    To apply \Cref{l:decompose}, set $\delta'=M^{31}$. By \Cref{eqn:zeta}, $\ptt\left(\Dprop\Hpropt\right)\poo\leq 2$. Thus, Equation \eqref{eqn:85} and \Cref{l:decompose} yield
    \begin{equation}\label{eqn:88}
        \bra{\psi_1}\left(2\Dprop\sum_{t\not\in P}H_t+\sum_{t\in P}G(a_t,b_t)\right)\ket{\psi_1}\geq \bra{\psi_1}\left( \delta'\left(\sum_{t\not\in P}H_t+\sum_{t\in P}F_t\right)\right)\ket{\psi_1}=:\delta'\bra{\psi_1} \Hpropl\ket{\psi_1}.
    \end{equation}
    Note $\Hpropl$ is a standard Feynman-Kitaev propagation Hamiltonian over all $m$ time steps. So, set $\Din=\delta'$, and combine Equation \eqref{eqn:zeta}, Equation \eqref{eqn:85}, and Equation \eqref{eqn:88} to obtain 
    \begin{eqnarray}
        \zeta&\geq& \ptt\Din\Hint+\Dprop\Hpropt\poo\\
        &\geq& \delta'\left(\ptt\Hint\poo+\bra{\psi_1}\Hpropl\ket{\psi_1}\right)\\
        &\geq& \bra{\psi_1} \delta'\left(\Hin+\Hpropl\right)\ket{\psi_1},\label{eqn:91}
    \end{eqnarray}
    where the last inequality follows since $\ptt \Hint\poo=2\bra{\psi_1}\Hin\ket{\psi_1}$ and since $\Hin\succeq 0$. Since $\Hin+\Hpropl$ is a standard Feynman-Kitaev construction, it is known\footnote{More accurately, Lemma 3 of \cite{GK12} shows this lower bound for $\hin+\hprop+\hstab$, but the $\hstab$ term is easily omitted while retaining the bound.} (Lemma 3 of \cite{GK12}) that its smallest non-zero eigenvalue scales as $\Omega(1/m^3)$. Moreover, for the null space of $\Hin+\Hprop$, since $\Hin$ requires time step $t=0$ to be initialized to $\ket{0\cdots 0}_{R_2}\ket{0}_{R_3}\ket{0}_{R_4}$, we have that \emph{conditioned on any $\ket{\psi_1}$} on system $L$, $\Hin+\Hpropl$ in system $R$ has \emph{unique} null vector
    \begin{equation}\label{eqn:psihistbasic}
        \ket{\psihist}=\frac{1}{\sqrt{m+1}}\sum_{t=0}^mV_t\cdots V_1\ket{0\cdots 0}_{R_2}\ket{0}_{R_3}\ket{t}_{R_4},
\end{equation}
    with unitaries $V_t$ for $t\in P$ defined as $V_t=U(a_t,b_t)$. (The uniqueness follows since there is no proof register in our setting, in contrast to the setting of the local Hamiltonian problem for QMA.) Note $\ket{\psihist}$ is \emph{not} our desired history state $\ket{\psihist(y)}$ (Equation \eqref{eqn:psihist}), since unitaries $U(a_t,b_t)$ do not necessarily simulate the honest action of applying $I$ or $X$ for the proof bit at step $t$. (Step 3 will show, however, that $\ket{\psihist}$ is nevertheless ``good enough''.)

    Finally, we combine these observations to confirm $\ket{\psi_1}$ can be made close to $\ket{\psihist}$ for our choice of $\delta'$. Write $\ket{\psi_1}=a\ket{\psihist}+b\ket{\psihist^\perp}$ for $\abs{a}^2+\abs{b}^2=1$. Then, by Equation~(\ref{eqn:91}),
    \begin{equation}
        \frac{\zeta}{\delta'}\geq \abs{b}^2\bra{\psihist^\perp}\Hin+\Hprop\ket{\psihist^\perp}\geq \frac{\abs{b}^2c}{m^3}
    \end{equation}
    for some $c\in \Theta(1)$ (recall $\Hin+\Hprop$ has min non-zero eigenvalue $\Omega(1/m^3)$). Thus,
    \begin{eqnarray}
        \trnorm{\ketbra{\psi_1}{\psi_1}^{\otimes 2}-\ketbra{\psihist}{\psihist}^{\otimes 2}}=2\sqrt{1-\abs{\braket{\psi_1}{\psihist}}^4}
        \leq 4\sqrt{\frac{m^3}{c\delta'}\zeta}
        \leq \frac{8}{\sqrt{c}}\frac{1}{M^{14}}
        =:\gamma_2,\label{eqn:gamma2}
    \end{eqnarray}
    where the second statement holds since $m^3\zeta/(c\delta')<1$, and since $\Din=\delta'$. One comment is important here: Above there is the subtlety that $\ket{\psihist}$ is conditioned on $\ket{\psi_1}$, so it would be more accurate to write $\ket{\psihist(\psi_1)}$. Thus, what the trace distance bound above shows is that any low-energy $\ket{\psi_1}$ (in the sense of Equation \eqref{eqn:zeta}) must be close to the history state $\ket{\psihist(\psi_1)}$ it defines.\\

    \noindent\emph{Step 3. Why $\ket{\psihist}$ is good enough.} We have shown that for any $t\in P$, there exist scalars $a_t,b_t\geq 0$, such that $\ket{\psihist}$ applies unitary $V_t=U(a_t,b_t)$ at time $t$. Recall that
    \begin{equation}
        U(a_t,b_t)= \frac{1}{\sqrt{a_t^2+b_t^2}}(a_tiX+b_tI).
    \end{equation}
    In the honest case, recall that for all $t\in P$ the history state would choose $\ket{\psi_1}$ on system $L$ so that either $a_t=0$ and $b_t=1$ (corresponding to streaming proof bit $0$ in step $t$) or $a_t=1$ and $b_t=0$ (corresponding to streaming proof bit $1$ in step $t$). We now argue that for any low-energy $\ket{\psihist}$, this must \emph{approximately} hold.

    First, by Equations (\ref{eqn:gamma1}), (\ref{eqn:gamma2}), the Hölder inequality, and the triangle inequality, for all $\Htt$,
    \begin{equation}
        \bra{\psihist}\bra{\psihist}\Dprop\Htt\ket{\psihist}\ket{\psihist}\leq \beta' + \left(\gamma_1+\gamma_2\right)\snorm{\Dprop\Htt}\leq\beta' + 2\Dprop\left(\gamma_1+\gamma_2\right),\label{eqn:99}
    \end{equation}
    where the last statement holds since $\snorm{\Htt}\leq 2$. But
    \begin{eqnarray}
        \bra{\psihist}\bra{\psihist}\Dprop\Htt\ket{\psihist}\ket{\psihist}&=&2\Dprop\bra{\psihist}\HttI\ket{\psihist}\bra{\psihist}\HttiX\ket{\psihist}\\
        &=&\frac{8\Dprop}{(m+1)^2}\left(1-\frac{b_t}{\sqrt{a_t^2+b_t^2}}\right)\left(1-\frac{a_t}{\sqrt{a_t^2+b_t^2}}\right).\label{eqn:ab}
    \end{eqnarray}
    Assume without loss of generality that $b_t\geq a_t\geq 0$. Then, combining Equation \eqref{eqn:99} with Equation (\ref{eqn:ab}) and rearranging yields
    \begin{equation}\label{eqn:whocares}
        \frac{b_t}{\sqrt{a_t^2+b_t^2}}\geq 1-\frac{m+1}{2}\sqrt{\frac{\beta'}{2\Dprop}+\gamma_1+\gamma_2}=:1-\epsilon
    \end{equation}
    for $\epsilon\geq 0$, where our parameter choices ensure $\epsilon \ll 1$. From this, we also conclude
    \begin{equation}
        \frac{a}{\sqrt{a^2+b^2}}\leq \sqrt{1-(1-\epsilon)^2}\leq \sqrt{2\epsilon}.
    \end{equation}
    We conclude that when $b\geq a$, it must be that $\ket{\psihist}$ applied a unitary close to $I$ at time $t$, i.e.
    \begin{equation}
        \snorm{U(a_t,b_t)-I}=\left\lVert \frac{a}{\sqrt{a^2+b^2}}iX+\left(\frac{b}{\sqrt{a^2+b^2}}-1\right)I \right\rVert_\infty\leq \sqrt{2\epsilon}+\epsilon\leq 4\sqrt{\epsilon},
    \end{equation}
    where the second statement follows since $1\geq b/(\sqrt{a^2+b^2})\geq 1-\epsilon$, and the last since $\epsilon\leq\sqrt{2\epsilon}$ for small $\epsilon$. An essentially identical calculation shows that in the complementary case when $a_t\geq b_t\geq 0$, $\snorm{U(a_t,b_t)-iX}\leq 4\sqrt{\epsilon}$. (Note $a_t=b_t$ is impossible, as otherwise Equation \eqref{eqn:whocares} yields a contradiction for small $\epsilon$.)

    Finally, recalling the definition of $\ket{\psihist}$ from Equation \eqref{eqn:psihistbasic}, we ``round'' to an \emph{honest} circuit $V'=V'_m\cdots V'_1$ as follows. For $t\not\in P$, set $V'_t=V_t$, and for $t\in P$, set $V'_t =I$ if $b_t>a_t$  and $V'_t=iX$ if $b_t<a_t$. Then, for all $t\in[m]$, we have $\snorm{V_t-V'_t}\leq 4\sqrt{\epsilon}$, from which we conclude via standard bounds that
    \begin{equation}\label{eqn:gamma3}
        \snorm{V_m\cdots V_1-V_m'\cdots V_1'}\leq 4m\sqrt{\epsilon}=4m\sqrt{\frac{m+1}{2}\sqrt{\frac{\beta'}{2\Dprop}+\gamma_1+\gamma_2}}=:\gamma_3.
    \end{equation}
    There is a minor subtlety we should clarify at this point. By construction, for any $t\in P$, $V'$ applies either $I$ or $iX$, as desired. Then, \Cref{def:stream} has the additional structure that each $W_i\in \set{I,X}$ in a ``compute'' proof phase (Step 1(b)i) is subsequently undone by a matching $W^\dagger_i\in\set{I,X}$ in the ``uncompute'' proof phase (Step 1(b)iii). Let $t,t+2\in P$ be an arbitrary pair of such ``compute'' and ``uncompute'' steps, respectively. Then, our construction only enforces that $V'_t,V'_{t+2}\in\set{I,iX}$, but not that $V'_{t+2}=(V'_t)^\dagger$. However, this is without loss of generality, since any streaming proof which deviates from this pattern can easily be simulated by a ``proper'' streaming proof without increasing the proof length\footnote{For example, suppose at step $t$ and $t+2$, $V'$ applies $iX$ and $I$. This corresponds to classically streaming bit $1$ in step $t$, but not uncomputing register $R_3$ from $\ket{1}$ back to $\ket{0}$ in step $t+2$. Logically, this just has the effect of negating the standard basis, so that when the next proof bit is streamed, $iX$ and $I$ now correspond to streaming bits $0$ and $1$, respectively (as opposed to $1$ and $0$).}. Thus, deviating from this pattern cannot increase the best acceptance probability over all streamed proofs $y$.\\

    \noindent \emph{Step 4: The contradiction.} Recall that $(Q_n,x)$ is a NO case, and that we have assumed, for sake of contradiction, that $\pto \Ht \pot \leq \beta'$. The former implies that for any streaming proof $y$, $Q_n$ accepts with probability at most $\beta$. But $V'$ is by construction the verifier $Q_n$, except with all gates at times $t\in P$ ``rounded'' to the closest gate in $\set{iX,I}$. Thus, $V'$ simulates $Q_n$ on \emph{some} streaming proof $y$, implying $V'$ also accepts with probability at most $\beta$. Since $\ket{\psihist}$ encodes circuit $V$ with $\snorm{V-V'}\leq \gamma_3$ (Equation \eqref{eqn:gamma3}), we conclude that
     $   \textup{Pr}(V\text{ accepts})\leq\beta+\gamma_3$.
    Thus,
    \begin{equation}
        \bra{\psihist}\bra{\psihist}\Houtt\ket{\psihist}\ket{\psihist}=2\bra{\psihist}\Hout\ket{\psihist}\geq 2\frac{1-\beta}{m+1}-\frac{2\gamma_3}{m+1},
    \end{equation}
    which by the Hölder inequality, Equation (\ref{eqn:gamma1}), and Equation (\ref{eqn:gamma2}) implies
    \begin{equation}\label{eqn:final}
        \beta'\geq \pto \Ht \pot\geq \pto\Houtt\pot \geq 2\frac{1-\beta}{m+1}-\frac{2\gamma_3}{m+1}-2(\gamma_1+\gamma_2),
    \end{equation}
    where we have used $\snorm{\Houtt}=2$. Combining \Cref{eqn:final} with Equations (\ref{eqn:gamma1}), (\ref{eqn:gamma2}), and (\ref{eqn:gamma3}), we obtain $\frac{2\gamma_3}{m+1}+2(\gamma_1+\gamma_2)<1/M=\frac{1}{(m+1)2^{r+1}}$, obtaining the desired contradiction. 
\end{proof}

\section{Applications of the Embedding Lemma}\label{scn:embeddingapps}

In this section, we apply the Embedding Lemma (\Cref{l:embed}) to obtain various corollaries. These proceed in two steps. \Cref{sscn:reductions} first reduces problems from various complexity classes into instances of Separable Sparse Hamiltonian (SSH). \Cref{sscn:containmentqmat} then shows how the exact structure of the SSH instances from \Cref{l:embed} can be exploited to obtain various upper bounds of form $\QMA(2,p,q,r)$ for appropriate $p,q,r$.

\subsection{Reductions to Separable Sparse Hamiltonian (\texorpdfstring{$\SSH$}{SSH})}\label{sscn:reductions}

The first corollary is immediate by recalling that without loss of generality, a $\SQCMASPACE$ circuit has $m\in \Theta(2^p)$.

\begin{restatable}{corollary}{corQSPACE}\label{cor:QSPACE}
    There exists a poly-time many-one reduction from any $\SQCMASPACE(p,q,r)$ instance to an instance of Separable Sparse Hamiltonian on $O(q+ \log p)$ qubits with promise gap $\Omega(2^{-p-r})$.
\end{restatable}
\noindent
The second corollary requires slightly more work, but still goes by combining \Cref{l:embed} with completely standard techniques.

\begin{restatable}{corollary}{corMIP}\label{cor:MIP}
    There exists a poly-time many-one reduction from any $\MIP(t,u,v,p,r,c,s)$ protocol to an instance of Separable Sparse Hamiltonian on $O(u+v+\log(tr\log(pt)))$ qubits with promise gap scaling as $\Omega\left(\left[2^{tr\log(pt)}(c-s)\right]^{-1}\right)$.
\end{restatable}
\begin{proof}[Proof sketch]
    Apply the standard trick of concatenating, for all possible sequences of questions from the verifier $V$ to the provers, the corresponding sequence of all answers from the provers. This will be the proof $y$ to be streamed, and it has length $\abs{y}= pt2^{tr}$. Without loss of generality, we may assume $y$ first records, in order, all possible answers from the provers to the verifier's first round of questions (call this ``block 1''), followed by all possible answers from the provers to the second round of questions (``block 2''), etc. (Note the questions in a given round can depend on the answers from all previous rounds). Thus, given streaming access to $y$, a \SQCMASPACE verifier $Q$ can straightforwardly simulate $V$ as follows: For each round $t$ of the MIP protocol, $Q$ simulates $V$ to select its questions. It then streams block $t$ of $y$, storing only the answers to the questions selected for round $t$. It then proceeds to round $t+1$.

    Let us analyze $Q$'s parameters. First, note that $Q$'s \emph{time} complexity increases to $\Theta(\abs{y})\in\Theta(pt2^{tr})$---this is because under \Cref{def:stream}, $Q$'s total gate count also counts the gates used to stream bits, of which there are $\abs{y}$. For clarity, of these $\Theta(pt2^{tr})$ time steps, only $\poly(n)$ are used to simulate computation steps of $V$. We now discuss space complexity. For this, random bits can be simulated via the principle of deferred measurement~\cite{Nielsen}. This requires the use of a fresh ancilla qubit for each measurement. Since $V$ uses $v$ random bits, and $u$ ancilla space, $Q$'s overall space requirement is $u+v$. Summarizing, for this construction, $Q$ has parameters $p=tr\log(pt)$, $q\in O(u+v)$, $r\in O(\log(c-s))$. Applying \Cref{l:embed} to $Q$ now yields the claim.
\end{proof}

%

From \Cref{cor:MIP}, we immediately obtain the following, since $\NP\subseteq \MIP(\log,\log, 2,1,1,1-1/\poly(n))$ (\Cref{scn:defs}).

\begin{corollary}\label{cor:NP}
    Any instance of an NP language is reducible under poly-time many-one reductions to an instance of Sparse Separable Hamiltonian on $O(\log(n))$ qubits with completeness $1$ and soundness $1-1/\poly(n)$.
\end{corollary}
\noindent Note the completeness $1$ arises since $\alpha'=0$ in the proof of \Cref{l:embed} (since the MIP has completeness $1$ and thus $\alpha=1$), implying the history state is a null state of $\Ht$. Observe the instance of Sparse Separable Hamiltonian above can be decided in NP, since it acts on $\log n$ qubits (and so the NP verifier can explicitly write out the matrix for the Hamiltonian). The corollary for \NEXP follows analogously by recalling $\NEXP=\MIP(\poly,\poly, \poly 2,1,1,2^{-r})$ for any desired polynomial $r$ (\Cref{thm:mipnexp}).
\begin{corollary}\label{cor:NEXP}
    Any instance of a \NEXP language is reducible under poly-time many-one reductions to an instance of Sparse Separable Hamiltonian on $O(\poly(n))$ qubits with completeness $1$ and soundness $1-1/\exp(n)$.
\end{corollary}
\noindent As was the case for NP above, here the instance of Sparse Separable Hamiltonian is decidable in NEXP, since it acts on $\poly(n)$ qubits.

\subsection{Containment in \texorpdfstring{$\QMA(2,p,q,r)$}{QMA(2,p,q,r)}} \label{sscn:containmentqmat}

Next, by combining \Cref{cor:QSPACE}, \Cref{cor:MIP}, \Cref{cor:NP} and \Cref{cor:NEXP} with the following lemma, we immediately obtain containment in $\QMA(2,p,q,r)$ for various appropriate $p,q,r$.

\begin{lemma}\label{l:QMAtcontain}
    Assume the notation of \Cref{l:embed}, and let $\Ht$ be the Sparse Separable Hamiltonian instance produced by the latter. Then, $\Ht$ can be decided in\footnote{Recall from \Cref{rem:1} that we omit big-Oh notation when listing class parameters, including for QMA.} $\QMA(2,q+\log m,q+\log m,r\log m)$, i.e. with proof and ancilla space scaling as $O(q+\log m)$, and promise gap as $O(1/(2^rm))$.
\end{lemma}
\noindent In words, the $\QMAt$ verifier preserves (up to linear overhead) both the number of qubits $\Ht$ acts on and its promise gap. The proof of \Cref{l:QMAtcontain} exploits the structure of the Hamiltonian produced by the Embedding Lemma, together with standard ideas. Curiously, at present we do \emph{not} know\footnote{Briefly, a natural approach is via phase estimation, as done in \cite{Chailloux2012} for $\QMAt$. However, the issue is that phase estimation requires exponential time in general to obtain exponential precision, which may be required in our setting since the weights $\Din,\Dprop,\Dsym$ scale as $\poly(m)$, which can be exponential in $n$.} how to show the analogue of \Cref{l:QMAtcontain} for \emph{arbitrary} sparse Hamiltonians (i.e. satisfying \Cref{def:sparse} and having worst-case exponential norm, but not promised to be of the form produced by \Cref{l:embed}).
\begin{proof}[Proof of \Cref{l:QMAtcontain}]
    We will need to explicitly reference the definitions below, reproduced for convenience:
        \begin{eqnarray}
    \Hint &=& (\hin)_L\otimes I_R + I_L\otimes (\hin)_R\label{eqn:100}\\
    \Hpropt &=& \sum_{t=1}^m \Htt,\quad \text{where $\Htt$ is defined as}\\
     \Htt&=&
    \begin{cases}
         (\HttI)_L\otimes (\HttiX)_R  +(\HttiX)_L\otimes (\HttI)_R & \text{ if }t\in P\label{eqn:101}\\
         (H_t)_L\otimes I_R + I_L\otimes (H_t)_R & \text{ if } t\not\in P
    \end{cases}\\
    \Houtt &=& (\hout)_L\otimes I_R + I_L\otimes (\hout)_R\label{eqn:102}\\
    \Hsym &=& I-\Psym_{LR}\quad\text{ for } \quad \Psym_{LR}=\frac{1}{2}\left(I_{LR}+\sum_{xy}\ketbra{xy}{yx}_{LR}\right),
    \end{eqnarray}
    The relevant facts regarding $\Ht = \Din\Hint+\Dprop\Hpropt+\Dsym\Hsym+\Houtt$ are:
    \begin{enumerate}
        \item $n$ is the input size to circuit $Q_n$, and all functions $m,p,q,r$ are parameterized in terms of $n$.
        \item $\Din,\Dprop,\Dsym$ are fixed polynomials in $m$ (the number of gates in the circuit $Q_n$, where recall $m\geq 2^p$ by assumption to allow enough time to process the full streamed proof),
        \item $\Ht$ acts on $O(q+\log m)$ qubits,
        \item the promise gap scales as $\abs{\alpha'-\beta'}\in \Omega((m2^r)^{-1})$ (recall $Q_n$ had promise gap $2^{-r}$), and
        \item in the YES case, there exists $\LR$ such that $\RL \Ht \LR \leq \alpha'$, and in the NO case, for all $\LR$, $\RL \Ht \LR \geq \beta'$.
    \end{enumerate}
    We construct a $\QMA(2,q+\log m,q+\log m,r\log m)$ verifier $V$ deciding whether $\Ht$ is a YES or NO instance.

    \paragraph{Constructing $V$.} We use Kitaev's original approach for placing the $k$-local Hamiltonian problem in QMA \cite[Proposition 14.2]{KSV02}: Pick a random ``term'' (defined shortly) of $\Ht$ and measure it against the claimed proof $\ket{\psi}=\LR$. The catch is that unlike in \cite{KSV02}, the ``terms'' of $\Ht$ are not $k$-local, so a slight bit more work is required to ensure $V$ can implement these measurements.

    To begin, define the ``terms'' of $\Ht$ as precisely the set of summands (with appropriate weights) on Equation (\ref{eqn:100}) (e.g. $\Din\Hin\otimes I$ is a term), Equation (\ref{eqn:101}) (e.g. for any $t\in P$, $\Dprop\HttI\otimes \HttiX$ and $\Dprop\HttiX\otimes \HttI$ are each terms), and Equation (\ref{eqn:102}) (e.g. $I\otimes \hout$), as well as $\Dsym\Hsym$. Then, there are precisely $K:=2m+5$ terms, on which we fix an arbitrary ordering. By construction, for all $i\in\set{1,\ldots, K}$, each term is a projector $\Pi_i$ up to scaling $w_i$ --- for example, since $\HttI$ and $\HttiX$ are projectors, so is $\Dprop\HttI\otimes \HttiX$ up to scaling $\Dprop$. (In this case, $\Pi_i = \HttI\otimes \HttiX$ and $w_i=\Dprop$.)

    We thus write $\Ht=  \sum_{i=1}^K w_i\Pi_i$ with $0\leq w_i\leq \poly(m)$, and define total weight $W := \sum_{i=1}^K w_i$. $V$ now acts as follows given proof $\ket{\psi}=\LR$:
    \begin{enumerate}
        \item Randomly select index $i\in[K]$ with probability $p_i = w_i/W$.
        \item Apply two-outcome projective measurement $M_0 := \Pi_i,\, M_1 := I - \Pi_i$ to $\ket\psi$.
        \item Accept on outcome $M_1$, reject on outcome $M_0$.
    \end{enumerate}
    The probability that $V$ accepts $\ket{\psi}$ is
    \begin{equation}
        \Pr[V\text{ accepts } \ket{\psi}] = \sum_{i=1}^K p_i \braketb{\psi}{(I-\Pi_i)} = 1-\frac{1}{W}\braketb{\psi}{\Ht}.
    \end{equation}
    Therefore, $V$ accepts with probability at least $1-\alpha'/W$ in the YES case and at most $1-\beta'/W$ in the NO case. Thus, $V$ has promise gap
    \begin{equation}
        \frac{\beta'-\alpha'}{W}\in \Omega\left(\frac{1}{m2^r}\cdot\frac{1}{\poly(m)}\right)\in \Omega\left(\frac{1}{2^r\poly(m)}\right),
    \end{equation}
    where we have used the fact that $W\in\poly (m)$ (since there are $2m+5$ terms, each with weight $w_i\in\poly(m)$).

    \paragraph{Efficiency of $V$.} It remains to argue that $V$ can be implemented efficiently, which in our setting means using $O(q+\log m)$ ancilla qubits and $\poly(n)$ gates. For Step 1 of $V$ (picking random $i\in[K]$), here is one approach to sample from distribution $\set{p_i}$ efficiently: Choose $j\in\set{1,\ldots, W}$ uniformly at random, where recall $W\in \poly(m)\in \exp(n)$ in the worst case. Then, compute the smallest $K'\in[K]$ such that $j\leq \sum_{i=1}^{K'}w_i$, and output $K'$. Both steps can clearly be done with $O(\log m)$ qubits, and $K'$ can be computed in time $\poly(n)$ since there are only a constant number of distinct weight values $w_i$ in our construction.

    As for Step 2 (projective measurements), the simplest measurement corresponds to term $\Hsym$, for which $M_1=\Psym$ (i.e. the projector onto the symmetric subspace) and $M_0=I-M_1$. This measurement is well-known to be efficiently implemented by the SWAP test~\cite{BCWW01} (\Cref{fig:SWAP}), which outputs $0$ with probability
    $\braketa{\bra{\psi_1}\bra{\psi_2}}{\Psym} = (1 + \abs{\braket{\psi_1}{\psi_2}}^2)/2$. The SWAP test clearly uses $O(q+\log m)$ qubits and is computable in time $\poly(n)$.
    \begin{figure}[t]
  \centering
  \begin{quantikz}[row sep={8mm,between origins}]
    \lstick{$\ket{0}$} & \gate{H} & \ctrl{1} & \gate{H} & \meter{} \\
    \lstick{$\ket{\psi_1}$} & \qw & \gate[wires=2]{\textup{SWAP}} & \qw & \qw \\
    \lstick{$\ket{\psi_2}$} & \qw & \qw  & \qw & \qw
  \end{quantikz}
\caption{The circuit for the SWAP test. The SWAP gate has action $\ket{x}\ket{y}\mapsto\ket{y}\ket{x}$ for any standard basis states $\ket{x},\ket{y}$. Note the inputs $\ket{\psi_1}$ and $\ket{\psi_2}$ are in tensor product. Measuring the first wire in the standard basis yields output $0$ with probability $(1 + \abs{\braket{\psi_1}{\psi_2}}^2)/2$, and postselecting on $0$ projects $\ket{\psi_1}\ket{\psi_2}$ onto the symmetric subspace.}
\label{fig:SWAP}
\end{figure}
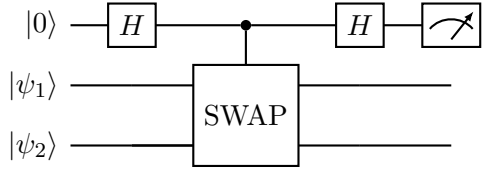

    As for the remaining terms of $\Ht$, we show how to efficiently implement for $t\not\in P$ the measurement corresponding to projector $H_t$ from Equation~(\ref{eqn:Ht}):
    \begin{equation}
            H_t:=-\frac{1}{2}V_t\otimes\ketbra{t}{ {t-1}}_{R_4} -\frac{1}{2}V_t^\dagger\otimes\ketbra{{t-1}}{t}_{R_4} +\frac{1}{2}I\otimes(\ketbra{ t}{t}+\ketbra{ {t-1}}{ {t-1}})_{R_4}
    \end{equation}
    The measurement of all remaining terms then follows similarly. Above, recall that $V_t$ is a $2$-qubit unitary, but the problem is the clock register $R_4$, which requires $O(\log m)$ qubits, which can scale as $\poly(n)$ in the worst case. However, this is easy to overcome --- informally, $V$ efficiently applies a change of basis to $R_4$ to map $\ket{t-1}$ and $\ket{t}$ to $\ket{0}$ and $\ket{1}$ (expressed in binary), respectively. Since the latter two differ only on their least significant bit, we have that under the change of basis, $H_t$ can be implemented by a $3$-local measurement (two qubits for $V_t$, one qubit for the clock).

    More formally, define permutation $U$ which swaps $\ket{t-1}_{R_4}$ with $\ket{0}_{R_4}$, swaps  $\ket{t}_{R_4}$ with $\ket{1}_{R_4}$, and otherwise acts invariantly on any $\ket{x}_{R_4}$ for $x\not\in\set{0,1,t-1,t}$. This permutation can clearly be implemented efficiently classically (and thus quantumly) with linear overhead space overhead, and in $\poly(n)$ time. Let $B$ denote the last qubit of $R_4$, and $A$ all other qubits of $R_4$. Then, expanding $R_4$ out in binary:
    \begin{eqnarray*}
        UH_tU^\dagger &=& \ketbra{0\cdots 0}{0\cdots 0}_A\otimes\left(-\frac{1}{2}V_t\otimes\ketbra{1}{ {0}}_{B} -\frac{1}{2}V_t^\dagger\otimes\ketbra{{0}}{1}_{B} +\frac{1}{2}I\otimes(\ketbra{ 1}{1}+\ketbra{ {0}}{ {0}})_{B}\right)\\
        &=:&\ketbra{0\cdots 0}{0\cdots 0}_A\otimes H_t'.
    \end{eqnarray*}
    A measurement corresponding to projector $UH_tU^\dagger$ can be efficiently implemented, since $H_t'$ is now a $3$-qubit operator. (For example, $V$ first measures $A$ in the standard basis, and conditioned on obtaining outcome $\ket{0\cdots 0}_A$, measures $H_t'$.) Thus, $V$ applies $UH_t U^\dagger$ to $U\ket{\psi}$ to complete the measurement. Again, each of these takes $O(q+\log m)$ space and $\poly(n)$ time, as required.
\end{proof}

With \Cref{l:QMAtcontain} in hand, the following corollary is immediate, and recovers the results of Blier and Tapp \cite{BT12} for NP and Pereszlényi for \NEXP \cite{Per12}. Below, recall $\PQMAlogt=\QMA(2,\log n,\log n, \log n)$, i.e. $\QMAt$ with log-size proof and ancilla and $1/\poly$ promise gap (technically, $\PQMAlogt$ also has perfect completeness by definition, which also matches the result we obtain below).
\begin{corollary}\label{cor:NEXP-QMA(2)}
    $\NP=\PQMAlogt$ (cf. \cite{BT12}) and $\NEXP=\PreciseQMAt$ (cf. \cite{Per12}).
\end{corollary}

\begin{proof}
    The containments $\PQMAlogt\subseteq\NP$ and $\PreciseQMAt\subseteq \NEXP$ are trivial. The containment $\NP\subseteq\PQMAlogt$ follows by mapping NP to a log-size SSH instance via \Cref{cor:NP}, followed by application of \Cref{l:QMAtcontain} to verify the SSH instance in $\QMA(2,\log n,\log n, \log n)=\PQMAlogt$. $\NEXP\subseteq\PreciseQMAt$ follows analogously by combining \Cref{cor:NEXP} with \Cref{l:QMAtcontain}.
\end{proof}

Via analogous arguments, we also obtain the following immediate corollaries.
\begin{corollary}\label{cor:QSPACEinQMAt}
    $\SQCMASPACE(p,q,r)\subseteq \QMA(2,q+\log p, q+\log p, p+r)$.
\end{corollary}
\noindent In words, $\SQCMASPACE$ with proof length $2^p$, $q$ ancilla qubits, and promise gap $1/2^r$ is contained in $\QMAt$ with $q+\log p$ proof and ancilla qubits, and promise gap $1/2^{p+r}$.

\begin{corollary}\label{cor:MIPinQMAt}
    It holds that
    \begin{equation}
        \MIP(t,u,v,p,r,c,s)\subseteq \QMA(2,u+v+\log(tr\log(pt)), u+v+\log(tr\log(pt)), tr\log(pt)+\log(c-s)).
    \end{equation}
\end{corollary}
\noindent In words, $\MIP$ with $t$ bits of communication per round, space $u$, $v$ random bits, $p$ provers, $r$ rounds, and completeness/soundness $c$ and $s$, respectively, is contained in $\QMAt$ with $u+v+\log(tr\log(pt))$ proof and ancilla qubits, and promise gap $2^{-tr\log(pt)+\log(c-s)}$. In more words, the amount of space is preserved, and the promise gap depends exponentially on the total amount of communication but only polynomially on the MIP promise gap.

\section*{Acknowledgements}
We thank Rolando Somma for pointing us to \cite{CBC21} and for interesting discussions, and Chinmay Nirke for feedback on this manuscript. SG acknowledges support from DFG grants 450041824 and 432788384.

\appendix
\section{\texorpdfstring{$\GSCONexp{}$ is \PSPACE-hard}{GSCON\_exp is PSPACE-hard}}\label{sec:gscon-exp-pspace}

During the proof of \Cref{thm:gscon-qcma}, \cite[Lemma 5.2]{Gharibian2018} shows the following result, which we restate here for completeness.
\begin{lemma}\label{lem:gscon-implicit}
    Let $A\in\hermp{\B^{\otimes n}}$ be a $k'$-local Hamiltonian.
    Consider the following promise problem $\Pi'$.
    \begin{enumerate}[align=left,labelwidth=\widthof{YES:},leftmargin=4em,labelsep*=1.1em]
        \item[YES:] There exists a sequence $(U_i)_{i=1}^{m'}$ of $l$-local unitaries such that $\braketb{\psi_A} A\le \alpha$ for $\ket{\psi_A} = U_{m'}\dotsm U_1\ket{0}^{\otimes n}$.
        \item[NO:] $\lmin(A)\ge\beta$.
    \end{enumerate}
    $\Pi'$ is polynomial-time reducible to $\GSCON$ with $m=2m'+2$, $\eta_1=\alpha$, $\eta_2 = \beta/(16m^2)$, $\eta_3=0$, $\eta_4=1/4$, $l=2$, $k=k'+2$, $\Delta = \eta_2-\eta_2$, if $\Delta>0$.
\end{lemma}
\begin{proof}
    The basic idea is to construct a Hamiltonian $H$ by adding three ``GO'' qubits to $A$, such that traversing the low energy space of $H$ forces one to simulate a protocol, which first prepares state $\ket{\psi_A}$ using local gates, then checks that $\ket{\psi_A}$ is indeed low energy, and finally uncomputes $\ket{\psi_A}$.

    Define $H\in\herm(\B^{\otimes(n+3)})$ acting on a \emph{Hamiltonian} register $h$ and \emph{GO} register $G$:
    \begin{equation*}
        H:=A_h\otimes P_G, \qquad P:=I-\ketbra{000}{000}-\ketbra{111}{111}
    \end{equation*}
    $H$ is $k$-local, as $P$ can be written $2$-locally~\cite{Gharibian2018}.
    The initial and final states are defined as $\ket\psi:=\ket{0}^{\otimes n} \ket{0}^{\otimes 3}$ and $\ket\phi := \ket{0}^{\otimes n}\ket{1}^{\otimes 3}$.
    $\Pi=(H,\aaa,\bbb,\ccc,\ddd,\Delta,l,m,\ket{\psi},\ket{\phi})$ is now a valid instance of \GSCON, and can be computed in polynomial time.

    \emph{Correctness:} Suppose $\Pi'$ is a YES instance, i.e. there exists a sequence $(U_i)_{i=1}^{m'}$ of $l$-local unitaries, such that $\braketb{\psi_A} A\le \alpha$ for $\ket{\psi_A} = U_{m'}\dotsm U_1\ket\psi$.
    We show that $\Pi$ is also a YES instance by constructing a sequence $(V_i)_{i=1}^{m}$ of $l$-local unitaries, such that $\ket\phi = V_m\dotsm V_1\ket\psi$ and $\braketb{\psi_i}H\le\aaa$ with $\ket{\psi_i} := V_i\dotsm V_1\ket\psi$ for all $i\in[m]$.
    $V_m\dotsm V_1$ implement the following steps:
    \begin{enumerate}
        \item \emph{Prepare $\ket{\psi_A}$:} Apply $(U_{m'}\dotsm U_1)_h$.
        \item \emph{Begin checking $\ket{\psi_A}$:} Apply $(X\otimes X\otimes I)_G$.
        \item \emph{Finish checking $\ket{\psi_A}$:} Apply $(I\otimes I\otimes X)_G$.
        \item \emph{Uncompute $\ket{\psi_A}$:} Apply  $(U_1^\dagger\dotsm U_{m'}^\dagger)_h$.
    \end{enumerate}
    This sequence has length $m=2m' + 2$ and maps $\ket\psi$ to $\ket\phi$ as desired.
    All intermediate states (besides the state after Step 2) $\ket{\psi_i}$ are in the nullspace of $H$, as $P$ maps their register $G$ to $0$.
    After Step 2, we have state $\ket{a_2} = \ket{\psi_A}_h\ket{110}_G$.
    By assumption, it holds that $\bra{a_2}H\ket{a_2}=\braketb{\psi_A}A\leq\alpha=\aaa$.

    \emph{Soundness:} Suppose $\Pi'$ is a NO instance.
    Let $S$ and $T$ be the image of projectors $I_h\otimes\kb{000}_G$ and $I_h\otimes\kb{111}_G$, respectively.
    $S,T$ are $2$-orthogonal, and $\ket\psi\in S,\,\ket\phi\in T$.
    Now fix any sequence $(V_i)_{i=1}^m$ of two-qubit unitaries.
    If $\norm{\ket{\psi_m} - \ket{\phi}}_2 \geq  1/4=\ddd$, $\Pi$ is already a NO instance for \GSCON.
    Otherwise, we can apply the Traversal Lemma (\Cref{lem:traversal-lemma}) with $\epsilon=1/4$ to conclude that there exists an $i\in[m]$, such that
    \[
    	\bra{\psi_i} P' \ket{\psi_i} \geq \left( \dfrac{1}{4m} \right)^2 = \dfrac{\eta_2}{\beta},
    \]
    where ${\ket{\psi_i}:=V_i\dotsm V_1\ket{\psi}}$ and $P' = I - \Pi_S - \Pi_T = I_h\otimes P_G$.
    Then, $\Pi$ is a NO instance because
    \[
        \bra{\psi_i} H \ket{\psi_i}
        =\bra{\psi_i} A \otimes P \ket{\psi_i}
        \geq \beta\bra{\psi_i} I_h \otimes P \ket{\psi_i}
        =\beta\bra{\psi_i} P' \ket{\psi_i}
        \geq \eta_2,
    \]
    where the first inequality follows since, by assumption, $\lmin(A)\ge \beta$.
\end{proof}
\noindent
To prove $\PSPACE$-hardness, we combine \Cref{lem:gscon-implicit} with two further insights.
Firstly, $2^{r(n)}$ unitary $2$-local gates are sufficient to construct any state $\ket\psi\in\B^n$ exactly (starting in $\ket{0^n}$) for some polynomial $r$ \cite{Nielsen}.
Therefore, the problem $\Pi$ defined in \Cref{lem:gscon-implicit} is equivalent to the problem whether $\lmin(H)\le \alpha$ or $\lmin(H)\ge\beta$ for $m'=2^{r(n)}$.

Secondly, $\QMA$ with an inverse exponential promise gap (i.e. $c-s = 2^{-\poly(n)}$), denoted $\PreciseQMA$, was shown by Fefferman and Lin to be $\PSPACE$-complete \cite{Fefferman2018}.
They also show that $\kLH$ with inverse exponential gap, denoted \emph{precise} $\kLH$, is $\PSPACE$-complete.
Their construction leads to the following lemma, which allows us to reduce $\PSPACE$ to a precise $\kLH$ instance with thresholds $\alpha$ and $\beta$, such that we can apply \Cref{lem:gscon-implicit} to solve it in $\GSCONexp{}$ with $m=2^{r(n)}$.

\begin{lemma}\label{lem:klh-exp}
    Any problem $\Pi$ in $\PSPACE{}$ is poly-time reducible to a $\kLH$ instance with $\beta
    /\alpha \ge 2^{p(n)}$ with $\alpha \ge 2^{-\poly(n)}$, where $p(n)$ is a freely chosen polynomial.
\end{lemma}
\begin{proof}
    This proof is based on \cite[Theorem 24]{Fefferman2018}.
    $\Pi$ can be reduced to $\PreciseQMA$ with completeness $c$ and soundness $s$, such that
    \[ 1-c = \epsilon, \qquad 1-s = -\epsilon + 2^{-g(n)}, \]
    for some polynomial $g(n)$ depending on $\Pi$ and any $\epsilon = 2^{-q(n)}$ for some polynomial $q(n)$ of our choice \cite{Fefferman2018}.

    The corresponding $\PreciseQMA$ verifier uses $T\le h(n,\log(1/\epsilon))$ unitaries, for some polynomial $h(x,y)$ \cite{Fefferman2018}.
    Hence, $\PreciseQMA$ with completeness $c$ and soundness $s$ can be reduced to a $3$-local Hamiltonian instance with thresholds
    \begin{align*}
        \alpha &= \frac{1-c}{T+1} = \frac{\epsilon}{T+1}, \qquad \beta = \frac{1-s}{T^3} = \frac{2^{-g(n)}-\epsilon}{T^3}.
    \end{align*}
    We can then choose a polynomial $q(n)\ge2g(n)$ such that
    \begin{align*}
        \frac\beta\alpha &= \frac{T+1}{T^3}\frac{2^{-g(n)}-\epsilon}{\epsilon}
        \ge T^{-2}\, \epsilon^{-1}\, 2^{-g(n)-1}
        = \frac{2^{q(n) - g(n)-1}}{h^2(n,q(n))}
        \ge 2^{p(n)}.
    \end{align*}
\end{proof}

\begin{theorem}\label{thm:PSPACEGSCON}
    \GSCONexp is \PSPACE{}-hard.
\end{theorem}
\begin{proof}
    Follows directly from \Cref{lem:gscon-implicit,lem:klh-exp} and the above discussion.
\end{proof}

\section{\texorpdfstring{\SSGSCONexp is \NEXP-complete}{SEPARABLE SPARSE GSCON\_exp is NEXP-complete}}\label{app:2}



Based on the fact that the separable sparse Hamiltonian problem is \NEXP-complete (\Cref{cor:NEXP}), we introduce a variant of \GSCON that is also \NEXP-complete.

\begin{definition}[Separable sparse ground state connectivity]
    \SSGSCON{} is defined as \GSCON, but the input $H\in\hermp{\B^{\otimes n}}$ is a sparse Hamiltonian (instead of a local one) with a bipartition $(L,R)$ of the qubits $H$ acts on, and every unitary $U_1,\dots,U_m$ acts either on $L$ or on $R$.
    \SSGSCONexp{} is defined analogously.
\end{definition}
\noindent
An analogue of \Cref{lem:gscon-implicit} for the separable sparse case also follows from the proof of \cite[Lemma 5.2]{Gharibian2018}.
\begin{lemma}\label{lem:gscon-implicit-sparse}
    Let $A\in\hermp{\B^{\otimes n}}$ be a sparse Hamiltonian with a bipartition of the qubits $A$ acts on into $(L,R)$.
    Consider the following promise problem $\Pi'$.
    \begin{enumerate}[align=left,labelwidth=\widthof{YES:},leftmargin=4em,labelsep*=1.1em]
        \item[YES:] There exists a sequence $(U_i)_{i=1}^{m'}$ of $l$-local unitaries acting either on $L$ or $R$ such that $\braketb{\psi_A} A\le \alpha$ for $\ket{\psi_A} = U_{m'}\dotsm U_1\ket{0}^{\otimes n}$.
        \item[NO:] For all unit vectors $\ket{\psi_1}\ket{\psi_2}$, $\braketa{\bra{\psi_1}_L\bra{\psi_2}_R}{\,A\,}\ge\beta$.
    \end{enumerate}
    $\Pi'$ is polynomial-time reducible to $\SSGSCON$ with $m=2m'+2$, $\eta_1=\alpha$, $\eta_2 = \beta/(16m^2)$, $\eta_3=0$, $\eta_4=1/4$, $l=2$, $\Delta = \eta_2-\eta_2$, if $\Delta>0$.
\end{lemma}
\begin{proof}
    The proof is almost the same as \Cref{lem:gscon-implicit}, so we just mention the differences.
    Here, $\Pi$ is a \SSGSCON instance.
    We choose the bipartition as $(L, R+G)$, i.e. the gates may act on $R$ and $G$ simultaneously.
    In the end, we need to show $\braketb{\psi_i}H\ge\eta_2$.
    For that, note that $\ket{\psi_i}$ is a product state of the form
    $\ket{\psi_i} = \ket{\gamma_1}_L\otimes\ket{\gamma_2}_{RG}$, where we can further decompose \[\ket{\gamma_2} = \sum_{x\in\bin^3}a_{x}\ket{\gamma_{2,x}}_R\ket{x}_G.\]
    Then for $X:=\bin^3\setminus\{000,111\}$,
    \[  \braketb{\psi_i}H = \sum_{x\in X}\abs{a_{x}}^2 \braketa{\bra{\gamma_1}\bra{\gamma_{2,x}}}A \ge \sum_{x\in X}\abs{a_{x}}^2\beta = \beta\bra{\psi_i} P' \ket{\psi_i} \geq \eta_2, \]
    where the first inequality holds by assumption of the NO case.
\end{proof}
\begin{theorem}\label{thm:SSGSCON-NEXP}
    \SSGSCONexp{} is \NEXP{}-complete.
\end{theorem}
\begin{proof}
    Containment of \SSGSCONexp in \NEXP is trivial. Hardness is analogous to \Cref{thm:PSPACEGSCON} using \Cref{lem:gscon-implicit-sparse} and \Cref{cor:NEXP}.
\end{proof}

\printbibliography

\end{document}